\documentclass[aoas]{imsart}

\usepackage{xr}
\makeatletter
\newcommand*{\addFileDependency}[1]{% argument=file name and extension
  \typeout{(#1)}% latexmk will find this if $recorder=0 (however, in that case, it will ignore #1 if it is a .aux or .pdf file etc and it exists! if it doesn't exist, it will appear in the list of dependents regardless)
  \@addtofilelist{#1}% if you want it to appear in \listfiles, not really necessary and latexmk doesn't use this
  \IfFileExists{#1}{}{\typeout{No file #1.}}% latexmk will find this message if #1 doesn't exist (yet)
}
\makeatother

\newcommand*{\myexternaldocument}[1]{%
    \externaldocument{#1}%
    \addFileDependency{#1.tex}%
    \addFileDependency{#1.aux}%
}
\myexternaldocument{sMRT-supplement}

\usepackage{hyperref}
\usepackage{amsmath,amssymb}
\usepackage{color}
\usepackage{bbm}
\usepackage{amsthm}
\usepackage{enumitem}
\usepackage{subcaption}
\usepackage{soul}
\usepackage{mathtools}% Loads amsmath

\usepackage{amsmath}
\usepackage{tikz}
\usepackage{dsfont}
\usetikzlibrary{matrix,arrows}
\usepackage{amsmath, amsfonts,amssymb,amsthm, txfonts, pxfonts,amscd}
\RequirePackage[OT1]{fontenc}
\usepackage{amsthm,amsmath}

\usepackage{natbib}
\usepackage{algorithm}
\usepackage{algpseudocode}
\usepackage{multirow}

\newcommand{\comments}{0}
\newenvironment{hproof}{%
  \renewcommand{\proofname}{Sketch Proof}\proof}{\endproof}

\def\pr{\text{pr}}
\def\given{\, | \,}
\def\Given{\, \big | \,}
\def\Var{\text{Var}}

\def\EE{\mathbb{E}}

\def\indep{\mathrel{\rlap{$\perp$}\kern1.6pt\mathord{\perp}}}
\def\Var{\text{Var}}

\def\EE{\mathbb{E}}

\def\f{\boldsymbol{f}}

\newcommand{\sam}[1]{\textcolor{blue}{[SAM:\ #1]}}
\newcommand{\walt}[1]{\textcolor{red}{[WD:\ #1]}}
\newcommand{\indicator}[1]{ \mathds{1}_{\{#1\}}}

\def\Var{\text{Var}}
\def\E{\mathbb{E}}

\def\p{{P}}
\def\Ep{\E_{\bf p}}
\def\Eeta{\E_{\eta}}
\def\f{\boldsymbol{f}}

\newtheorem{thm}{Theorem}[section]
\newtheorem{lemma}[thm]{Lemma}

\newtheorem{defn}[thm]{Definition}

\newtheorem{assumption}[thm]{Assumption}

\begin{document}
\begin{frontmatter}
%\title{Assessing treatment effects using Sample size calculations for
%stratified \\ micro-randomized trials in mobile health}
\title{The stratified micro-randomized trial design: 
  sample size considerations for testing nested causal effects of time-varying treatments}
%\runtitle{Sample size for stratified MRTs}
\runtitle{Stratified micro-randomized trial design}

\begin{aug}
\author[add1]{Walter Dempsey}\ead[label=e1]{wdem@umich.edu},
\author[add2]{Peng Liao}\ead[label=e1]{pengliao@umich.edu},
\author[add3]{Santosh Kumar}\ead[label=e1]{santosh.kumar@memphis.edu},
%\author[add3]{Mustafa al'Absi}\ead[label=e1]{malabsi@d.umn.edu},
and
\author[add1]{Susan A. Murphy}\ead[label=e1]{samurphy@umich.edu}
\affiliation[add1]{Harvard University}
\affiliation[add2]{University of Michigan}
\affiliation[add3]{University of Memphis}
%\affiliation[add3]{University of Minnesota}
\runauthor{W. Dempsey et al.}
\end{aug}

\begin{abstract}
Technological advancements in the field of mobile devices and wearable sensors have helped overcome obstacles in
the delivery of care, making it possible to deliver  behavioral treatments anytime and anywhere.
Increasingly the delivery of these treatments is triggered by predictions of risk or engagement
which may have been impacted by prior treatments.
Furthermore the treatments are often designed to have
an impact on individuals over a span of time during which
subsequent treatments may be provided.

Here we discuss our work on the design of a mobile health smoking
cessation  experimental study in which two challenges arose. First the
randomizations to treatment should occur at times of stress and second
the outcome of interest accrues over a period that may include
subsequent treatment.  To address these challenges we develop the ``stratified
micro-randomized trial,'' in which each individual is randomized among
treatments at times determined by predictions constructed from outcomes to prior treatment and with
randomization probabilities depending on these outcomes. We define
both conditional and marginal proximal treatment effects. Depending on
the scientific goal these effects may be defined over a period of time
during which subsequent treatments may be provided.  We develop a
primary analysis method and associated sample size formulae for
testing these effects.

\end{abstract}

\end{frontmatter}

% \walt{
%   Key reviewer question: Why didn't we want the more
%   complex simulation model initially?
%   Typically we have small prior datasets that
%   can be used to inform the sample size calculator.
%   The low-dimensional models are simple, fast, and
%   ``We will use the data in two ways.
%   First, we will use the data to inform the
%   sample size calculator.  Second, we will use
%   the data to suggest potential deviations
%   and assess robustness of the sample size
%   calculator to such deviations.
%   ''
%   The paper shows how to step through this
%   procedure in practice.
%   Suppose more general
% }

\section{Introduction}
\label{section:introduction}

%{\color{red}

The rise of wearable technologies has generated increased
scientific interest in the use and development of
mobile interventions.
Such mobile technology holds promise in providing
accessible support to individuals in need.
Mobile interventions to maintain adherence to HIV medication
and smoking cessation, for example, have shown sufficient effectiveness
to be recommended for inclusion in health services~\citep{Freeetal2013}.
Increasingly scientists aim to trigger delivery of treatments
based on predictions, such as of risk or engagement, which are
outcomes of prior treatments.
In these settings scientists are increasingly 
interested in assessing nested treatment effects.
For example, a scientist may want to understand if providing a
treatment  at  high risk time~\citep{Hovsepian:2015} is
effective.  Often times of high risk occur
infrequently. In these cases randomization to treatment might be triggered by a  risk prediction so as to avoid providing treatment at the wrong time and potentially providing too much treatment. Furthermore the scientist may want to detect these
treatment effects over the next hour during which subsequent
treatments may be delivered.

%Increasingly mobile health scientists develop ``risk'' or other types of predictions.
%The goal being to understand if providing a treatment when the risk
%prediction is high is effective~\citep{Hovsepian:2015}.
%Often periods of high risk occur infrequently. In these cases the
%randomization probabilities depend on the risk predictions so as to
%obtain sufficient treatment randomizations at occasions when the
%prediction of risk is high (and/or low) as well as to spread out these
%randomizations across the periods predicted  as high risk.

In this paper, we propose the \emph{stratified micro-randomized trial
design} because it is critical to stratify randomization 
to ensure sufficient occasions where the variable of
interest~(denoted~$X_t$), such as risk, 
takes a particular value~$x$ and treatment is
provided and sufficient occasions where~$X_t=x$ and treatment is not
provided. 
In these settings, the outcome of interest may require a period of time 
over which to develop; during this time period further treatment might
be provided.
To address this we provide a careful definition of the desired
treatment contrast and introduce the notion of a reference
distribution. 
We proceed by developing an appropriate test statistic for
the desired treatment contrast.
The associated sample size calculation is non-trivial due to 
unknown form of the non-centrality parameter.
Moreover, the distribution of~$X_t$ over time,~$t$, is unknown.
Therefore we develop an approach to formulating a simulation based
sample size calculator to accommodate the unknown longitudinal
distribution of~$X_t$.
The calculator requires the scientist to specify a generative model 
for the history~$H_t$ which achieves the specified 
alternative treatment effect.
However existing data sets that include the use of the required sensor
suites and thus can be used to guide the form of the generative model
are often small and do not include treatment.
To address this we provide a protocol for the use of such noisy, small 
datasets to inform the selection of the generative model,
leading to a data-driven, simulation-based sample size calculator.
We also illustrate how exploratory data analysis and over-fitting of
the same data can be used in constructing a feasible set of deviations
to which the sample size calculator should be robust.

% In these trials the randomization probabilities depend on a
% variable that may be impacted by prior treatment.  Additionally
% whether or not it is feasible or scientifically appropriate to deliver
% a treatment at any given occasion may depend on variables that may
% have been impacted by prior treatment.

This work is motivated by our participation in a mobile health smoking
cessation study, in which an average of 3 stress-reduction treatments
should be delivered per day, 1/2 at times the participant is
classified as stressed and 1/2 at times the participant is \emph{not}
classified as stressed. 
We use data from an observational, no treatment,
study of individuals~\citep{Sarker:book:2017, Saleheen:2015}
who are attempting to quit smoking
to construct the generative model underlying the simulation
based sample size calculator.
The data directly informs the generative model {under no treatment}.
We then build a generative model {under treatment}
by combining the generative model under no treatment
with the targeted alternative treatment effect.
We next over-fit the noisy, small data to suggest potential deviations
to which we assess robustness of the sample size calculator.

\subsection{Related work}

We build  upon prior work in experimental design and
on data analysis methods for time-varying causal effects.
We outline this related work below, highlighting key
differences to our current setting.
%We reference this work throughout the paper where
% appropriate.

\subsubsection{Micro-randomized trials}

Recently micro-randomized trial designs \citep{Liaoetal2015, Dempsey_Significance}
were developed for testing proximal and delayed effects
of treatment \citep{Klasjnaetal}.
%In these experimental designs,
While in these trials treatment is sequentially randomized per participant, this approach does not permit  the randomization probabilities to  depend
on features of the participant's observation history.
%outcomes which may be impacted by prior treatment (i.e.,
%the observation history).
%but may depend on time.
This restriction is quite problematic.  Indeed due to the rapid
increase in sensor technology and the ability of various machine
learning methods to provide real-time predictions, it is now feasible
for scientists to trigger treatments based on these predictions or
other features of the participant's observation history.
A critical question is whether triggering a treatment based on such features is effective.
Often these features may be impacted by prior treatment.
Furthermore the responses of greatest interest may be
 defined over a span of time
during which subsequent treatments may be delivered yet the approach
developed in \citep{Liaoetal2015} does not accomodate this.
%Finally, scientists want to allow delivery of  treatment to depend on variables that may be impacted by prior treatment.
We designed the stratified micro-randomized trial specifically for this more complex setting.
%We believe the form of this proximal response will be increasingly common in future studies.

%We believe the stratified
%micro-randomized design will become increasingly common as
%data collected from such a trial more readily addresses
%the scientific questions of interest within the mobile health community.

% Analytic formula were discovered due to simpler
% setting!

% We need to make it clear we discuss tradeoff more
% readily.

\subsubsection{N-of-1 trials}

 At first glance, the  micro-randomized trial design appears
similar to the N-of-1 trial design frequently used in the behavioral sciences.
However the estimand is quite different.
We will, as is typical in statistical causal inference,
consider average causal effects, possibly conditional on covariates.
In the behavioral field N-of-1 trials are used most often to
ascertain individual level causal effects~\citep{McDonaldetal}.
%\sam{Mcdonald et al., 2017}.
A variety of nuanced assumptions about individual behavior
using behavioral science theory is brought to bear as
scientists attempt to triangulate on individual level effects;
see the section on \lq\lq Measuring behavior over time\rq\rq\
in \cite{McDonaldetal} for a discussion.
In the clinical field, N-of-1 trials were developed for
settings in which scientists wish to compare the
effect of one treatment versus another (treatment A versus
treatment B) on an outcome but it is very expensive to
recruit many participants.
%The desired comparison does not generally have a dynamic component, nor should there by carryover effects.
In both settings a common  assumption underlying the analysis of N-of-1 trials
is that there are no carry-over effects.
Additionally one often assumes that the treatment effect is constant over time.
%Usually in these trials, each participant is subject to periods of treatment  interspersed with periods of no treatment.
%For example during periods when a participant is on treatment one might expect the
%response to be generally higher than the  response during the time periods in which the
%participant is off treatment.
An excellent overview of N-of-1 designs and their use
for evaluating technology based interventions
is~\cite{Dallery2013}. See \cite{Kravitz2014}
%\sam{Kravitz 2014}
for a review of this design in pharmacotherapy trials.

\subsection{Outline}

This paper is organized as follows.
In section~\ref{section:strat-mrt} we discuss
the stratified micro-randomized trial and describe in
greater detail the motivating smoking cessation study.
In section~\ref{section:cond_effects} we define two
types of  treatment effects:
a conditional treatment effect, conditional on a stratification
variable, and a treatment effect that is marginal
over the stratification variable.
Section~\ref{section:cond_test_statistic} provides
primary analysis methods and associated
theory for the proposed trial design.
We then provide a simulation-based method for determining the sample size
for a stratified micro-randomized trial in
section~\ref{section:sample_size}.  This simulation-based sample size calculator requires a generative model for the trial data.
%These entail a number of computational issues which are also discussed.
We develop a generative model for the smoking cessation example in section~\ref{section:smoking_example} and develop the simulation based sample size calculator for this example.  In this example the development of the generative model begins with the development of model under no treatment.  This latter model is constructed using summary 
statistics on
	data collected in an observational, no treatment,
	smoking cessation study of
	% $61$ 
	cigarette smokers
	\citep{Saleheen:2015}.
	Section~\ref{subsubsection:calculator_smoking} describes the dataset and how it is used to inform the generative model.
We also conduct a variety of robustness checks and subsequently revise the generative model.  Here too, the observational, no treatment, smoking cessation study is used to indicate where robustness is required. 
Section~\ref{section:alternatives} provides a discussion.

\section{Stratified Micro-Randomized Trial}
\label{section:strat-mrt}

\subsection{Motivating example -- Smoking cessation study}
\label{subsection:motex}
Here we provide a simplified description of the smoking cessation study
which we are involved in through the
Mobile Data to Knowledge Center (https://md2k.org/).
This is a 10 day mobile health intervention study focused on developing a mobile health intervention
aimed at aiding individuals who are attempting to quit smoking.
Participants wear both an AutoSense chest band~\citep{Autosense:Ertin:2011} as well as bands on each wrist for 10 hours per day.
Sensors in the chestband and wristband measure various physiological responses
and body movements to robustly assess {physiological stress}.
In particular a pattern-mining algorithm uses the sensor data
 to construct a binary time-varying stress classification
(see Section~\ref{section:smoking_example} and \cite{Sarker:2016} for  further details) at each minute of sensor
wearing throughout the entire day.

Each participant's smartphone contains a number of ``mindfullness apps'' that
can be accessed 24/7 to engage in guided stress-reduction exercises.
In this study the treatment is a  smartphone notification to remind
the participant to access the app and  practice  the stress-reduction exercises.
Theoretically, a treatment can be delivered at any minute during the 10 hour day.
However in practice, treatment will only be delivered when the participant is available.
That is, at some time points it is inappropriate for scientific, ethical or burden
reasons to provide treatment.
In this example, one of the reasons why a participant would not be available at decision time~$t$ is if the participant received a treatment in the past hour (see Section~\ref{section:smoking_example} for further details on availability specific to this trial).

At each minute availability is ascertained and if the participant is available, then
the participant is randomized to receive or not receive a treatment.
In this study the repeated randomizations are stratified to ensure that  each participant should receive
an average of 1.5 treatments per day while classified as stressed and an average of 1.5
treatments per day
% at minutes for which the participant is not  classified as
while not classified as stressed.

We  consider primary analyses and sample size formula when the
primary aim of this type of study is to address  scientific questions such as:
\begin{center} \emph{
Is there an effect of the treatment on the proximal response? And is there an effect of the treatment if  the individual is currently experiencing stress?
}
\end{center}
The stratified micro-randomized trial is an experimental design
intended to provide data to address such questions.

%\sam{here state that to determine the sample size we will use  mustafa's data and describe this data}

\subsection{A Stratified Micro-Randomized Trial}
\label{subsection:notation}

A \emph{micro-randomized trial}~\citep{Liaoetal2015,
  Dempsey_Significance}  consists of a sequence of
within-person decision times~$t=1,\ldots, T$, e.g. occasions,
at which treatment may be randomized.  For example, in the smoking cessation study the
decision times are at minute intervals during a 10 hour day over a period of 10 days
(i.e., $T=600*10$ decision times) for each participant.
As discussed in the introduction
we are  interested in
treatment effects at particular values of a variable $X_t$ that
are likely impacted by prior treatment
(in the smoking cessation study, $X_t$ is an indicator of stress and
treatment is intended to impact the occurrence of stress); often in
these settings  some values of $X_t$ occur more rarely
(e.g., participants experience many fewer
minutes of stress than non-stress minutes in a day)
and thus to ensure sufficient treatment exposure at these
values we stratify  the randomization.  We call such trials \emph{stratified micro-randomized trials}.
We assume the sample space
for the covariate~$\mathcal{X}$ is finite and small.
That is, $X_{t}$ is a time-varying
categorical (or ordinal)  variable with
support~$\mathcal{X} = \{ 0, 1, \ldots, k \}$ % =: [k]$
where $k$ is small.   In the case of the smoking cessation example,
$X_t=1$ if the participant is classified as
stressed at decision time $t$ and $X_t=0$, otherwise, thus $k=1$.

$O_t$ ($t \geq 1)$ denotes  observations collected
after time~$t-1$ and up to and including time $t$ (including the time varying stratification variable, $X_t$);  $O_{1}$ contains baseline covariates.
$O_t$ also contains the availability indicator: $I_t=1$ if available for treatment and $I_t=0$ otherwise.
Availability at time $t$ is determined before treatment randomization.
In this paper, we consider  binary treatment (e.g., on or off);
$A_t$ denotes the indicator for the randomized treatment at time~$t$.
A randomization only occurs if $I_t=1$.
In the smoking cessation example $A_t=1$ if at minute $t$,
the participant is notified to practice stress-reduction
exercises and $A_t=0$ otherwise.
In particular if the participant is unavailable~(i.e., $I_t = 0$)
there can be no notification to practice stress-reduction
exercises (i.e.,~$A_t = 0$).
The ordering of the data at a decision time $t$ is $O_t, A_t$.
Let~$H_{t} = ( \{  O_s, A_s \}_{s=1}^{t-1} , O_{t} ) $ denote
the observation history up to and including time~$t$,
as well as the treatment history at all
decision times up to, but not including, time~$t$.

In general the randomization probability for $A_t$ will depend
not only on the stratification variable, $X_t$ but also  other variables  in $H_t$.  The
$\pr  ( A_{t} = 1 \given H_{t})$ is a known function of $H_t$,
denoted by  $p_{t} (1 \given H_{t} )$.
We define~$p_{t} (1 \given H_{t}) = 0$ when
the participant is currently unavailable
(i.e., $I_t = 0 \subset H_t$).
Appendix~\ref{app:randprobs} provides an example,
suitable in the smoking cessation example, of a formula
for $p_{t} (a \given h_{t} ),  t=1,\dots, T$
for any possible value of history given by $h_t$.
From here on, we assume the investigator has access
to a formula for these randomization probabilities.
Let $\p_{\bf p}$  denote the distribution of the data if collected using
randomization probabilities determined by this formula.

\sloppy
The proximal response, denoted by~$Y_{t,\Delta}$, is a known
function of the participant's data within a subsequent window
of length~$\Delta$ (i.e.,  $\{ O_{t+1}, A_{t+1}, \ldots,  O_{t+\Delta-1}, A_{t+\Delta-1}, O_{t+\Delta}\}$).
In the smoking cessation study, for example, the length of window might be $\Delta=60$ minutes
with proximal response
\[
Y_{t,\Delta} = \Delta^{-1} \sum_{s=1}^{\Delta} {\bf 1}_{ X_{t+s} = 1}.
\]
In this smoking cessation example, the response is a deterministic
function of \emph{only} the stratification
covariate, $X_t$; this need not be the case.  For example in a physical activity study in which the treatments are activity messages $X_t$ may be a binary variable indicating currently sedentary or not yet the response might be the number of steps over subsequent $x$ minutes.

\fussy
\section{Proximal effect of treatment}
\label{section:cond_effects}

The primary question of interest is whether the treatment has a proximal effect;
that is, whether there is an effect of treatment at decision time $t$ on the proximal
response $Y_{t,\Delta}$. In particular we aim to test if the proximal effect is zero.
Note we are only interested in treatment effects conditional on availability ($I_t=1$).
We consider two types of proximal effects:
an effect that is defined conditionally on the value of the stratification variable,
$X_t$ and $I_t=1$ or an effect that is conditional only on $I_t=1$, so marginal with
respect to the distribution of $X_t$.

\subsection{Proximal effect of treatment, Potential outcomes \& Reference distribution}
\label{section:prox_effects_pot_outcome}
We use potential outcomes~\citep{Robins,Rubin} to define  both the
conditional and marginal proximal effect.
At time 2, the potential observations are $\{O_2(a_1)\}_{a_1 \in\{0,1\}}$.
The potential observations and availability
at decision time~$t$
are $\{ O_t (\bar{a}_{t-1}) \}_{\bar{a}_{t-1} \in\{0,1\}^{t-1}}$.
Recall that the proximal response is a known function of the
participant's data within a subsequent window of length~$\Delta$.
Thus the potential outcomes for the response at time~$t$ are
$\{Y_{t,\Delta}(\bar{a}_{t+\Delta-1})\}_{\bar{a}_{t+\Delta-1} \in\{0,1\}^{t+\Delta-1}}$;
each individual has $2^{t+\Delta-1}$ potential responses at time~$t$.

\begin{defn}[Proximal treatment effects] \normalfont
At the individual level, the effect of providing treatment versus not
providing treatment at time $t$ is a difference in potential outcomes
for the proximal response and is given by
\begin{eqnarray}
\label{ind_txt}
Y_{t,\Delta} (\bar{a}_{t-1},1, a_{t+1},\dots,a_{t+\Delta-1} ) -Y_{t,\Delta} (\bar{a}_{t-1},0, a_{t+1},\dots,a_{t+\Delta-1} ).
\end{eqnarray}
There are $2^{t+\Delta-2}$ of these treatment differences for each individual,
each corresponding to a value for $(\bar{a}_{t-1}, a_{t+1},\dots,a_{t+\Delta-1})$.
The  \lq\lq fundamental problem of causal inference\rq\rq~\citep{Imbens2015,Pearl2009}
is that we can not observe any one of  these individual differences.
Thus we provide a definition of the treatment effect that is an average across individuals.
Furthermore to define the effect of treatment we must specify a reference distribution,
that is the distribution of the treatments prior to time $t$, $\bar{a}_{t-1}$
and if $\Delta>1$ then we must also define the distribution of the treatments
after time $t$, $(a_{t+1},\dots,a_{t+\Delta-1})$.  If the reference distribution is
not a point mass then, in the definition of the treatment effect, here too, the
treatment effect will be an average; the average is over the above differences
(\ref{ind_txt}) with respect to the reference distribution.  So in summary the
treatment effect at time $t$ will be an average of the differences in (\ref{ind_txt})
both over the distribution across individuals in potential outcomes as well as over
the reference distribution for the treatments.

The question is, \lq\lq Which reference  distribution should be used for the treatments?\rq\rq\
The choice of which distribution to use for $(a_{t+1},\dots,a_{t+\Delta-1})$ might differ
by the type of inference desired.   For example in the smoking cessation study,
it makes sense to consider setting the treatments~$a_{t+1},\dots,a_{t+\Delta-1}$
to $0$.
In this case we can interpret the treatment effect as the effect of providing
a notification at time $t$ to practice stress-reduction exercises and no more notifications
within the next hour versus no notification at time $t$ nor over
the next hour on the fraction of time stressed in the next hour
(i.e., the proximal response).

In  this paper, we set treatment at the subsequent $\Delta-1$ times equal to $0$ as described above. In order to select the reference distribution for $\bar{a}_{t-1}$  we follow
common practice in observational mobile health studies; here longitudinal methods
such as GEEs and random effects models~\citep{Liang1986} might be used to model how
a time-varying variable, such as physical activity, varies with  current mood. In this case the mean model
in these analyses is marginal over the past distribution of  mood.
A similar strategy in the randomized setting  is to use the past treatment randomization probabilities
as the reference distribution.

With the  reference distribution set to  the randomization probabilities for past treatment and  set to no treatment for the subsequent $\Delta-1$ times, the average causal effect at time $t$ can be viewed as an \lq\lq excursion.\rq\rq\ That is, participants get to time $t$ under treatment according to  the randomization probabilities, then at time $t$ (if available) the effect is the contrast between two opposing excursions into the future.  In one excursion, we treat at time $t$ and then do not treat for $\Delta-1$ further times; in the  opposing excursion, we do not treat at time $t$ nor do we treat for $\Delta-1$ subsequent times.

Using the above reference distribution, the marginal, proximal treatment effect at time $t$, $\beta(t)$, is:
\begin{eqnarray*}
  \frac{ \E\left[\sum_{\bar a_{t-1}}\left(\prod_{j=1}^{t-1} p_j(a_j|H_j(\bar a_{j-1})) \right)
  (  Y_{t,\Delta} (\bar{a}_{t-1}, 1, \bar 0)  -
  Y_{t,\Delta} (\bar{a}_{t-1}, 0, \bar 0 )) I_t (\bar{a}_{t-1})\right]}
  {\E\left[\sum_{\bar a_{t-1}} \left(\prod_{j=1}^{t-1} p_j(a_j|H_j(\bar a_{j-1})) \right)
  I_t (\bar{a}_{t-1}) \right]}
\end{eqnarray*}
where the expectation, $\E$ is over the distribution of the potential outcomes
and $\bar 0$ is a row vector of length $\Delta-1$.
Define the conditional, proximal effect, $\beta(t; x)$, as follows:
 \begin{eqnarray*}
   \frac{ \E\left[\sum_{\bar a_{t-1}}\left(\prod_{j=1}^{t-1} p_j(a_j|H_j(\bar a_{j-1})) \right)
   (  Y_{t,\Delta} (\bar{a}_{t-1}, 1, \bar 0)  -
   Y_{t,\Delta} (\bar{a}_{t-1}, 0,\bar 0 )) I_t (\bar{a}_{t-1})1_{X_t(\bar a_{t-1})=x}\right]}
   {\E\left[\sum_{\bar a_{t-1}} \left(\prod_{j=1}^{t-1} p_j(a_j|H_j(\bar a_{j-1})) \right)
   I_t (\bar{a}_{t-1}) 1_{X_t(\bar a_{t-1})=x}\right]}.
\end{eqnarray*}
\end{defn}

The proximal effects can be defined for other reference distributions over
$(\bar{a}_{t-1}, a_{t+1},\dots,a_{t+\Delta-1})$.   Careful consideration is required
in selecting the reference distribution. For example, a natural alternative to
setting the treatments $a_{t+1},\dots,a_{t+\Delta-1}$ to $0$ in the above definition
would be to use a definition which averaged over the randomization distribution, $\p_{\bf p}$.
Consider the smoking cessation example.  Here if at time $t$ treatment is delivered
then according to the randomization protocol the participant cannot be provided further
treatment in the subsequent hour.   On the other hand, if treatment is not provided at
time $t$ then  the participant may be provided treatment in the subsequent hour.
Thus defining the proximal treatment effect with respect
to the randomization distribution~$\p_{\bf p}$ means that the treatment contrast is
between providing treatment at time $t$ versus the combination of  delaying treatment to
later time points in the next hour or not providing treatment in the next hour.

A further consideration in selecting a reference distribution is that if the reference distribution is far from the randomization
distribution then treatment
effects may be very difficult to estimate.  That is, the sample size
necessary to achieve the requisite power to detect treatment effects will
be practically infeasible (i.e, astronomical). Consider again the smoking cessation study example.
Using data from other studies on smokers who are trying to quit we know that there are
only a few times per day at which the smoker is classified as
stressed.
In the subset of the observational, no treatment, study
used to inform our generative models, the mean (standard deviation) of the
number of episodes classified as stressed per day per person
was $2.8$ ($3.2$). The mean (standard deviation) of the
number of episodes \emph{not} classified as stressed per
day per person was $17.1$ ($12.6$).
%the m\sam{here is a chance to bring in mustafa's data}
These statistics support the conclusion that most of the day
the smoker is not stressed.  Recall the randomization distribution
must satisfy the restriction that on average 1.5 treatments are
provided while a smoker is classified as stressed and
on average 1.5 treatments are provided while a smoker is classified
as non-stressed.
This is over a 10 hour day. This means that at any given minute,
the participant is likely classified as not stressed and
the probability of treatment at this minute is very low.
As a result the product of randomization
probabilities~$\prod_{j=t+1}^{t+\Delta-1}p_{j}(0|H_{j})$
is close to $1$ and thus close to a reference distribution
that provides no treatment at times $t+1,\ldots, t+\Delta-1$.
This means that there will be much data from the study
that is consistent with the reference distribution.
If, however the randomization probabilities had to satisfy a
restriction specifying a much larger number of treatments,
then there would be very little data consistent with the
reference distribution.

For the reminder of this paper, the proximal effects are
defined using the randomization distribution for past
treatments ($\bar{a}_{t-1}$) and  $(a_{t+1},\dots,a_{t+\Delta-1})$
are set to 0 (no treatment).

\subsection{Proximal effect of treatment \& Observable Data}
\label{section:prox_effects_data}
To express the causal treatment effects, $\beta(t)$ and $\beta(t;x)$ in terms
of the observable data,\\ e.g. $\{  O_{1}, A_{1},
\ldots, O_{t}, A_{t},
\ldots, O_{T}, A_{T}, \ldots, A_{T+\Delta-1}, O_{T+\Delta}
\}$,  we use the following three assumptions.
\begin{assumption} \normalfont
  \label{consistency}
  We assume consistency, positivity, and sequential ignorability~\citep{Robins}:
  \begin{itemize}
  \item Consistency: For each~$t \leq T+\Delta$,
    $O_t (\bar{A}_{t-1} ) = O_t$.
    That is, the observed values are equal the corresponding potential outcomes.
  \item Positivity: if the joint density~$\{ H_t = h, A_t = a\}$ is greater
    than zero, then~$\pr (A_t = a_t \given H_t = h_t ) > 0$.
  \item Sequential ignorability: for each~$t \leq T$, the
    potential outcomes,\\ $\{O_2(a_1), \ldots,
    O_{T+\Delta}(\bar a_{T+\Delta-1}) \}_{\bar a_{T+\Delta-1}\in \{0,1\}^{T+\Delta-1}}$,
    are independent of~$A_t$ conditional on the history~$H_t$.
  \end{itemize}
\end{assumption}

Sequential ignorability and, assuming all of the randomization probabilities
are bounded away from $0$ and $1$, positivity, are guaranteed for a stratified micro-randomized trial
by design. Consistency is a necessary assumption for linking the potential
outcomes as defined here to the data. When an individual's outcomes may be influenced
by the treatments provided to other individuals, consistency may not hold.
In such instances, a group-based conceptualization of potential outcomes is used~\citep{Hong2006,Vanderweele2013}.
In particular if the mobile intervention includes treatments that aim to
produce social ties between participants, then consistency as stated above
will not hold. For simplicity we do not consider such  mobile interventions here.

\begin{lemma}
  \label{lemma:cond_effect}
  Under assumption~\ref{consistency}, the marginal treatment effect satisfies
  \begin{eqnarray}
    \label{marg_eff}
    \beta(t) =  \E \left[  \E\left[ \prod_{j=t+1}^{t+\Delta-1}
    \frac{1_{A_j=0}}{p_j(A_j|H_j)}Y_{t,\Delta} \bigg| A_t = 1 , H_t
    \right] \bigg| I_t = 1 \right] -\phantom{bbbbbbbbbbb}
    \cr
    \phantom{bbbbbbbbbbb}\E \left[ \E \left[ \prod_{j=t+1}^{t+\Delta-1}
    \frac{1_{A_j=0}}{p_j(A_j|H_j)}Y_{t,\Delta} \bigg| A_t = 0 , H_t
    \right] \bigg|  I_t = 1 \right]
  \end{eqnarray}
  and the conditional treatment effect satisfies
  \begin{eqnarray}
    \label{cond_eff}
    \beta(t; x) = \E \left[  \E\left[ \prod_{j=t+1}^{t+\Delta-1}
    \frac{1_{A_j=0}}{p_j(A_j|H_j)}Y_{t,\Delta} \bigg| A_t = 1 , H_t
    \right] \bigg| X_t = x, I_t = 1 \right] -\phantom{bbbbb}
    \cr
    \phantom{bbbbb}\E \left[ \E \left[ \prod_{j=t+1}^{t+\Delta-1}
    \frac{1_{A_j=0}}{p_j(A_j|H_j)}Y_{t,\Delta} \bigg| A_t = 0 , H_t
    \right] \bigg| X_t = x, I_t = 1 \right]
  \end{eqnarray}
  for all $x \in \{0,\ldots,k\}$
  where $\E$ denotes the expectation with respect to distribution of the data
  generated via a stratified micro-randomized trial with randomization distribution, $\p_{\bf p}$.
\end{lemma}
Note that the above products, e.g. $\prod_{j=t+1}^{t+\Delta-1}\frac{1_{A_j=0}}{p_j(A_j|H_j)}$,
are set to $1$ if $\Delta=1$.  Proof of Lemma~\ref{lemma:cond_effect} can be
found in Appendix~\ref{app:technical}.
In the following we focus on designing a stratified micro-randomized trial
for the primary purpose of testing whether the treatment effect at
any time point differs from 0.

\section{Test statistic}
\label{section:cond_test_statistic}

Our main objective is the development of a sample size formula that
will ensure sufficient power to detect alternatives to the null
hypothesis of no proximal treatment effect.  For the conditional
proximal effect the null hypothesis is~$H_0 : \beta(t; x) = 0, t = 1\ldots, T$ and $x \in \{0,\dots, k\}$.
For the marginal proximal effect the null hypothesis is~$H_0 : \beta (t) = 0, t = 1\ldots, T$.
The proposed sample size formulas are simulation based and will
follow from consideration of the distribution of test statistics under
alternatives to the above null hypotheses.  The sample size will be denoted by $N$. Our test statistic will be based on a  generalization of the test statistics developed by~\cite{Boruvkaetal} to accommodate the fact that the response $Y_{t,\Delta}$ covers a time interval during which subsequent treatment may be delivered  (in~\cite{Boruvkaetal}, $\Delta=1$ throughout) and the conceptual insight that  these estimators   can be interpreted as $L_2$ projections.
% In particular  our generalization allows the response to be   defined over a span of time during which subsequent treatments might be delivered.
 These test statistics are quadratic forms based on estimators of the
coefficients involved in  $L_2$ projections.

In the following we describe $L_2$ projections, and provide the test statistics.
First in the conditional setting the test statistic is based on an empirical projection
 of $\{\beta(t; x) \}_{ t = 1\ldots, T;x \in \{0,\dots, k\}}$ on the  space spanned
by a $q_c$ by $1$ vector of features involving $t$ and $x$, denoted by $f_t (x)$.
We denote the projection by $f_t (x)^\prime \beta_c$. The $\beta_c$ weights in
this projection are given by
\[
  \beta_c^\star=\arg\min_{\beta_c}  \E\left[\sum_{t=1}^T I_t
    \tilde{p}_t(1|X_t)(1-\tilde{p}_t(1|X_t))  \left( \beta(t; X_t)
      -  f_t (X_t)^\prime \beta_c \right)^2 \right]
\]
where $ \{ \tilde{p}_t(1|x)\}_{t=1,\ldots, T; x\in \{0,\ldots,k\}}$ are pre-specified
probabilities used to define the weighting across time and stratification
distribution in the projection.  Note that if desired, one can set
$\tilde{p}_t(1|x)=1/k$ for all $t,x$. See Section~\ref{subsection:computations}
for further comments on the choice of the pre-specified probabilities
and on the choice of $f_t (x)$.

Second, in the marginal setting, the test statistic is based on estimators of
the coefficients involved in an $L_2$ projection of $\{\beta(t) \}_{ t = 1\ldots, T}$
on the space spanned by a $q_m$ by $1$ vector of features involving $t$,
denoted by $f_t$. We denote the projection by $f_t^\prime \beta_m$.
The $\beta_m$ weights in this projection are given by
\[
  \beta_m^\star=\arg\min_{\beta_m}  \E\left[\sum_{t=1}^T
    I_t \tilde{p}_t(1)(1-\tilde{p}_t(1))  \left( \beta(t)
      -  f_t^\prime \beta_m \right)^2 \right]
\]
for pre-specified probabilities, $ \{ \tilde{p}_t(1)\}_{t=1,\ldots, T}$.
Again these probabilities are used to specify the weighting across time
and stratification distribution in the projection.

Here we discuss the estimators of the coefficients in the $L_2$ projections.
The estimators will form the basis for the test statistics.
Note that neither treatment effect, $\beta(t;x)$ in (\ref{cond_eff})
nor $\beta(t)$ in (\ref{marg_eff}), are conditional expectations of an
observable variable (rather the effects are defined by differences in repeated
conditional expectations). Thus instead of minimizing a standard least squares
criterion, we minimize a  generalization of the criterion in \cite{Boruvkaetal} (see (\ref{eq:conditional_ls}), (\ref{eq:marginal_ls}) below).

In some settings there will be sufficient a priori information (e.g. using data
on individuals from a similar population) that will permit the simulation based
sample size formula to depend on ``control variables.'' These variables are used
to help reduce the variance of the estimators with the goal that the resulting
test statistic is more powerful in detecting particular alternatives to the null hypothesis.
See Section~\ref{subsection:computations} for further discussion concerning the choice
of the control variables. For example in the smoking cessation study a natural
control variable would be the fraction of time stressed in the hour prior to
time $t$ as this pre-time $t$ variable may be expected to be highly correlated with
the fraction of time stressed in the hour subsequent to time $t$, $Y_{t,60}$.

Given a $q^\prime$ by $1$ vector of ``control variables'' $g_t(H_t)$, define
$g_t(H_t)^\prime \alpha_c^\star$ as an $L_2$ projection; in particular
\[\alpha^\star_c=\arg\min_\alpha \E \left[ \sum_{t=1}^T \ I_t w_{ct} (H_{t+\Delta-1})
  \Big(  Y_{t,\Delta} -  g_t(H_t)^\prime\alpha_{c}  \Big)^2\right]\]
where $w_{ct} ( H_{t+\Delta-1} ) = \frac{\tilde{p}_t (A_t|X_{t}) {\prod_{s=1}^{\Delta-1} }
  {\bf 1} [A_{t+s} = 0] }{\prod_{s=0}^{\Delta-1} p_{t+s} ( A_{t+s } \given H_{t+s}) }$.
Also define $g_t(H_t)^\prime \alpha_m^\star$ as an $L_2$ projection; in particular
\[\alpha^\star_m=\arg\min_\alpha \E \left[ \sum_{t=1}^T \ I_t w_{mt} (H_{t+\Delta-1})
  \Big(  Y_{t,\Delta} -  g_t(H_t)^\prime\alpha_{m}  \Big)^2\right]\]
where $w_{mt} ( H_{t+\Delta-1} ) = \frac{\tilde{p}_t (A_t) {\prod_{s=1}^{\Delta-1} }
  {\bf 1} [A_{t+s} = 0] }{\prod_{s=0}^{\Delta-1} p_{t+s} ( A_{t+s } \given H_{t+s}) }$.
Note one can choose $g_t(H_t)$ equal to the scalar, $1$.
Again see Section~\ref{subsection:computations} for further discussion.
See appendix~\ref{app:tradeoff} for a discussion
of the trade-off between the approximation error of the~$L_2$
projection of~$\{ \E[ w_{ct} (H_{t+\Delta-1} ) Y_{t,\Delta} \given H_t, I_t = 1 ] \}_{t=1,\ldots,T}$
onto the control variables~$g_t (H_t) \alpha^\star_c$,
sample size~$N$, and statistical power~$1-\beta_0$.

Recall the proposed test statistic is based on an estimator of $\beta_c^\star$ or $\beta_m^\star$.
Here we consider an estimator of $\beta_c^\star$ which is the minimizer of the following weighted, centered
least-squares criteria, minimized over $(\alpha_c, \beta_c)$:
\begin{equation}
  \label{eq:conditional_ls}
  \mathbb{P}_n \left[ \sum_{t=1}^T I_t \, w_{ct} (H_{t+\Delta-1} )
    \left( Y_{t,\Delta} - g_t(H_t)^\prime \alpha_c
      -  (A_t - \tilde{p}_t (1|X_t) ) f_t (X_t)^\prime \beta_c \right)^2 \right]
\end{equation}
where $\mathbb{P}_n [ \phi (H_{t+\Delta-1}) ]$ is defined as the average
of a function,~$\phi(H_{t+\Delta-1})$, over the sample. The centering refers to the centering of the
treatment indicator $A_t$ in the above weighted least squares
criteria. This criterion is similar to~\cite{Boruvkaetal};
however \cite{Boruvkaetal} restrict to $\Delta=1$ and thus
the weight $w_{ct}$ does not contain the ratio, $\frac{
  {\prod_{s=1}^{\Delta-1} } {\bf 1} [A_{t+s} = 0]}{\prod_{s=1}^{\Delta-1}
  p_{t+s} ( A_{t+s } \given H_{t+s}) }$.
Also \cite{Boruvkaetal}  assume a model for the treatment effect  $\beta(t; X_t)$ (as opposed to estimating the projection of this effect as is the case here).  Under finite moment
and invertibility assumptions,  the minimizers $(\hat\alpha_c, \hat\beta_c)$,
are consistent, asymptotically normal estimators  of $(\alpha_c^\star, \beta_c^\star)$.
The limiting variance of $\sqrt{N}(\hat\beta_c-\beta_c^\star)$ is given by  $Q_c^{-1} {W_c} Q_c^{-1}$ where
\begin{align*}
W_c = \E \bigg[ &\sum_{t=1}^T I_t \, w_{ct} ( H_{t+\Delta-1} )\,
{\epsilon}_{ct} (A_t - \tilde{p}_t (1 \given X_t) )  f_t(X_t) \\
&\times\sum_{t=1}^T  I_t \,w_{ct}( H_{t+\Delta-1} )\, {\epsilon}_{ct}
(A_t - \tilde{p}_t (1 \given X_t) )  f_t(X_t)^\prime \bigg], \nonumber \\
{\epsilon}_{ct} = Y_{t,\Delta} - &g_t(H_t)^\prime {\alpha_c^\star}
-  (A_t - \tilde{p}_t (1|X_t) )
f_t (X_t)^\prime{\beta_c^\star},
\text{ and } \nonumber \\
Q_c = \sum_{t=1}^T \E  &\bigg[ I_t \, \tilde{p}_t (1|X_t) (1 -
\tilde{p}_t (1|X_t) ) ) f_t(X_t) \, f_t(X_t)^\prime \bigg]. \nonumber
\end{align*}
See Appendix~\ref{app:asymptotics} for technical details.

The estimators of the coefficients in the projection of the marginal treatment effect,
$\beta_m$ minimize the following least-squares criteria over $(\alpha_m, \beta_m)$:
\begin{equation}
  \label{eq:marginal_ls}
  \mathbb{P}_n \left[ \sum_{t=1}^T I_t \, w_{mt} (H_{t+\Delta-1} ) \left( Y_{t,\Delta} - g_t(H_t)^\prime \alpha_m
      -  (A_t - \tilde{p}_t(1) ) f_t^\prime \beta_m \right)^2 \right]
\end{equation}
where the probability $\tilde{p}_t (a)$ defines the projection
(see above and Section~\ref{subsection:computations}).
Similarly under finite moment and invertibility assumptions,
the minimizers $(\hat\alpha_m, \hat\beta_m)$, are consistent,
asymptotically normal estimators  of $(\alpha_m^\star, \beta_m^\star)$.
See Appendix~\ref{app:asymptotics} for technical details.  For
expositional simplicity we focus on the test for the conditional
treatment effect in the remainder of this paper.
See Appendix~\ref{app:marginal_ss} for a
parallel discussion in the case of the marginal treatment effect.

The proposed sample size formula in the conditional setting is based on the test statistic
\begin{eqnarray}
  T_{cN}= N \hat{\beta}_c^\prime  \hat{Q}_c \hat{W}_c^{-1} \hat{Q}_c \hat{\beta}_c
  \label{eq:teststatistic}
\end{eqnarray}
where $N$ is the sample size and  $\hat{W}_c$ is given by
\begin{align*}
  \mathbb{P}_n \bigg[ &\sum_{t=1}^T I_t \, w_{ct} ( H_{t+\Delta-1} )\,
    \hat{\epsilon}_{ct} (A_t - \tilde{p}_t (1|X_t) )  f_t(X_t) \\
    &\times\sum_{t=1}^T  I_t \,w_{ct}( H_{t+\Delta-1} )\,
    \hat{\epsilon}_{ct} (A_t - \tilde{p}_t (1|X_t) )  f_t(X_t)^\prime \bigg]
\end{align*}
with $\hat{\epsilon}_{ct}= Y_{t,\Delta} - g_t(H_t)^\prime \hat{\alpha}_c
-  (A_t - \tilde{p}_t (1|X_t) ) f_t (X_t)^\prime\hat{\beta}_c$, and
$\hat{Q}_c$ is given by
\[
  \sum_{t=1}^T\mathbb{P}_n  \left[ I_t \, w_{ct} (H_{t+\Delta-1})
    (A_t - \tilde{p}_t (1|X_t) )^2 f_t(X_t) \, f_t(X_t)^\prime \right].
\]
Here we have implicitly assumed that $\hat{W}_c $ is invertible.
The following lemma provides the distribution of $T_{cN}$:

\begin{lemma}[Asymptotic Distribution of  $T_{cN}$]
  \normalfont
  \label{lemma:centering}
  Under finite moment and invertibility assumptions,
  $$
  N \left(\hat{\beta}_c-\beta_c^\star\right)^\prime \hat{Q}_c
  \hat{W}_c^{-1} \hat{Q}_c \left(\hat{\beta_c}-\beta_c^\star\right)
  \longrightarrow_d \chi^2_{q_c}.
  $$
\end{lemma}

From a technical perspective the above test statistic, $T_{cN}$, is very
similar to the quadratic form test statistics based on weighted
regression used in Generalized Estimating Equations method~\citep{Liang1986,Diggle2002}.
In this  field much work has been done on how to best adjust these
test statistics and their distribution when the sample size $N$ might
be small \citep{Liaoetal2015, Mancl2001}.  The adjustments are based on
the intuition that the quadratic form is akin to  the multivariate
T-test statistic used to test whether a vector of means is equal to
$0$ and thus Hotellings T-squared distribution is used to approximate
the distribution when $N$ may be small.
Here we follow the lead of this well developed area and use a
non-central Hotelling's T-squared distribution to approximate the distribution of $T_{cN}$.
Recall that if a random variable~$X$ has non-central Hotelling's
T-squared distribution with degrees of freedom~$(d_1, d_2)$ and
non-centrality parameter~$\lambda$ then~$\frac{d_2}{d_1 (d_1 + d_2
  - 1)} X$ has non-central F-distribution with the same
degrees of freedom and non-centrality parameter
\citep{Hotelling1931}.
In our setting  $d_1=q_c$ and  $d_2 = N - (q^\prime + q_c)$
and~$\lambda = N \gamma_c$ with
\begin{align}
  \gamma_{c} =  \left(\beta_c^\star\right)^\prime {Q_c} {W_c}^{-1} {Q_c}
  &\beta_c^\star \label{noncent}.
\end{align}
Recall that $q^\prime$ is the dimension of $\alpha_c$ and $q_c$ is the dimension of $\beta_c$.
See Appendix~\ref{app:technical} for a discussion of how for
large~$N$, we recover the Chi-Squared distribution given in Lemma~\ref{lemma:centering}.

Thus the rejection region for the test $H_0: \beta(t; x) = 0, t =
1\ldots, T$ and $x \in \{0,\dots, k\}$ is:
\begin{equation}
  \label{eq:reject}
  \left \{
    T_{cN}
    > \frac{q_c \, ( N - (q^\prime+1) )}{N- (q^\prime+q_c )}
    F_{q_c, N - (q^\prime+q_c);0}^{-1} \left( 1-\alpha_0 \right) \right \}
\end{equation}
with $\alpha_0$ a specified significance level. For details regarding
further small sample size adjustments, used when analyzing the data,  see
Appendix~\ref{app:ssa}.

\section{Sample size formulae}
\label{section:sample_size}

To plan the stratified micro-randomized study, we need to determine
the sample size needed, $N$, to detect a specific alternative
with a given power ($1-\beta_0$) at a given significance level ($\alpha_0$).
The sample size is the smallest value~$N$ such that
\begin{equation}
  \label{eq:ss}
  1-F_{q_c, N- (q^\prime+q_c ); N \gamma_c}
  \left( \frac{N-(q^\prime +1)}{N - (q^\prime+q_c)} F^{-1}_{q_c, N- (q^\prime+q_c ); 0} (1- \alpha_0) \right)
  \geq 1- \beta_0.
\end{equation}
$F_{d_1, d_2; \lambda }$ and $F^{-1}_{d_1, d_2; \lambda }$ denote the
cumulative and inverse distribution
functions respectively for the non-central $F$-distribution
with degrees of freedom~$(d_1, d_2)$ and non-centrality
parameter~$\lambda$.
Calculation of the sample size $N$ is non-trivial due to the unknown
form of the noncentrality parameter, $N\gamma_c$
(where $\gamma_c$ is defined in~\eqref{noncent}).
This is in contrast to micro-randomized trials where,
under certain working assumptions,
\cite{Liaoetal2015} were able to find an analytic
form for the noncentrality parameter~$ N \gamma_c$.

We outline a simulation based sample size calculation, starting with
general overview and comments in Section~\ref{subsection:computations}
and employ this calculator to design the smoking cessation study 
in Section~\ref{section:smoking_example}.

\subsection{Simulation based sample size calculation}
\label{subsection:computations}

As discussed above, calculation of the sample size~$N$ is non-trivial
due to the unknown form of the non-centrality parameter.
Here, we propose a three-step procedure for sample size calculations.

In the first step, equation~\eqref{noncent} and information elicited
from the scientist is used to calculate, via Monte-Carlo integration,
$\gamma_c$ in the non-centrality parameter.
The resulting value,%non-centrality parameter,
~$\hat{\gamma}_c$, is plugged in to
equation~\eqref{eq:ss} to solve for an \emph{initial} sample size~$\hat{N}_0$.
In the second step we use a binary search algorithm to search over a neighborhood of~$\hat{N}_0$;
in our simulations we found the binary search quickly resulted in a solution.
For each sample size~$N$ required by the binary search algorithm,
$K$ samples each of $N$ simulated participants are run.
Within each simulation, the rejection region for the test is given by
equation~\eqref{eq:reject} at the specified significance level.
The average number of rejected null hypotheses across the $K$ simulations
is the estimated power for the sample size~$N$.
The  sample size is the minimal~$N$ with estimated power
above the pre-specified threshold~$1-\beta_0$.

In the last, third, step we conduct a variety of simulations to assess the robustness of the sample size
calculator to any assumptions and to make adjustments to ensure robustness.
See our use of these simulations to test robustness in the case of the smoking cessation study in Section~\ref{section:smoking_example}.

Our sample size formula requires  the following information for $t=1,\ldots,T; x \in \{0,\ldots, k\}$:
\begin{enumerate}
\item desired type 1 and type 2 error rates,
\item targeted alternative~$\beta(t; x)$,
\item selected 	 probabilities~$\{ \tilde{p}_t (1 \given x) \}$,
\item  selected ``control variables''~$g_t (H_t)$,
\item the randomization formula  used to determine  $p_t(1|h)$  given a history $h$  and
\item a generative model for~$\{H_t\}_{t=1,\ldots,T}$.
\end{enumerate}

We provide general comments concerning the choice of
the above items and then build the sample size calculator for the smoking cessation study of Section~\ref{section:smoking_example}.
First we elicit information from the scientist to construct a specific
alternative form for ~$\beta(t; x)$.  A simple approach is to consider linear
alternatives,~$\{ \beta(t; x) = f_t (x)^\prime \beta_c^\star \}_{t=1,\ldots,T; x \in \{0,\ldots, k\}}$
so that the $L_2$ projection and the alternative coincide.  Stratification variables
are often categorical ($X$ is categorical); as a result we model the alternative
separately for each value of $X=x; x \in \{0,\ldots, k\}$.
Furthermore if we suspect
that the effect will be generally decreasing (with study time)  due to habituation,
then we might consider a vector feature, $f_t$ that represents a linear in time, $t$ trend.
Or we might believe that the effect of the treatments might be low at the beginning of the
study and then increase as participants learn how to use the treatment and then decrease
due to habituation; here we might consider a vector feature, $f_t$ that results in a quadratic trend.

The less complex the projection (smaller $q_c$) of the alternative~$\beta(t; x)$,
the smaller the required sample size, $N$, becomes. On the other hand,
the use of a simple projection for the alternative may not reflect the true
alternative $\beta(t;x)$ very well  (see appendix~\ref{app:tradeoff}
for a discussion of this tradeoff).   We suggest sizing a study for primary hypothesis tests
using the \emph{least} complex alternative possible.  For example, while there may be within day variation in
treatment effect, the study might still be  sized to detect treatment effects averaged across such variation --
i.e., a constant alternative within a day can result in a hypothesis test with sufficient power against a wide range of alternatives.
For example in the smoking cessation study the feature, $f_t (x)$ might be $(1, x, d(t), d(t) x)$
with $d(t)$ equal to the number of days following the \lq\lq quit smoking\rq\rq\ date.
The linear trend in days would be used to detect an approximately decreasing treatment effect, $\beta(t;x)$, with increasing $t$.

An objection to the above approach might be as follows.  Suppose that the scientific team believes that there will be an effect only at a very few
decision points within a day and thus a test statistic based on an $L_2$ projection
that averages over all decision points within the day would result in a test with low power.
However if investigators suspect this might be the case then more care should be
taken in selecting the decision points.   Consider the example of
Heartsteps~\citep{Klasjnaetal}, a mobile health intervention
focused on promoting physical activity and reducing sedentary behavior among sedentary office workers.
HeartSteps uses an  activity tracker to monitor steps taken on a per minute basis.
Originally decision points were set to match the frequency of data collection (i.e., each minute).
Upon reviewing activity data, it was discovered that the highest within person variability in
step count occurred at five timepoints throughout the
day with much less within person variability at other times.\footnote{These times were pre-morning commute,
mid-day, mid-afternoon,  evening commute and after dinner. Data collected was on individuals with ``regular'' daytime jobs.}
This information combined with the types of treatments being considered indicates that the treatment
might be most effective  at these 5 timepoints and potentially less effective at other times.
Therefore, decision times were selected to align with the five discovered timepoints.

To select the  probabilities~$  \{ \tilde{p}_t (1 \given x) \}_{t=1,\ldots,T; x \in \{0,\ldots, k\}}$,
recall that these probabilities define the weighting across time and across the stratification
distribution of the  alternative when operationalized as an $L_2$ projection.  To see this
suppose we decide to target a constant-across-time alternative and
select~$f_t(X_t)=( {\bf 1}_{X_t = 1}, {\bf 1}_{X_t = 2},\ldots, {\bf 1}_{X_t = k})^\prime$,
then~$\beta^\star_c = (\beta^\star_{c,1}, \beta^\star_{c,2}, \ldots, \beta_{c,k})$
where
\begin{align*}
  \beta_{c,x}^\star = &\left[ \sum_{t=1}^T \E [ I_t \, {\bf 1}_{X_t = x} ]
                        \tilde{p}_t (1 \given x) (1- \tilde{p}_t (1 \given x) )  \right]^{-1}  \\
                      &\left[ \sum_{t=1}^T \E [ I_t \, {\bf 1}_{X_t = x} ]
                        \tilde{p}_t (1 \given x) (1- \tilde{p}_t (1 \given x) ) \beta(t; x) \right]
\end{align*}
for~$x \in \{0,\ldots, k\}$.
If we set the reference probabilities to be constant in~$t,x$ then
\begin{align*}
  \beta_{c,x}^\star = &\left[ \sum_{t=1}^T \E [ I_t \, {\bf 1}_{X_t = x} ] \right]^{-1}
                        \left[ \sum_{t=1}^T \E [ I_t \, {\bf 1}_{X_t = x} ] \beta(t; x) \right].
\end{align*}
In this case $\beta_{c,x}$ is an average treatment effect across time weighted by
the fraction of time the participant is available and in stratification level $x$.
In our work we usually set $ \tilde{p}_t (1 \given x)$ to be constant in $(t,x)$ so
as to more easily discuss the targeted alternative with collaborators.

Next a decision should be made about which control variables~$g_t (H_t)$
should be included in the construction of the test statistic.
A natural control variable is the pre-decision time version of the proximal
response as this variable is likely highly correlated with the proximal response
and thus can be used to reduce variance in the estimation of the coefficients
for the $L_2$ projection.
For example in the smoking cessation study a natural
control variable is the fraction of time stressed in the hour prior to time~$t$.
One might want to include in the $q^\prime$ by $1$ vector, $g_t (H_t)$, many variables
so as to maximally reduce variance and thus increase the size of the noncentrality
parameter in \eqref{noncent}; indeed for fixed  $q^\prime$, the larger the noncentrality
parameter, the smaller the sample size $N$.
However from equation~\eqref{eq:ss} we
see that fixing all other quantities, the sample size $N$
increases with increasing $q^\prime$.
So intuitively there is a tradeoff between increasing the size of
the noncentrality parameter by including more variables in~$g_t (H_t)$   with the
resulting reduction in degrees of freedom in the denominator of the F test caused
by increasing  $q^\prime$, the number of variables in~$g_t (H_t)$.
See  appendix~\ref{app:tradeoff} for further discussion.

In the smoking cessation example below, we calculate the sample size with the vector of control variables~$g_t (H_t)$
set equal to~$f_t (X_t)$;
this maintains a hierarchical regression yet keeps $q^\prime$
as small as possible.
Incidentally this simplifies the development of the generative model as
additional time-varying variables are not included.

Generally the randomization formula has been determined by considerations of treatment burden,
availability and whether it is critical for the scientific question that the randomization
depend on a time-varying variable such as a prediction of risk.   Treatment burden considerations might impose a
constraint such as,  on average around $n$ treatments should occur over a specified time
period (e.g. an average of $n$ treatments per day); also the randomization formula might be
developed so as to limit the variance in the number of treatments in the specified time period.
In the smoking cessation study, the
randomization probability, $p_t(1|H_t)$ at decision point $t$ depends on at most
$\{X_s, I_s\}_{s=1,\ldots, t}$
(as opposed to the entire history, $H_t$).

%{\color{red}

The sample size formula requires the specification of a generative
model for the history~$H_t$ which achieves the specified alternative
treatment effect.  However existing data sets that include the use of
the required sensor suites and thus can be used to guide the form of
the generative model are often small and do not include treatment.
In the smoking cessation study, for example, we require 
a generative model for the multivariate distribution of $\{X_t, I_t,
A_t \}_{t=1}^T$ of which only the  distribution of $A_t$ given
$(H_t, I_t=1)$ is known (e.g. $p_t(1|H_t)$).
We have access to a small, observational, no-treatment data
set that included the required sensor suites and thus can be
used to guide the form of the generative model.  Because the data set
is small, in Section~\ref{section:smoking_example} we construct a low
dimensional Markovian generative model.  Here and in general, the
prior data does not include treatments.  Thus we use the
prior data to develop a generative model under no treatment.

The relatively simple generative model allows us to use only a
few summary statistics from this small noisy data set. This of
course, may lead to bias -- this bias would be problematic if the bias
results in sample sizes for which the power to detect the desired
effect is below the specified power.
Thus we also use the small data set to guide our assessment of
robustness of the sample size calculator. In particular, more complex
generative models can be proposed by exploratory data analysis. Of
course such complex alternatives may be due to noise and not reflect
the behavior of trial participants. In
Section~\ref{subsection:SMCdev}, we present results
of an exploratory data analysis in which we over-fit the noisy, small
data to suggest a particular complex
deviation from the simple Markovian generative model.

%}

We follow the three steps outlined at the beginning
of this subsection to provide a sample size $N$.
Our calculator also  provides standardized effect sizes.
That is, given the alternative
effect $\beta(t;x)$ and a generative model we calculate
the average conditional variance given by
$\bar{\sigma}_x^2 = (1/T) \sum_{t=1}^T \E [ \Var \left( Y_{t,\Delta}
  \given I_t = 1, A_t, H_t \right) \given I_t=1, X_t = x ]$.
Table~\ref{tab:std_effect1} in Appendix~\ref{app:add_details} provides standardized
treatment effect sizes, defined as, $d (t; x) = \beta (t;x) / \bar{\sigma}_x$.

\section{Smoking Cessation Study}
\label{section:smoking_example}
In the following, we use the above three step
procedure to form a sample size calculator for the
smoking cessation study.  Recall the last step involves
a variety of simulations to assess robustness to the
assumption underlying the generative model; this step is provided
in section~\ref{subsection:thirdstep}.

As noted previously, the smoking cessation study is a 10 day study; the first day is the
\lq\lq quit day\rq\rq, the day the participant  quits smoking.  Recall that participants wear the AutoSense sensor suite~\citep{Autosense:Ertin:2011} which provides a variety of physiological data streams that  are used by the stress classification algorithm.  A high
level view of the stress classification algorithm is as follows. First,
every minute % 5 seconds
a support vector machine (SVM)  algorithm is applied
to a number of ECG and respiration features constructed from the prior
one minute stream of sensor data.  The output of the SVM, e.g. the
distance of the features from the separating hyperplane, is then
transformed via a sigmoid function to obtain a stress likelihood in $(0,1)$;
see \cite{Hovsepian:2015} for details.
This output (in $(0,1)$) across the minute %5 second time
intervals is further
smoothed to obtain a smoother \lq\lq stress likelihood time series."
Next, a Moving Average Convergence Divergence approach is used to
identify minutes at which the trend in the stress likelihood  is going up
and  when it is going down; see \cite{Sarker:2016} for details.  The
beginning of an episode  is marked by the start of a positive-trend
interval; the \emph{peak} of an episode is the end of a positive-trend interval
followed by the start of a negative-trend interval.
If the area under the curve
from the beginning of the episode to the peak of the episode exceeds a
threshold then the episode is declared to be a stress episode.   The
threshold is based on prior data from lab experiments
and was evaluated on independent test data sets (from both lab and
field) in terms of the F1 score (a combination of sensitivity and
specificity~\citep{wiki:f1score}) for use in detecting physiological stress.
% \walt{Idea: Add a visualization to help understanding. From left to
%   right : (1) plot of noisy sensor measurements over an hour,
%   (2) stress likelihood, (3) a plot of the induced stress
%   classification, (4) arrows pointing from (3) to the inputs
%   to our SS calc.} \sam{ I would not do this as it gets people off track and  this is not our work}

A participant is available, $I_t=1$, for a treatment at minute $t$ if the
participant has  not received a treatment in the prior hour, if
this minute corresponds to a peak of an episode and if the minute is during the 10 hours since attiring  Autosense.
The stratification variable at every available minute (decision point)
$t$ is whether the criterion for stress is met ($X_t=1$) or  whether
the criterion for stress was not met  ($X_t=0$).
There are 600 decision times per day  (i.e., $10$ hours/day $\times$ $60$ minutes/hour)
at which, assuming the participant is available,
the participant may receive a treatment notification.  We plan the trial with 11 hour days in which during the
final hour participants cannot receive treatment. The final hour of data collection ensures we can calculate the proximal response
for the final  decision time each day.
Each participant should receive a daily average (over  the 10 hours) of 1.5 treatment notifications (notifications to practice the stress-reduction exercise on the app)
when $X_t=1$ and a daily average of 1.5 treatment notifications when $X_t = 0$.

Next, we build the simulation-based calculator
assuming the primary hypothesis is $H_0 : \beta(t; x) = 0; t = 1\ldots, T; x \in \{0,1\}$
and the  test statistic is as given in (\ref{eq:teststatistic}).
Small sample corrections are used in constructing the test statistic
as discussed in Section~\ref{section:cond_test_statistic};
see Appendix~\ref{app:ssa} for additional details.
% \sam{if you are going to use small sample corrections in this test statistic then you need to tell reader}

\subsection{Simulation-based calculator}
\label{section:calculator_smoking}
We start by choosing inputs for the sample size formula as outlined in Section~\ref{subsection:computations}.
We set the desired type~$1$ and type~$2$ error rates to be $5$\% and $20$\% respectively.
We next specify the targeted alternative~$\beta(t; x) = f_t (x)^\prime \beta^\star_c$
for~$\beta^\star_c \in \mathbb{R}^{q_c}$.
Suppose the scientific team suspects that if there is an effect of the mindfulness reminders, then this
effect might be negligible at the beginning of the study, increase as participants
begin to practice the mindfulness exercises and then the effect may decrease due
to habituation. Thus, we select
$f_t (X_t)^\prime =  \left( f^\prime_t \cdot {\bf 1}_{X_t = 0}, f^\prime_t \cdot {\bf 1}_{X_t = 1} \right)$
where
$f_t^\prime = \left( 1, \left \lfloor \frac{t-1}{600} \right \rfloor ,
  \left \lfloor \frac{t-1}{600} \right \rfloor^2 \right)$.
This leads to a non-parametric treatment effect model in the stratification variable~$X_t$, and
a piece-wise constant treatment effect model in time given~$X_t = x$ that
is quadratic as a function of ``day in study.''  In this case, the dimension of the~$L_2$ projection is~$q_c = 3 \cdot 2 = 6$,
$\beta_c^\star = (\beta_{c,0}^\star, \beta_{c,1}^\star) \in \mathbb{R}^{6}$ and the targeted
alternative is~$\beta(t;x) = f_t^\prime \beta_{c,x}^\star$
for~$x = 0,1$.
Next to elicit enough information from the scientist to specify~$\beta^\star_c$,
we ask scientists to specify for each level of $X$,
(1) an initial conditional effect,
(2) the day of maximal effect ($t^\star_x$)
and (3) the average conditional treatment effect~$\bar{\beta}_{c,x}
= T^{-1} \sum_{t=1}^{T} \beta(t;x)$.
This set of conditions uniquely identifies the
subvector~$\beta^\star_{c,x}$; therefore, the conditions
over each level of~$X$ combine to uniquely identify the
vector~$\beta^\star_c = (\beta^\star_{c,0}, \beta^{\star}_{c,1})$
as desired.
For this example, we will target the same alternative for both levels of the
stratification variable~$X_t$, thus $\beta^\star_{c,0}= \beta^{\star}_{c,1}$.   To set this common alternative, we use the following values:
the day of maximal effect is day~$5$ and
the initial conditional effect is $0$.   We consider three possible
common values of $\bar{\beta}_{c,0} = \bar{\beta}_{c,1}$
denoted~$\bar{\beta}$ in Table~\ref{tab:est_sample_size1}.

Here we set the control variables  to  $g_t (H_t)=f_t(X_t)$.
Furthermore suppose the formula for randomization probability depends only on past values of the
time-varying variable~$X_t$ and availability $I_t$.
We use the formula for~$p_t (a \given h_t)$ provided in Appendix~\ref{app:randprobs}.
One of the inputs to the randomization formula at an available decision point $t$  is the expected number of  episodes during the remaining part of the  day that will be classified as stressed ($X=1$) and the expected number of  episodes during the remaining day that will not be classified as stressed ($X=0$).  The generative model developed below is used to provide this input.   See appendix~\ref{app:randprobs} for further details and the specification of other inputs to this randomization formula.

\subsubsection{Generative Model}
\label{subsubsection:calculator_smoking}

%{\color{red} 
We now use a subset of the
data collected in an observational, no treatment,
smoking cessation study of $61$ cigarette smokers
\citep{Saleheen:2015}
to inform the generative model of longitudinal outcomes~$\{ X_t, I_t
\}_{t=1}^{T}$.    
%}
Study enrollment was restricted to smokers who reported smoking $10$ or more cigarettes per day
for the prior $2$ years and a high motivation to quit.
Enrolled participants select a smoking quit date.
Two weeks prior to the specified quit date,  participants
wore the AutoSense sensor suite [Ertin et al., 2011]  for~$24$ hours in their natural environment.
Participants again wore the sensor suite for~$72$ hours in
their natural environment starting on the specified quit date.
The same classification algorithm that is used in the smoking cessation example can be used with this data to produce the stress likelihood and associated episodes as described at the beginning of Section~\ref{section:smoking_example}.
Of the $61$ participants, $50$ had sufficiently high-quality
electrocardiogram data to construct the episodes and infer the stress classification for the $72$ hours post-quit.
% Each participant self-reported $10$ or more cigarettes per day
% for at least the past $2$ years and a high motivation to quit. Enrolled participants select a smoking quit date. Two weeks prior to the specified quit date, the participant wore a sensor suite for~$24$ hours in their natural environment.
% participants  wore the Autosense sensor suite for~$72$ hours in   their natural environment starting on their quit date.
% The sensor suite collected electrocardiogram (ECG) and respiration data which was used to infer the time-varying stress classification and auxiliary variable.
% Of the $61$ participants, $53$ had sufficiently high-quality   electrocardiogram (ECG) data to infer the stress   classification. of which only 50 have post-quit 72 hours.
This subset is reported in~\cite{Sarker:book:2017}. 
From this data we calculate the sample moments:
\begin{enumerate}
\item \label{p1} For each episode type (i.e., $x \in \{0,1\}$),
  the probability that the next
  episode will be a stress episode --
  i.e., a $2$ by 1 vector~$\bar{W}$
\item \label{p2} For each episode type (i.e., $x \in \{0,1\}$), the average episode length --  i.e., a $2$ by 1 vector~$\bar{Z}$
\end{enumerate}
These are: $\bar{W} = (6.7\%,51.9\%)$ and~$\bar{Z} =
(10.9,12.0)$; that is, the fraction of episodes not classified as stressed that are followed by an episode classified as stressed is $6.7\%$, the fraction of episodes classified as stressed that are followed by an episode classified as stressed is $51.9\%$,
the average length of episode not classified as stressed is $10.9$ minutes and the average length of an episode classified as stressed is $12.0$ minutes.

Using these sample moments we construct a no-treatment transition matrix for the joint process~$V_t = (X_t, U_t), t=1,\ldots, 600$ where~$X_t$ is the
time-varying stress classification and $U_t$ is the time-varying variable
indicating phases of the current episode -- ``pre-peak'', ``peak'',
and ``post-peak'' given by $U_t = 0,1,$ and $2$ respectively.  $U_t$ will be used to generate an availability indicator, $I_t$.
%the summary statistics for the~$72$ post-quit hours   for the~$50$ of the~$53$ participants with data for post-quit hours.
  % here, as mentioned in Section~\ref{subsection:motex}, we use summary
  % statistics calculated from a subset of data collected in an
  % observational, no treatment, study of individuals who were attempting
  % to quit smoking and for whom the AutoSense wearable suite was used
  % and stress classifications were recorded~\citep{Sarker:book:2017}.
%}
Each episode ends in state $V_{t} = (x, 2)$
for $x \in \{0,1\}$ and transitions to the beginning
of the next episode, $V_{t+1} = (x^\prime,0)$ for $x^\prime \in \{0,1\}$.
We restrict the transition matrix such that
for $x \in \{0,1\}$:
\begin{itemize}
\item $(x, 0)$ can \emph{only} transition to states
$(x,0)$ or $(x,1)$ (i.e., stay in state ``pre-peak'' or
transition to state ``peak'') from one minute to the next minute.
\item $(x,1)$ transitions immediately to $(x,2)$ with probability one
(i.e., $\pr (V_{t+1} = (x,2) \given V_{t} = (x,1) ) = 1$).
In other words the process inhabits the ``peak'' state for \emph{only} one minute.
\item $(x,2)$ can \emph{only} transition to states
$(x,2)$, $(0,0)$, \emph{or} $(1,0)$ (i.e., stay in state
``post-peak'' or end the episode and begin a new one).
\end{itemize}
We label each episode depending on the value~$x$.
The added complexity of the joint process~$V_t$
(in lieu of a generative model solely for $X_t$)
is used to accomodate the fact that  the scientific team decided to deliver treatment, if at all, only at  ``peaks'' of an
episode (i.e., $U_t \neq 1$ then $I_t = 0$).  Note that at the peak of
the episode, the episode is classified as stressed or not classified
as stressed. Define  $\tilde{Z}_{(x,u)}$ to be the length of the phase
$u$ in an episode of type $x$ after the chain enters state
$(x,u)$. Then $\tilde{Z}_{(x,1)} = 0$ for each~$x$ because as soon as
the chain enters the peak ($u=1$) of an episode, the chain departs.
Otherwise set $\tilde{Z}_{(x,u)} = (\bar{Z}_{x}-3)/2$ for $u=0$ and
$u=2$\footnote{We subtract three as we are guaranteed one pre-peak,
  one peak and one post-peak minute in each episode. Dividing by two
  splits the remaining average time evenly between pre-peak and
  post-peak phases of an episode.} (recall that $\bar{Z}_x$ is the
elicited average length, in minutes, of an episode classified as $x$,
under no treatment).

We set the no-treatment transition probability matrix to
\[
P^{(0)}_{(x,u), (x,u)} = \tilde{Z}_{x,u}/(\tilde{Z}_{x,u}+1)
\hspace{0.2cm} \text{ and } \hspace{0.2cm}
P^{(0)}_{(x,1), (x,2)} = 1.0
\]
for~$x \in \{0,1\}$
and $u \in \{0,2\}$,
and then set
\[
P^{(0)}_{(x, 2), (0,0)} = (1-\bar{W}_{x}) (1 - P_{(x,2), (x,2)})
\hspace{0.2cm} \text{ and } \hspace{0.2cm}
P^{(0)}_{(x, 2), (1,0)} = \bar{W}_{x} (1 - P_{(x,2), (x,2)})
\]
for $x \in \{0,1\}$ (recall that $\bar{W}_x$ is the elicited probability that the next
episode will be a stress episode).
All other entries of~$P^{(0)}$ are set to zero.
Thus~$P^{(0)}$ is a deterministic function of
the moments~$\bar{W}$ and $\bar{Z}$.  See Figure~\ref{tab:mcapprox2} for the transition matrix~$P^{(0)}$.

\begin{table}[!h]
  \caption{${P}^{(0)}$: Transition Matrix for the Markov chain, $V_t$, under No Treatment}
  \label{tab:mcapprox2}
  \centering
  \begin{tabular}{| c | c | c c c c c c c|} \hline
    \multicolumn{2}{|c|}{}
    & \multicolumn{3}{c}{Non-stress} & & \multicolumn{3}{c|}{Stress} \\
    \cline{3-5} \cline{7-9}
    \multicolumn{2}{|c|}{} & Pre-peak & Peak & Post-peak & & Pre-peak & Peak & Post-peak \\ \hline
    \multirow{3}{*}{Non-stress} & Pre-peak & 0.80 & 0.20 & 0.00 &
      & 0.00 & 0.00 & 0.00  \\
    & Peak & 0.00 & 0.00 & 1.00 & & 0.00 & 0.00 & 0.00 \\
    & Post-peak & 0.19 & 0.00 & 0.80 & & 0.01 & 0.00 & 0.00 \\ \cline{1-2}
    \multirow{3}{*}{Stress} & Pre-peak & 0.00 & 0.00 & 0.00 &
      & 0.82 & 0.18 & 0.00   \\
    & Peak & 0.00 & 0.00 & 0.00 &
    & 0.00 & 0.00 & 1.00 \\
    & Post-peak & 0.09 & 0.00 & 0.00 &
    & 0.09 & 0.00 & 0.82 \\ \hline
  \end{tabular}
\end{table}
%Table~\ref{tab:mcapprox2} provides the transition matrix based on this calculation for the time-varying variable~$V_t$ for the smoking cessation study using Inputs~\ref{p1} and~\ref{p2}.
The transition matrix~$P^{(0)}$ specified in Table~\ref{tab:mcapprox2}
has stationary distribution $(\pi_{(0,0)} = 39.4\%, \pi_{(0,1)} =
8.0\%, \pi_{(0,2)} = 39.4\%, \pi_{(1,0)} = 6.1\%, \pi_{(1,1)} =
1.1\%, \pi_{(1,2)} = 6.1\%)$.

\subsection{Generative model under treatment}
\label{section:undertreatment}
Next we form the generative model under treatment.   We make the simplifying assumption  that following treatment (i.e.,~$A_t = 1$)
stress, $V_{t+j}$,
evolves as a discrete-time Markov chain but with respect to
a different transition matrix~$P^{(1)}_{t}$ for each of the \emph{subsequent $j=1,\ldots, 60$ minutes}. After the hour, assuming a subsequent treatment notification is not provided,
the time-varying stratification variable returns to
evolution as a Markov chain with transition matrix~$P^{(0)}$. Thus,
\[
  \pr ( V_{t} = (x,u) \given V_{t-1} = (x^\prime,u^\prime), H_{t-1} )
  = \Bigg \{ \begin{array}{c c} \left[ P^{(0)}
               \right]_{(x^\prime,u^\prime), (x,u)} 
               & \text{if } A_{t-s} = 0,s = 1,\ldots,60
               \\ \left[ P^{(1)}_{t}\right]_{(x^\prime,u^\prime),
               (x,u)} 
               & \text{otherwise} \end{array}.
\]
Because the alternative~$\beta(t;x)$  is constant within each day, we will construct a transition matrix, $P^{(1)}_{t}$, that will only depend on $t$ through the day of decision $t$.  Thus we  use the notation ${P}_{d(t)}^{(1)}$  instead of  $P^{(1)}_{t}$ where $d(t)$ is the day of decision time~$t$.

Recall that in the smoking cessation study,  the treatment effect is the effect of providing
a notification at time $t$ to practice stress-reduction exercises and no more notifications
within the next hour versus no notification at time $t$ and no notifications over
the next hour on the percent of time stressed in the next hour.  Thus the  reference policy sets  the treatments~$a_{t+1},\dots,a_{t+\Delta-1}$
to $0$ and the expected proximal response under the reference policy is
\[
  \E \bigg [ \E \bigg [ \prod_{j=t+1}^{t+\Delta-1}
  \frac{1_{A_j=0}}{p_j(A_j|H_j)}Y_{t,\Delta} \, \bigg|
  \, A_t = a, H_t \bigg] \Given I_t = 1, X_t = x \bigg ]
\]
can be computed analytically for any combination of~$x$ and $a$
($\Delta=60$).
See Appendix~\ref{app:example_treatment_effect} for derivations of
the below analytic forms.
When~$a = 1$, under the proposed generative model the above
expectation is equal to~$\Delta^{-1} \sum_{s=1}^{\Delta} \sum_{u\in
  \{0,1,2\}} \left[ \left( P_{d(t)}^{(1)} \right)^{s} \right]_{(x,1),
  (1,u)}$.
When~$a = 0$, the expectation is equal to the fraction of time stressed within
the next hour under the reference policy
of no actions for that hour~$\Delta^{-1}  \sum_{s=1}^{\Delta} \sum_{u\in \{0,1,2\}}
\left[ \left( P^0 \right)^{s} \right]_{(x,1), (1,u)}$.

Given the alternative~$\beta(t;x)$ for a particular day,
we set ${P}_{d(t)}^{(1)}$ equal to
\[
  \arg \min_{Q \in \mathcal{P}} \sum_{x \in \{0,1\}} \left( \, \Delta^{-1}
    \sum_{s=1}^{\Delta} \sum_{u\in \{0,1,2\} }
    \left( \left[ Q^{s} \right]_{((x,1),(1,u))}
      - \left[ \left(P^{(0)}\right)^{s} \right]_{((x,1),(1,u))}
    \right) - \beta(t; x) \right)^2
\]
where~$\mathcal{P}$ denotes the set of transition matrices
which satisfy the constraints discussed above.
The set~$\mathcal{P}$ can be parameterized in
order to use general-purpose,
box-constrained optimization methods to
calculate~${P}_{d(t)}^{(1)}$ efficiently.
For all calculations, we initialize with inputs equivalent to
the transition matrix~$P^{(0)}$.
Using this procedure, the maximum squared distance
across all alternatives~$\beta(t;x)$
considered in this paper is $2.71 \times 10^{-11}$ (i.e.,
low approximation error).

% \sam{below seems off track--we want to illustrate how people confronted with the type of data mustafa can provide can actually form a generative model that is useful, e.g provides sample sizes that result in desired power}

\subsection{Generating the simulated data}
The prior section yields the no-treatment and treatment transition 
matrices~(i.e.,~$P^{(0)}$ and~$\{ P_{d}^{(1)} \}_{d=1}^{10}$)) given
the specified alternative~$\{ \beta (t;x) \}$.
We briefly show how to use this information along with the randomization
probability formula to generate synthetic data arising from a stratified micro-randomized trial.
First, we generate data for each day independently.  On a given day at time~$t$,
we first generate~$V_t$ using the transition equation in section~\ref{section:undertreatment}.
We then assess availability,~$I_t$, which is a deterministic function
of the current value of~$V_t$ and the past sixty minute history of actions~$\{ A_{t-s}
\}_{s=1}^{60}$. That is,~$I_t = {\bf 1} [ \sum_{s=1}^{60} A_{t-s} = 0
] \times {\bf 1} [U_t = 1]$.
Given~$I_t = 1$, we take the history~$H_t$ and 
generate the action at time~$t$, $A_t$, using the
given randomization probability formula~$p_t (1 \given H_t)$
found in appendix~\ref{app:randprobs}.
%\sam{remind reader where to find the formula for the randomization probability}
In order to compute the proximal outcome~$Y_{t,\Delta}$ for every minute over the ten hour day
(i.e.,~$t = 1,\ldots, 600$), we simulate an additional eleventh
hour during which participants cannot receive treatment (i.e.,
participants are unavailable). The above procedure generates
synthetic data for one participant in a stratified micro-randomized
trial.
% To generate simulated data, we 
% \sam{ we have described the construction of the no-treatment and treatment transition matrices but we have not told the reader how we will use these matrices.   I know we need to state that::::
% %using $\bar{W}$ (i.e., input~\ref{p1}) and $\bar{Z}$ (i.e., input~\ref{p2}); the assumption being that during a day with no treatment the time-varying $V_t$ evolves as a discrete-time Markov chain with this transition matrix.
% We generate the time-varying variable~$V_t$ for each day independently.   remind reader about the randomization probabilities.  Do we need to write something about how availabiity is generated.   Try not to be repetitive but put it all together.  Recall that in the example trial each day is an 11 hour day and during the final hour participants cannot receive treatment.
% %In other words, each day contains 10 hours of decision points and in the last 11th hour we only collect data.  we need to remind
% What else do we need to put here?}

\subsubsection{The test statistic}
\label{subsubsection:teststat_calculator_smoking}
The above provides the generative model for use in the simulation based sample size calculator.   Next consider the choice of the test statistic for use
in calculating the sample size.   In the test statistic,
(\ref{eq:teststatistic}), we set  the time $t$ reference probability as
  $\tilde{p}_t (1 \given x) =
  \sum_{x=0,1} \pi_{(x,1)} \left(1.5 / [ (600 - 1.5 \cdot 60)  \pi_{(x,1)} ]\right)
  = 2 \left(1.5 / [ (600 - 1.5 \cdot 60) ]\right)
  = 5.88 \times 10^{-3}$
(recall that the numerator of the weight, $w_{ct}$, in \eqref{eq:conditional_ls} is
$\tilde{p}_t (A_t\given x) {\prod_{s=1}^{\Delta-1} } {\bf 1} [A_{t+s}
= 0]$).
The probability,  $\tilde{p}_t(1|x)$ is equal to the daily average number of treatments while in state~$x$ divided by
the daily average number of times the participant is available and in state~$x$, marginalized
over the state~$x$.
In the denominator, the term~$1.5 \cdot 60$ is subtracted off the total number of decision points due to the availability
constraints following treatment; that is, the participant is unavailable for $60$ decision points following
a treatment notification and we require on average $1.5$ daily treatments while the participant is classified in
state~$x$, so the remaining number of decision points on average after taking into
account this deterministic constraint is approximately~$600 - 1.5 \cdot 60$.

The   test statistic, (\ref{eq:teststatistic}), with the above choice of reference probabilities, and the above generative model are used to generate the sample sizes in Table~\ref{tab:est_sample_size1}.  The column labeled, Sample Size, in this table provides the estimated
sample size to detect a specified alternative for the
conditional proximal effect
given power of~$80\%$ and significance level~$5.0\%$
for the smoking cessation study.
Recall that our input for the day of maximal effect
is day~$5$ and the input for
the initial conditional effect is $0$ for both levels
of the time-varying variable~$X_t$.
The average treatment effects~$\{ \bar{\beta}_x = T^{-1} \sum_{t=1}^T
\beta(t;x) \}_{x=0,1}$ are assumed equal across levels~$X$
and set to~$\bar{\beta}$; in the tables below  three values of $\bar{\beta}$ are considered.

\begin{table}[!htb]
    \caption{Estimated sample size, $N$, and
%    \sam{achieved}
    achieved power.
    }
    \label{tab:est_sample_size1}
    \centering
    \begin{tabular}{l | c c}
      & Sample size & Power \\ \hline
      $\bar{\beta} = 0.030$ & 50 & 80.6 \\
      $\bar{\beta} = 0.025$ & 67 & 80.7 \\
      $\bar{\beta} = 0.020$ & 127 & 80.6 \\ \hline
    \end{tabular}
\end{table}

\subsection{Evaluation of Simulation Calculator for the Smoking Cessation Study}
\label{subsection:thirdstep}

First to assess the quality of the sample size calculator
under an ideal setting we perform $1000$ simulations.
Each simulation is based on the
transition matrices~$P^{(0)}$ and~$\{ P_{d(t)}^{(1)} \}_{d(t)=1}^{10}$,
participant being unavailable for the hour following
treatment and at non-peak times, and the randomization probability~$p_t (1 \given H_t)$.
See the last column in Table~\ref{tab:est_sample_size1}.
Each simulation consists of generating data for~$N$ individuals
and performing the hypothesis test using equation~\eqref{eq:reject}
with the small-sample size adjustment described in Appendix~\ref{app:ssa}.
Appendix~\ref{app:marginal_ss} discusses the sample size calculations
with respect to marginal proximal effect for the smoking cessation study.

Recall the relatively simple generative model allowed us to use only a
very few statistics from a small data set, namely the data set described in Section~\ref{subsubsection:calculator_smoking}. This may lead to
bias which is problematic if the bias results in
sample sizes for which the power to detect the desired effect is below
the specified power. Therefore, here we  construct 
a feasible set of alternative generative models 
to which the sample size calculator should be robust.

First we evaluate  the sensitivity of the calculator
to the assumptions on the form of the transition matrix $P^{(0)}$. in the
next section we assess robustness to
the form of the transition matrix and, how as a result of the
assessment, we make the calculator more robust to the assumptions.

Second we evaluate the sensitivity of the calculator to deviations from a Markovian generative model.  Here we once again make use of the data set described in Section~\ref{subsubsection:calculator_smoking}.

\subsubsection{Misspecification of transition matrix~$P^{(0)}$}
\label{subsubsection:ball}

We start by testing robustness of the sample size calculator
to misspecification of the transition matrix~$P^{(0)}$
for the Markov chain,~$V_t$, under no treatment;
the treatment effect is still correctly specified.
We suppose the misspecification stems from
noise related to the use of sample moments from a small data set.
Let~$B_{(\epsilon, \epsilon^\prime)}$ denote
an~$(\epsilon,\epsilon^\prime)$-ball around the
inputs~$(\bar{W}, \bar{Z})$; that is,
\[
  B_{(\epsilon, \epsilon^\prime)} = \{ \, (W, Z)
  \hspace{0.2cm} | \hspace{0.2cm}
  \| W - \bar{W} \|_{\infty} \leq \epsilon
  \text{ and }
  \| Z - \bar{Z} \|_{\infty} \leq \epsilon^\prime
  \, \}.
\]
For each~$(W,Z) \in B_{(\epsilon,\epsilon^\prime)}$, we wish to
compute the achieved power under the alternative generative model
where~$V_t$ under no treatment
evolves as a Markov chain with transition matrix~$P$
constructed from inputs~$W$ and~$Z$.
In practice, this is computationally prohibitive as
the cardinality of~$B_{(\epsilon, \epsilon^\prime)}$ is large.
Simulation suggests power to be a smooth, non-increasing
function of both~$\epsilon$ and~$\epsilon^\prime$,
so instead we focus on computing
power for the following subset of $B_{(\epsilon, \epsilon^\prime)}$:
\begin{equation*}
  \Omega_{(\epsilon, \epsilon^\prime)} =
  \{ (W,Z) \, | \,
  W \in \bar{W} \pm \left \{ (\epsilon, -\epsilon), (\epsilon,
    \epsilon)\right \}
  \text{ and }
  Z \in \bar{Z} \pm \left \{ (\epsilon^\prime, -\epsilon^\prime),
    (\epsilon^\prime,\epsilon^\prime)\right \}.
\end{equation*}
For each pair~$(W,Z) \in \Omega_{(\epsilon, \epsilon^\prime)}$
we compute the associated transition matrix~$P$;
then we compute the sequence of transition matrices~$P_{d(t)}^{(1)}$
which maintain the correct alternative treatment effect.
We define the power for~$B_{(\epsilon, \epsilon^\prime)}$ to be the
minimum power across $(W,Z) \in \Omega_{(\epsilon, \epsilon^\prime)}$.
Table~\ref{tab:ball_1} presents
 achieved power under the previously calculated
sample sizes for~$\Omega_{(0.02,4)}$ and~$\Omega_{(0.01,2)}$
respectively.
For both~$(\epsilon,\epsilon^\prime) = (0.01,2)$
and~$(\epsilon,\epsilon^\prime) = (0.02, 4)$,
the achieved power is significantly
below the pre-specified 80\% level
for all three choices of the average treatment effect~$\bar{\beta}$.

\begin{table}[!htb]
%    \caption{Global caption}
      \caption{Misspecification of transition matrix~$P^{(0)}$:
        minimum achieved power \\
        over set of matrices in~$\Omega_{\epsilon, \epsilon^\prime}$}
      \label{tab:ball_1}
      \centering
        \begin{tabular}{c |c c}
          & \multicolumn{2}{c}{$(\epsilon, \epsilon^\prime) = $} \\
          & $(0.02,4)$ & $(0.01,2)$ \\ \hline
          $\bar{\beta} = 0.030$ & 57.5 & 61.5 \\
          $\bar{\beta} = 0.025$ & 43.9 & 52.2 \\
          $\bar{\beta} = 0.020$ & 40.4 & 65.6 \\ \hline
        \end{tabular}
\end{table}

\subsubsection{Deviations from a time-homogenous transition matrix
under no treatment}
\label{subsubsection:weekend}
Next we test robustness of the sample size calculator
to a different type of misspecification of the transition matrix~$P^{(0)}$, that of time-inhomogeneity;
as before the treatment effect is still correctly specified.
In particular suppose that the assumed transition matrix, $P^{(0)}$, is correct for weekdays but not for weekends;
in particular, suppose in reality that  the transition matrix
under no treatment on the weekend is~$P^{(0)}_{\text{weekend}} \neq P^{(0)}$.
The weekend is defined as~$d(t) = 6$ and $7$ (i.e., all
participants enter the study on a Monday).
We specify~$P^{(0)}_{\text{weekend}}$ via inputs~$(\bar{W}_{\text{weekend}},\bar{Z}_{\text{weekend}})$
which we set to two possible values
\[
\underbrace{( (0.04, 0.45), (10.9, 12.0) )}_{\text{weekend inputs }(1)}
\hspace{0.2 cm} \text{ or }  \hspace{0.2 cm}
\underbrace{( (0.10, 0.60), (10.9, 12.0) )}_{\text{weekend inputs }(2)}.
\]
Using the inputs we construct two alternate versions of what the true transition matrix~$P^{(0)}_{\text{weekend}}$ might be.
In the former, the individual is less likely to
enter a stress episode over the weekend; in the
latter, the individual is more likely to enter
a stress episode over the weekend.  In both cases,
the average episode lengths are assumed equal to~$\bar{W}$.

To test the calculator, we generate data using the no-treatment
transition matrices $P^{(0)}_{\text{weekend}}$ (for the weekend)
and $P^{(0)}$(for the weekday).  This data is simulated so that
the treatment effect used by the calculator is still correct (e.g. we
select the transition matrices under treatment, $P_{d(t)}^{(1)}$, to
ensure this).
However the expectation~$\E[w_{ct} (H_{t+\Delta-1}) Y_{t,\Delta}\mid
H_t, I_t=1]$ will not be quadratic in day-of-study.

Table~\ref{tab:power_weekendeffect}
presents achieved power  under these
alternative generative models.
We see that the achieved power is below the
pre-specified 80\% threshold in each case
except for~$\bar{\beta} = 0.020$ under weekend
input 1.
If the scientist thought such deviations feasible,
then the above analysis suggests for the smoking
cessation example that the sample size be set
to ensure a least~$80\%$ power \emph{over a
set of feasible choices for time-inhomogeneous
choices for the no-treatment transition matrix.}

\begin{table}[!ht]
  \caption{Estimated power under generative
    model with time-inhomogeneous Markov chain.}
  \label{tab:power_weekendeffect}
  \centering
  \begin{tabular}{c | r r}
    & \multicolumn{2}{c}{Estimated power} \\
    & Weekend Input 1 & Weekend Input 2 \\ \hline
    $\bar{\beta} = 0.030$ & 79.2 & 69.8 \\
    $\bar{\beta} = 0.025$ & 72.5 & 66.0 \\
    $\bar{\beta} = 0.020$ & 81.5 & 76.4 \\ \hline
  \end{tabular}
\end{table}

\subsubsection{Deviations from a Markovian generative model}
\label{subsection:SMCdev}

Here we use the data set described in Section~\ref{subsubsection:calculator_smoking} 
to construct feasible deviations to the simple Markovian generate
model. In particular, we present an exploratory data analysis
where we over-fit the noisy, small data  to build 
a more complex semi-Markovian generative model.  Due to the small size
of this data set, such complex alternatives may be due to noise and
not reflect the behavior of trial participants.  However these complex
alternatives can be used to assess robustness of our sample size
calculator. Therefore, after presenting data analysis suggesting the
semi-Markovian deviation, we then assess robustness of the sample size
calculator to this particular deviation.

%\sam{it must be clear where you are overfitting in the discussion below}

We start by considering the episodic transition
rule.  The Markovian model assumes that the
episode transitions only depend on the prior
episode classification. We test this by fitting
a logistic regression with episode classification
as the response variable with lagged values of episode classification
as well as additional summaries of past history, including
prior episode durations and time of day.
Analysis suggests that neither time of day
nor prior episode duration were statistically
significant.  We used forward selection to determine
the number of lagged values of episode classification.
Using this procedure, we include two lagged values of
episode classification in our over-fit model.
Table~\ref{tab:app_SMCtrans} presents the
estimates of the logistic regression along
with robust standard errors and confidence intervals.
The likelihood ratio test failed to reject the null when
comparing this model to the larger model in which
all interactions among the lagged variables were included
(i.e., a nonparameteric Markovian model).

\begin{table}[!ht]
\caption{Parameter estimates for the logistic regression.
Response is indicator of current episode being a stress
episode.}
\label{tab:app_SMCtrans}
\centering
\begin{tabular}{l | r r r r r}
\hline
Parameter & Estimate & Std. Error & $95\%$ LCL
  & $95\%$ LCL \\ \hline
  Intercept & $-2.83$ & $0.10$ & $-3.03$ & $-2.63$ \\
  $1$L Stress Ep. & $2.75$ & $0.20$ & $2.37$ & $3.14$ \\
  $2$L Stress Ep. & $0.71$ & $0.22$ & $0.27$ & $1.14$ \\ \hline
\end{tabular}
\end{table}

The over-fit, two-lagged Markovian model leads to
slightly distinct behavior of the transition
rules.  For example, given the prior episode was a stress episode,
the probability of the next episode being a
\emph{stress} episode ranges from 48.0\% (prior
episode was non-stress) to 65.2\% (prior
episode was stress).  Given the
prior expisode was a non-stress episode,
the probability of the next episode being a
\emph{stress} episode ranges from 5.6\% (prior
episode was non-stress) to 10.7\% (prior two
episodes was stress).
Table~\ref{tab:app_SMCtrans} suggests a
different Markovian model in which the state
is~$(X_t, U_t, L^{(1)}_t)$
where~$L^{(1)}_t$ denotes the classification of the
prior episode.

\begin{figure}[!ht]
\centering
\begin{subfigure}{.5\textwidth}
  \centering
  \includegraphics[width=.9\linewidth]{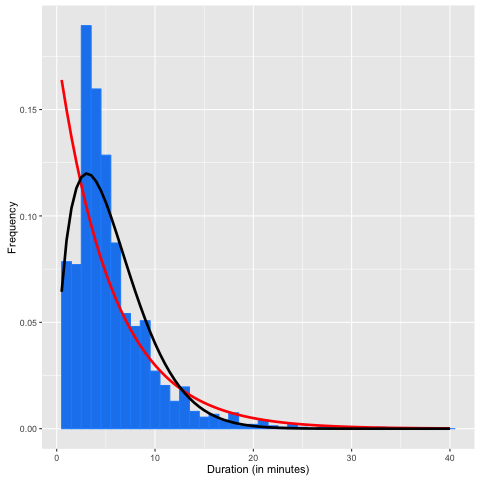}
  \caption{Pre-peak duration}
  \label{fig:sub1}
\end{subfigure}%
\begin{subfigure}{.5\textwidth}
  \centering
  \includegraphics[width=.9\linewidth]{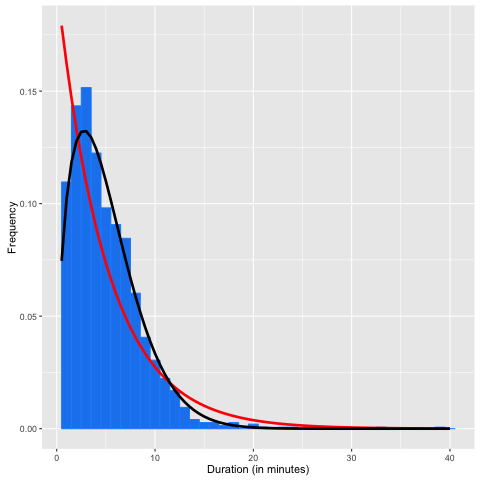}
  \caption{Post-peak duration}
  \label{fig:sub2}
\end{subfigure}
\caption{Histograms of duration for pre/post-peak durations.
  Empirical bayes pdfs for exponential (red) and weibull (black) densities
  are overlayed.}
\label{fig:durations}
\end{figure}

We next examine the pre and post peak durations.
Under the Markovian model, the duration of each
period is exponentially distributed.
Figure~\ref{fig:durations} shows histograms of the
duration of pre and post peak durations in the
analyzed subset of data along with empirical
Bayes estimates of the probability density functions
under both exponential and Weibull distribution
specifications.
We recognize the durations are discrete
and the above distributions are continuous.  These are
fit for simplicity. When generating the episode duration 
we generate a random variable from the continuous distribution 
and take the integer part of that random variable.
It is evident from the figures that the Weibull
distribution is more appropriate.  This is supported
by data analysis presented below.

Table~\ref{tab:app_SMCdur} presents the
parameter estimates for this over-fit model to the duration
data assuming a Weibull distribution\footnote{Models are fit to
duration minus one as pre and post peak durations are guaranteed
to be greater than one. Thus we are modeling the duration in
the state above the minimum value of one.}.  Like
the episodic transition rules, the post and pre
peak durations now depend on the current episode classification
as well as the prior episode classifications.
The exploratory data analysis suggests a semi-Markovian model in which
the pre/post peak durations are Weibull distributed, 
and the state is given by~$(X_t, U_t, L^{(1)}_t, L^{(2)}_t)$
where~$L^{(i)}_t$ denotes the classification of the
$i$th prior episode.

\begin{table}[!ht]
\caption{Parameter estimates for each Weibull
survival regression.}
\label{tab:app_SMCdur}
\centering
\begin{tabular}{l | r r r r r r r}
  & \multicolumn{3}{c}{Pre-peak} & & \multicolumn{3}{c}{Post-peak} \\ \cline{2-4} \cline{6-8}
Parameter & Estimate & Std. Error & p-value
  & & Estimate & Std. Error & p-value \\ \hline
  Intercept & $1.78$ & $0.016$ & $0.000$ &
                     & $1.59$ & $0.02$ & $0.000$ \\
  $0$L Stress Ep. & $-0.20$ & $0.037$ & $0.000$ &
                     & $0.45$ & $0.07$ & $0.000$ \\
  $1$L Stress Ep. & - & - & - & 
                     & $-0.21$ & $0.058$ & $0.004$ \\
  $2$L Stress Ep. & - & - & - &
                     & $-0.16$ & $0.07$ & $0.020$ \\
  Log(scale) & $-0.24$ & $0.015$ & $0.000$ &
                     & $-0.31$ & $0.05$ & $0.000$ \\ \hline
\end{tabular}
\end{table}

% The episodic transition rule may depend not only
% on the prior episode classification but the
% prior three episode classifications.
% See Table~\ref{tab:app_SMCtrans} in Appendix~\ref{app:SMC}
% for the transition model parameter estimates.
% Figure~\ref{fig:durations} shows histograms of the
% duration of pre and post peak durations in the
% analyzed subset of data.
% The figure suggests prior and post-peak duration
% are distributed Weibull.
% This is further supported by data analysis
% which also suggests dependence on
% the current and prior two episode classifications.
% See Table~\ref{tab:app_SMCdur} in Appendix~\ref{app:SMC}
% for the duration model parameter estimates.

Next we test robustness of the sample size calculator
to the semi-Markovian deviations described above.
To test the calculator, we generate data using the
no-treatment semi-Markov model specified in Appendix~\ref{app:SMC}.
The data is simulated so that the treatment effect used
by the calculator is correct. See Appendix~\ref{app:SMC} for
a discussion of how this was achieved.

Table~\ref{tab:power_SMCeffect} 
presents achieved power under these alternative 
generative models. 
We see that the achieved power is 
well above the pre-specified 80\%
threshold in each case. 
Therefore the sample size calculator is robust to
such complex deviations from the Markovian generative model.
For the given the alternative effect $\beta(t;x)$ 
and semi-Markov generative model we calculate
the standardized effects.  These are provided in
Table~\ref{tab:std_effect2} in Appendix~\ref{app:add_details}.

\begin{table}[!ht]
  \caption{Estimated power under 
    semi-Markov generative.}
  \label{tab:power_SMCeffect}
  \centering
  \begin{tabular}{c | r }
    & Estimated power \\ \hline
    $\bar{\beta} = 0.030$ & 93.6 \\
    $\bar{\beta} = 0.025$ & 88.0 \\
    $\bar{\beta} = 0.020$ & 93.6 \\ \hline
  \end{tabular}
\end{table}

\subsection{Adjustments to the  simulation-based calculator}

In section~\ref{subsection:thirdstep} we evaluated the simulation calculator
built in section~\ref{section:calculator_smoking}.
Here we make adjustments to the simulation calculator to ensure
robustness. 
First, we note that the simulation calculator is robust
to the potential semi-Markovian deviation discussed 
in Section~\ref{subsection:SMCdev}.
Next we make the decision that we are not
concerned with lack of robustness to
deviations from a time-homogenous transition matrix as
discussed in section~\ref{subsubsection:weekend}.
Therefore we focus on making the simulation calculator
robust to misspecification of Markov transition matrix as
discussed in section~\ref{subsubsection:ball}.

Analysis in section~\ref{subsubsection:ball} suggests for the smoking
cessation example that the sample size should be
set to ensure at least 80\% power \emph{over a set of
feasible choices for the transition matrix~$P^{(0)}$}.
We fix~$(\epsilon,\epsilon^\prime) = (0.01,2)$ to be our
tolerance to misspecification of the inputs.
For each set of inputs~$(W,Z) \in \Omega_{0.01,2}$,
we compute a sample size using the simulation calculator built in
Section~\ref{section:calculator_smoking}.
The maximum of this set of computed sample sizes is chosen
to ensure tolerance to misspecification of the transition
matrix. Table~\ref{tab:est_sample_size_robust} presents
the sample size under this procedure as well as the
achieved \emph{minimum power} over the set~$\Omega_{\epsilon,\epsilon^\prime}$.

\if1\comments
\sam{is the last column in table 4 the achieved power or is the
last column the power you computed  against the worst case in ~$\Omega_{\epsilon,
\epsilon^\prime}$? Tell reader. Need to be clear about this.
Separate the sample size calculation from the evaluation of the calculator.
What power do we get if the worst case scenario does not hold?
I suggest to make this table include both the power you calculated as
well as the achieved power}
\walt{I was a bit confused by this comment so I add here some
text to hopefully clear up the confusion.  Table~\ref{tab:ball_1}
is aimed at providing the \emph{achieved} power and is designed to
evaluate the calculator.  The table suggests a large dip
in achieved power.  The subsequent analysis is meant to serve
as an ``edit'' to the sample size calculation.  So the second column
in Table~\ref{tab:est_sample_size_robust} is the achieved power in the
more complex sample size calculation.}
\fi

\begin{table}[!htb]
    \caption{Estimated sample size, $N$, and
      computed power under~$\epsilon = 2$
      and $\epsilon^\prime = 0.01$.}
    \label{tab:est_sample_size_robust}
    \centering
    \begin{tabular}{l | c c}
      & Sample size & Minimum Power \\ \hline
      $\bar{\beta} = 0.030$ & 69 & 81.9 \\
      $\bar{\beta} = 0.025$ & 107 & 80.4 \\
      $\bar{\beta} = 0.020$ & 208 & 80.5 \\ \hline
    \end{tabular}
\end{table}

We have now used the three-step procedure to form a
sample size calculator for the smoking cessation study example.
For illustration suppose we wish to detect  an average conditional treatment
effect $\bar{\beta}$ equal to $0.025$.
%\sam{this paragraph does not tell the reader what we are doing to deal with the lack of
%robustness illustrated in tables 3 and 5.  summarize this for the reader.
%In particular we made the decision that we don't have to be concerned
%about the lack of robustness in table 5 but this is not explicitly stated.}
Based on the above discussion a sample size,~$N$, of $107$ would
be recommended to ensure power above the pre-specified 80\% threshold
across a set of feasible deviations from the assumed generative model.

\section{Discussion}
\label{section:alternatives}
In this paper we introduced the ``stratified micro-randomized trial''
and provided a definition and discussion of proximal treatment effects
along with the dependence of this definition on a reference
distribution. We proposed a simulation-based approach for determining
sample size and used this approach to determine the sample size for a simplified
version of the MD2K smoking cessation study.  We expect that similar
trial designs would be applicable in areas such as  marketing and
advertising in which each client is tracked and provided incentives,
e.g. treatments, repeatedly over time, and it is of interest to
determine in which contexts particular treatments are most effective.

%At first glance, the stratified micro-randomized
%trial design appears similar to the single case design
%frequently used in the behavioral sciences.
%However the estimand is quite different.  Indeed these trials are
%employed when scientists wish to compare the effect of one treatment
%versus another (treatment A versus treatment B) on an outcome but it
%is very expensive to recruit many participants.  The desired
%comparison does not generally have a dynamic component, nor should
%there by carryover effects.  Indeed  a critical assumption is that
%the effect of the treatment is only temporary
%(no carry-over effect) so that each participant can act as his own control.
%Additionally in general one assumes that the effect of a treatment is
%constant over time. In these trials, each participant is subject to
%periods of treatment  interspersed with periods of no treatment.
%For example during periods when a participant is on treatment one might expect the
%response to be generally higher than the  response during the time periods in which the
%participant is off treatment.
%An excellent overview of single case designs and their use
%for evaluating technology based interventions
%is~\cite{Dallery2013} and the data analyses focus on
%the examination of visual trends for each participant separately.
% This paper illustrates
%the visual analyses that would be conducted on each participant's
%data. \cite{Kratochwill2010}
%consider a variety of ways to introduce
%randomization into the single case design, including randomizing the
%order of treatment/baseline phases.

While the focus here is sample size considerations, stratified
micro-randomized studies yield data for a variety of interesting
secondary data analyses.  For example, understanding
predictors of future availability is of general interest
as keeping participants engaged in the mobile health intervention
is often of high concern.
Moreover, there is interest in using the data in constructing
``dynamic treatment regimes'' (e.g., just-in-time adaptive
interventions~\citep{Spruijt2014}). The stratified micro-randomized
trial improves such analyses by reducing causal confounding.

\bibliographystyle{plainnat}
\bibliography{str-mrts-refs}

\appendix
\section{Randomization probabilities}
\label{app:randprobs}

Here we provide a brief discussion of
how the randomization probabilities~$p_{t} ( a \given H_{t})$ might be calculated.
Suppose we require the participant to receive on average a certain number of interventions per day
at the various levels of the time varying covariates,~$N_x$ for $x \in
\{0,\dots,k\} := [k]$.
In the smoking cessation example, we have the time-varying covariate taking values in
$\mathcal{X} = \{ 0 = \text{Non-stressed}, 1 = \text{Stressed} \}$.
We aim for participants to receive on average one and a half
interventions per day when classified as stressed and one and a half
interventions per day when \emph{not} classified as stressed.
Formally, our randomization algorithm
is designed to satisfy the following:
\begin{equation}
\label{eq:numint_constraint}
\mathbb{E} \left[ \sum_{i=1}^T A_t I_t{\bf 1} [ X_{t_i} = x ]\right] = N_x
\end{equation}
for each~$x \in [k]$. The inputs of the randomization
algorithm are~$\{N_x\}$, a tuning parameter~$\lambda \in [0,1]$, and
a prediction, denoted by $g_t(x,r;h)$, at time $t$ of the number of
times in state~$x$ and available during the remaining time~$r$ given data $h$.
%We assume the participant is available for~$T^\prime \leq T$ minutes per day (e.g., $T^\prime = 600$ and $T = 6000$ in our motivating example).
%Randomization probabilities are reset daily.
The probability to assign treatment at time~$t$ given~$H_t$
is given by
\begin{equation}
\label{eq:randprobs}
p_t (1 \given H_{t}) = \frac{N_{X_{t}} - \sum_{s=1}^{t-1} \left[ \lambda_{s} A_{s} + (1 - \lambda_{s})
	p_s ( 1 \given H_{s} ) \right] {\bf 1} [ X_{s} = X_{t} ]}{ 1 + g_t(X_{t}, T - t; H_t )}
\end{equation}
where $\lambda_s = \lambda^{t-s}$.
In addition, we restrict the randomization probability within the interval $[\epsilon, 1-\epsilon]$.

To derive  (\ref{eq:randprobs}) we start by re-writing equation~\eqref{eq:numint_constraint} as
\begin{align*}
\EE \left[ \sum_{t=1}^{T} A_t \indicator{X_t = x} \right] = \EE \left[ \sum_{t=1}^{T} p_t (1 \given H_t)
\indicator{X_t = x} \right] %\label{obj}
\end{align*}
or for $\lambda_t \in (0,1)$,
\begin{align*}
\EE \left[ \sum_{t=1}^{T} A_t \indicator{X_t = x} \right]
=\EE \left[\sum_{t=1}^{T} \big(\lambda_t A_t + (1-\lambda_t) p_t (1 \given H_t) \big) \indicator{X_t = x}\right]
\end{align*}
The conditional expectation of the latter given the current history $H_s$ is
\begin{align*}
&\EE\left[\sum_{t=1}^{T} \big(\lambda_t A_t + (1-\lambda_t) p_t (1 \given H_t) \big) \indicator{X_t = x}
\given H_s \right] \notag
\\
=& \sum_{t=1}^{s-1}  \left( \lambda_t A_t + (1-\lambda_t) p_t (1 \given H_t) \right) \indicator{X_t = x}  + p_s (1 \given H_s)
+ \EE \left[\sum_{t=s+1}^{T} p_t (1 \given H_t) \indicator{X_t = x} \given H_s \right] %\label{condexpec}
\end{align*}
%Note that the above is true for any $\lambda_t \in \sigma(H_t)$.

We aim to find $p_s (1 \given H_s)$ such that
$N_x \approx \EE[\sum_{t=1}^{T} \left(\lambda_t A_t + (1-\lambda_t) p_t (1 \given H_t)\right)\indicator{X_t = x}| H_s]$
for each~$x \in [k]$.
However, at the end of $s$ decision time, we do not have access to
future randomization probabilities, i.e. $p_t (1 \given H_t)$ under
$X_t = x$ for $t \geq s+1$, which appears in the last term above. As
such, we approximate $p_t (1 \given H_t)$ by $p_s (1 \given H_s)$
whenever $X_t = x$, e.g. using the same randomization probabilities
for future time points, and obtain $p_s (1 \given H_s)$ by solving:
\begin{align*}
N_x = \sum_{t=1}^{s-1}  \left( \lambda_t A_t + (1-\lambda_t) p_t (1 \given H_t) \right) \indicator{X_t = x}  + p_s (1 \given H_s) +
\EE\left[\sum_{t=s+1}^{T} p_s (1 \given H_s) \indicator{X_t = x} \given H_s \right]
\end{align*}
That is,
\[
p_s (1 \given H_s) = \frac{N_x -\sum_{t=1}^{s-1}  \left[ \lambda_t A_t + (1-\lambda_t) p_t (1 \given H_t) \right]
	\indicator{X_t = x} }{1 + \EE\left[ \sum_{t=s+1}^{T}  \indicator{X_t = x } \given H_s \right]}
\]
The remaining problem is to approximate $g_t(1, T-s;H_s)=\EE \left[ \sum_{t=s+1}^{T}  \indicator{X_t = 1} \given H_s \right]$,
which is, at time $s$, the expected total number of future, available time points classified as $x$ given the current history.
This approximation is essentially a prediction problem.
There are a variety of approaches one can take depending on the data
available.  These prediction approaches may use distributional
assumptions, such as Markovian assumptions, on $X_t$ if there is only
a small amount of data to form these predictions or can employ more
black box predictions if there is a great deal of data.
In the actual smoking cessation study, there was only a small amount
of pre-existing data so a Markovian model was used to form the
predictions. For the smoking cessation study example, we set~$N =
(1.65, 2.15)$, $\lambda =0.3$, $\epsilon = 0.001$ and the remainder
function is set according to following rule:
\begin{itemize}
\item If~$(T-t) - N_{x} \cdot 60 < 60$ 
  then set $g_t (x, T-t; H_t) = (T-t) \pi_x$, 
\item else if~$(T-t) - N_{x} \cdot 60 < 120$ 
  then set $g_t (x, T-t; H_t) = (T-t - 60) \pi_x$, 
\item else set $g_t (x, T-t; H_t) = (T-t - 120) \pi_x$, 
\end{itemize}
where~$\pi$ is the stationary distribution of the markov transition
matrix~$P^{(0)}$. The reason for this complex rule is we must 
incorporate the fact that a participant is unavailable for the hour
following delivery of an intervention.

% We approximate this quantity by assuming stationarity of the Markov chain and therefore the denominator is equal to $(T-s) \pi_{(X_s,1)}$ (i.e., the proportion of remaining time in state~$(x,1)$ is simply the number of remaining decision points times the marginal probability $V_t = (x,1)$, $\pi_{(x,1)}$.) This provides a detailed derivation of equation~\eqref{eq:randprobs}.  Algorithm~\ref{alg_emi} provides pseudo-code for the sequential randomization probabilities with the remainder function~$g(x, r) = r \pi_{(x,1)}$ where~$\pi_{(x,1)}$ is the fraction of time in state~$(x,1)$.

%\begin{algorithm}
%	\caption{Randomization for EMI on each day}
%	\label{alg_emi}
%	\begin{algorithmic}[1]
%		\State \textbf{Input}: $N_x$ = average constraints.
%		$\epsilon_x$ = set of boundary constraints.  $\lambda \in [0,1]$ = weighting value.
%
%		\State \textbf{Initialize: } $S = vector(0, length = k)$, $s=1$
%
%		\While{$s < T$}
%
%		\State \textbf{Observe} $X_s$ and $\Delta_s$
%
%		\If{$X_s$ = unavailable}
%		\State $A_s = 0$ (Provide Nothing)
%		\Else
%		\If{$X_s = x$}
%		\State $S_x = \sum_{t=1}^{s-1} (\lambda^{s-t} A_s + (1-\lambda)^{s-t} p_s) {\bf 1}[X_t = X_s]$
%		\State $p_s = [N_x-S_x]/[1 + (T-(s+1)) \cdot \pi_{(x,1)}]$
%		\EndIf
%		\State Truncate $p$ at $[\epsilon_x, 1-\epsilon_x]$
%		\State Generate $A_s \sim \text{Binomial}(p_s)$
%
%		\EndIf
%		\State $s = s+1$
%		\EndWhile
%	\end{algorithmic}
%	\label{alg}
%\end{algorithm}

\section{Technical Arguments}
\label{app:technical}

\subsection{Treatment effects under potential outcomes framework}

\begin{proof}[Proof of Lemma \ref{lemma:cond_effect}]
We establish Lemma~\ref{lemma:cond_effect} for the marginal treatment effect.
The argument for the conditional treatment effect follows from a similar argument.

For~$a \in \{ 0,1\}$, we consider the
 \begin{align*}
\E &\left[\left(\prod_{j=1}^{t-1} p_j(a_j|H_j(\bar a_{j-1})) \right)
 Y_{t,\Delta} (\bar{a}_{t-1}, a, \bar 0)  I_t (\bar{a}_{t-1})\right] \\
=  \E & \left[\left(\prod_{j=1}^{t-1} p_j(a_j|H_j(\bar a_{j-1})) \right)
I_t (\bar{a}_{t-1})
 \E \left[ Y_{t,\Delta} (\bar{a}_{t-1}, a, \bar 0)  \given H_t (\bar{a}_{t-1}) \right]  \right]
 \end{align*}
Since the history~$H_t$ includes availability at time~$t$.  By consistency,~$H_t(\bar{A}_{t-1}) = H_t$
and~$I_t (\bar{a}_{t-1}) = I_t$ so the above is equal to
 \begin{equation} \label{eq:cons_version}
\E \left[\left(\prod_{j=1}^{t-1} p_j(a_j| H_j ) \right) I_t
 \E \left[ Y_{t,\Delta} (\bar{a}_{t-1}, a, \bar 0)  \given H_t \right]  \right]
 \end{equation}
Sequential ignorability implies that
\begin{align*}
 \E &\left[ Y_{t,\Delta} (\bar{a}_{t-1}, a, \bar 0)  \given H_t \right]  \\
 = \E &\left[ Y_{t,\Delta} (\bar{a}_{  t-1}, a, \bar 0)  \given H_{t}, A_t = a \right]
\end{align*}
Sequential ignorability also implies that
 $\E [ Y_{t,\Delta} (\bar{a}_{t-1}, a, \bar 0)  \given H_{t+k} ] \E [ 1_{A_{t+k} = 0} \given H_{t+k} ]$
 is equal to
$\E [ Y_{t,\Delta} (\bar{a}_{t-1}, a, \bar 0) 1_{A_{t+k} = 0} \given H_{t+k} ]$.
We apply this to show
\begin{align*}
 \E &\left[  Y_{t,\Delta} (\bar{a}_{t-1}, a_t, \bar 0)  \given H_t, A_t = a \right]  \\
 = \E &\left[ \E \left[ Y_{t,\Delta} (\bar{a}_{t-1}, a_t, \bar 0)  \given H_{t+1} \right] \given H_t, A_t = a \right]   \\
 = \E &\left[  \E \left[ Y_{t,\Delta} (\bar{a}_{t-1}, a_t, \bar 0)  \given H_{t+1} \right] \frac{\E [ 1_{A_{t+1} = 0} \given H_{t+1} ]}{p_{t+1} ( 0 \given H_{t+1} )}
 		\given H_t, A_t = a\right] \\
 = \E &\left[ \E \left[ Y_{t,\Delta} (\bar{a}_{t-1}, a_t, \bar 0)   \frac{1_{A_{t+1} = 0}}{p_{t+1} ( 0 \given H_{t+1} )}
 		\given H_{t+1} \right] \given H_t, A_t = a \right]
\end{align*}
Iteratively applying this argument, we end up with the following expectation
\begin{align*}
\E \bigg[ \frac{1_{A_{t+1} = 0}}{p_{t+1} ( 0 \given H_{t+1} )} \E &\left[ \cdots \E \left[ Y_{t,\Delta} (\bar{a}_{t-1}, a_t, \bar 0)   \frac{1_{A_{t+\Delta-1} = 0}}{p_{t+\Delta-1} ( 0 \given H_{t+\Delta-1} )} \given H_{t+\Delta-1} \right] \given H_{t+1} \right]
\given H_t, A_t = a \bigg] \\
&=  \E \left[ \prod_{j=t+1}^{t+\Delta-1} \frac{1_{A_{j} = 0}}{p_{j} ( 0 \given H_{j} )}  \, Y_{t,\Delta} (\bar{a}_{t-1})  \given H_t, A_t = a \right]
\end{align*}
Plugging this result into equation~\eqref{eq:cons_version}, we have
\[
\E \left[\left(\prod_{j=1}^{t-1} p_j(a_j| H_j ) \right) I_t
 \E \left[ \prod_{j=t+1}^{t+\Delta-1} \frac{1_{A_{j} = 0}}{p_{j} ( 0 \given H_{j} )}  \, Y_{t,\Delta} (\bar{a}_{t-1})  \given H_t, A_t = a \right] \, \right].
\]
Summing over all potential outcomes~$\bar{a}_{t-1}$ and normalizing yields
\begin{align*}
\E &\left[ \sum_{\bar{a}_{t-1}} \frac{\left(\prod_{j=1}^{t-1} p_j(a_j| H_j ) \right) I_t}{\E [ \sum_{\bar{a}_{t-1}} \left(\prod_{j=1}^{t-1} p_j(a_j| H_j ) \right) I_t ]}
 \E \left[ \prod_{j=t+1}^{t+\Delta-1} \frac{1_{A_{j} = 0}}{p_{j} ( 0 \given H_{j} )}  \, Y_{t,\Delta} (\bar{a}_{t-1})  \given H_t, A_t = a \right] \, \right] \\
 = \E &\left[ \sum_{\bar{a}_{t-1}} \frac{\left(\prod_{j=1}^{t-1} p_j(a_j| H_j ) \right) I_t}{\E [ \sum_{\bar{a}_{t-1}} \left(\prod_{j=1}^{t-1} p_j(a_j| H_j ) \right) I_t ]}
 \E \left[ \prod_{j=t+1}^{t+\Delta-1} \frac{1_{A_{j} = 0}}{p_{j} ( 0 \given H_{j} )}  \, Y_{t,\Delta} (\bar{a}_{t-1})  \given H_t, A_t = a \right] \given I_t = 1 \right] \\
  = \E &\left[ \E \left[ \prod_{j=t+1}^{t+\Delta-1} \frac{1_{A_{j} = 0}}{p_{j} ( 0 \given H_{j} )}  \, Y_{t,\Delta}  \given H_t, A_t = a \right] \given I_t = 1 \right].
\end{align*}
In the final equation, the outer expectation is with respect to the history~$H_t$ conditional on~$I_t = 1$.  That is, over \emph{both} past actions~$A_s$
and past observations~$O_s$ for $s < t$.

The above shows
\begin{align*}
 &\frac{\E \left[\left(\prod_{j=1}^{t-1} p_j(a_j|H_j(\bar a_{j-1})) \right)
	 Y_{t,\Delta} (\bar{a}_{t-1}, a, \bar 0)  I_t (\bar{a}_{t-1})\right]}
 	{\E \left[\left(\prod_{j=1}^{t-1} p_j(a_j|H_j(\bar a_{j-1})) \right)
	 I_t (\bar{a}_{t-1})\right]}\\
 =& \E \left[ \E \left[ \prod_{j=t+1}^{t+\Delta-1} \frac{1_{A_{j} = 0}}{p_{j} ( 0 \given H_{j} )}  \, Y_{t,\Delta}  \given H_t, A_t = a \right] \given I_t = 1 \right]
\end{align*}
which completes the proof.
\end{proof}

\subsection{Asymptotic consistency and normality}
\label{app:asymptotics}

We provide a detailed proof of asymptotic normality and consistency
in the conditional setting for the weighted-centered
least squares estimator.
The proof in the marginal setting
follows in a similar manner with only minor modification.
For ease of notation, we write~$\E_{\bf p}$ to denote
expectations where the distribution over actions is
with respect to the randomization probability~$\p_{\bf p}$,
and~$\E_{\eta}$ to denote expectations where the distribution
over actions is with respect to the chosen
\emph{reference distribution} -- that is,
providing a prompt at time~$t$ with
probability~$\tilde{p}_t (1 \given X_t)$
followed by no prompt over the next hour.

\begin{lemma}[Change from randomization probability to reference distribution]
\label{lemma:change_of_reference}
For any function~$\phi (H_{t+k})$ of the history
up to time~$t+k$, for~$k \geq 0$,
\[
\Ep \left[ w_{ct} (H_{t+\Delta-1} ) \phi ( H_{t+k} ) \given H_t
\right] = \Eeta \left[ \phi (H_{t+k} ) \given H_t \right]
\]
where
\[
  w_{ct} ( H_{t+\Delta-1} ) = \frac{\tilde{p}_t (A_t|X_{t}) {\prod_{s=1}^{\Delta-1} }{\bf 1} [A_{t+s} = 0] }{\prod_{s=0}^{\Delta-1} p_{t+s} ( A_{t+s } \given H_{t+s}) }.
\]
\end{lemma}

 \begin{proof}[Proof of Lemma~\ref{lemma:change_of_reference}]
   Suppose $\Delta = 2$.  Then
   \begin{align*}
     &\Ep \left[ w_{ct} ( H_{t+\Delta-1} ) \phi (H_{t+k}) \given H_t \right] \\
     = &\Ep \left[ \frac{\tilde{p}_t (A_t|X_{t})}{p_{t} ( A_{t } \given H_{t}) }
         \E \left[ \frac{
         {\bf 1} [A_{t+1} = 0]
         }{p_{t+1} ( A_{t+1} \given H_{t+1}) }
         \phi ( H_{t+k} )
         \given H_{t+1}, A_{t+1} \right] \given H_t \right]  \\
     = &\E \left[ \frac{\tilde{p}_t (A_t|X_{t})}{p_{t} ( A_{t } \given
         H_{t}) } \sum_{a} \frac{{\bf 1} (a = 0)}
         {\pi ( A_{t+1} = a \given H_{t+1} )} \pi (A_{t+1} = a \given H_{t+1} )
         \E \left[ \phi(H_{t+k}) \given H_{t+1}, A_{t+1} = a \right] \given H_t \right]\\
     = &\E \left[ \frac{\tilde{p}_t (A_t|X_{t})}{p_{t} ( A_{t } \given
         H_{t}) } \sum_{a} {\bf 1} [ a = 0 ] \cdot
         \E \left[ \phi(H_{t+k}) \given H_{t+1}, A_{t+1}=a \right] \given H_t \right]\\
     = &\E \left[ \frac{\tilde{p}_t (A_t|X_{t})}{p_{t} ( A_{t } \given
         H_{t}) } {\bf 1} [ A_{t+1} = 0] \E \left[  \phi(H_{t+k}) \given H_{t+1}, A_{t+1} \right] \given H_t \right] \\
     = &\sum_{a} \frac{\tilde{p}_t (A_t = a |X_{t})}{p_{t} ( A_{t} = a
         \given
         H_{t}) } p_t (A_t = a \given H_t) \E \left[  {\bf 1} [
         A_{t+1} = 0] \E \left[  \phi(H_{t+k}) \given H_{t+1}, A_{t+1}
         \right] \given H_t \right] \\
     = &\sum_{a} \tilde{p}_t (A_t = a |X_{t}) \E \left[  {\bf 1} [
         A_{t+1} = 0] \phi(H_{t+k}) \given H_{t} \right]
   \end{align*}
   Applying this argument iteratively leads
\[
     \sum_{a} \tilde{p}_t (A_t = a |X_{t}) \E \left[  \left( \prod_{j=1}^\Delta {\bf 1} [
         A_{t+j} = 0] \right) \phi(H_{t+k}) \given H_{t}, A_t = a \right]
     = \Eeta \left[ \phi (H_{t+k}) \given H_t \right]
\]
as desired.
 \end{proof}

Lemma~\ref{lemma:change_of_reference} yields many
important corollaries.
First,
\[
\Ep[ w_{ct} ( H_{t+\Delta-1}) \given H_t] = \Eeta [ 1 \given H_t ] = 1,
\]
which implies~$\Ep [w_{ct} ( H_{t+\Delta-1}) \given X_t ] = 1$.
Also, define
\begin{align*}
  \Ep [ w_{ct} (H_{t+\Delta-1}) Y_{t,\Delta} \given H_t]
  = \Eeta [ Y_{t,\Delta} \given H_t ]
  = \alpha (t;H_t)
\end{align*}

\begin{lemma} \label{lemma:2}
  For any $k \geq 0$ and function $\phi(H_{t+k})$, we have
  \begin{align*}
    \Ep &[w_{ct} (H_{t+\Delta-1}) \big(A_t - \tilde{p}_t(1 \given X_t) \big) \phi(H_{t+k}) \given X_t] \\
           &= \tilde{p}_t(1 \given X_t) (1-\tilde{p}_t(1 \given X_t))\E\left[
             \Big(\Eeta [  \phi (H_{t+k}) \given H_t, A_t = 1]
             - \Eeta [  \phi(H_{t+k}) \given H_t, A_t = 0]\Big) \given X_t\right]
  \end{align*}
  In particular, we have $\Ep [w_{ct} (H_{t+\Delta-1}) Y_{t,\Delta} (A_t - \tilde{p}_t(1 \given X_t)) \given X_t = x] =
  \tilde{p}_t(1 \given X_t) (1-\tilde{p}_t(1 \given X_t))\beta (t; X_t)$.
  Note that if~$\phi (H_{t+k})$ is a constant function of~$H_t$ and~$A_t$
  then the above expression is equal to zero.
\end{lemma}

\begin{proof}
\begin{align*}
  &\Ep [w_{ct} (H_{t+\Delta-1})
    \left( A_t - \tilde{p}_t(1 \given X_t) \right) \phi (H_{t+k}) \given X_t] \\
  = &\Ep \left[ \frac{\tilde{p}_t (A_t|X_{t})}{p_t (A_t|H_{t})} \left(A_t - \tilde{p}_t(1 \given X_t) \right)
      \Ep \left[ \frac{ {\prod_{s=1}^{\Delta-1} }{\bf 1} [A_{t+s} = 0] }{\prod_{s=1}^{\Delta-1} p_{t+s} ( A_{t+s } \given H_{t+s}) }\phi (H_{t+k}) \given H_t, A_t \right] \given X_t \right] \\
  = &\Ep \left[ \sum_{a \in \{0,1\}}  \frac{\tilde{p}_t (A_t \given X_{t})}
    {p_{t} ( A_{t } \given H_{t}) } p_t (A_t \given H_t)
    \left(a - \tilde{p}_t(1 \given X_t) \right) \Ep \left[
    \frac{ {\prod_{s=1}^{\Delta-1} }{\bf 1} [A_{t+s} = 0] }{\prod_{s=1}^{\Delta-1} p_{t+s} ( A_{t+s } \given H_{t+s}) }
    \phi (H_{t+k}) \given H_t, A_t = a \right]  \given X_t \right] \\
  = &\tilde{p}_t (1 \given X_{t}) \left(1 - \tilde{p}_t(1 \given X_t) \right)
    \E \left[ \Eeta \left[ \phi (H_{t+k}) \given H_t, A_t=1 \right]
    - \Eeta \left[ \phi (H_{t+k}) \given H_t, A_t=0 \right] \given X_t \right]
\end{align*}
If~$\phi (H_{t+k}) = w_{t+1,\Delta} Y_{t,\Delta}$, then by Lemma~\ref{lemma:change_of_reference} the above equals
\begin{align*}
  \tilde{p}_t (1 \given X_{t}) \left(1 - \tilde{p}_t(1 \given X_t) \right)
  &\E \left[ \E_{\eta} \left[ Y_{t,\Delta} \given H_t, A_t=1 \right] -
    \E_{\eta} \left[ Y_{t,\Delta} \given H_t, A_t=0 \right] \given X_t \right]\\
  = \tilde{p}_t (1 \given X_{t}) \left(1 - \tilde{p}_t(1 \given X_t) \right)
  &\E \bigg[ \Ep \left[
  \frac{ {\prod_{s=1}^{\Delta-1} }{\bf 1} [A_{t+s} = 0] }{\prod_{s=1}^{\Delta-1} p_{t+s} ( A_{t+s } \given H_{t+s}) }
  \phi (H_{t+k}) \given H_t, A_t = 1 \right] \\
  - &\Ep \left[
  \frac{ {\prod_{s=1}^{\Delta-1} }{\bf 1} [A_{t+s} = 0] }{\prod_{s=1}^{\Delta-1} p_{t+s} ( A_{t+s } \given H_{t+s}) }
  \phi (H_{t+k}) \given H_t, A_t = 0 \right] \given X_t \bigg] \\
  = \tilde{p}_t (1 \given X_{t}) \left(1 - \tilde{p}_t(1 \given X_t) \right) &\beta(t; X_t)
\end{align*}
as desired.
\end{proof}

Lemma~\ref{lemma:2} implies the function~$\tilde{p}_t$ must be
conditionally independent of~$H_t$ given~$X_t$ in order to guarantee a consistent
estimator of the projection of~$\beta^\star_c$.

\begin{lemma}
	The solutions~$(\hat{\alpha},\hat{\beta})$ that minimize equation~\eqref{eq:conditional_ls}
	are consistent estimators for:
	\begin{align*}
	\alpha^\star & = \left(\E \left[ \sum_{t=1}^T  g_t(H_t) g_t(H_t)'\right]\right)^{-1}
                       \E \left[ \sum_{t=1}^T g_t(H_t) \alpha (t; H_t) \right] \\
	\beta^\star &= \left(\E \left[  \sum_{t=1}^T \tilde{\sigma}_t^2 (X_t)
                      f_t(X_t) f_t(X_t)'\right]\right)^{-1}
                      \E \left[ \sum_{t=1}^T \tilde{\sigma}_t^2 (X_t)  f_t(X_t) \beta (t; X_t) \right]
	\end{align*}
        where~$\tilde{\sigma}_t^2 (X_t) = \tilde{p}_t(1 \given X_t) (1-\tilde{p}_t(1 \given X_t))$.
\end{lemma}

\begin{proof}
The solutions~$(\hat{\alpha},\hat{\beta})$ that minimize equation~\eqref{eq:conditional_ls}
are consistent estimators for the solutions that minimize the following
\[
\E \left[ \sum_{t=1}^T w_{ct} (H_{t+\Delta} ) \left( Y_{t,\Delta} - g_t(H_t)^\prime \alpha^\star
-  (A_t - \tilde{p}_t (1 \given X_t) ) f_t (X_t)^\prime \beta^\star \right)^2 \right]
\]
Differentiating the above equation with respect to~$\alpha^\star$
yields a set of~$p$ estimating equations.
\begin{align*}
0_{q^\prime}&= \Ep \left[ \sum_{t=1}^T w_{ct} (H_{t+\Delta-1}) \left( Y_{t,\Delta} - g_t(H_t)^\prime \alpha^\star
-  (A_t - \tilde{p}_t (1 \given X_t) ) f_t (X_t)^\prime \beta^\star \right) g_t(H_t) \right] \\
&= \sum_{t=1}^T \E \left[ \Ep \left[ w_{ct} (H_{t+\Delta-1}) \left( Y_{t,\Delta} - g_t(H_t)^\prime \alpha^\star
-  (A_t - \tilde{p}_t (1 \given X_t) ) f_t (X_t)^\prime \beta^\star \right)
\given H_t \right] g_t(H_t) \right]
\end{align*}
By Lemma~\ref{lemma:2}, $\Ep \left[ w_{ct} (H_{t+\Delta}) (A_t - \tilde{p}_t (1 \given X_t) )
f_t (X_t)^\prime \beta^\star \given H_t \right] = 0$.
Therefore, we have,
\begin{align*}
0_{q^\prime}&= \sum_{t=1}^T \E \left[ \Ep \left[ w_{ct} (H_{t+\Delta-1}) \left( Y_{t,\Delta} - g_t(H_t)^\prime \alpha^\star \right) \given H_t \right] g_t(H_t) \right]\\
&= \sum_{t=1}^T \E \left[ g_t(H_t) \alpha (t; H_t) - g_t(H_t)  g_t(H_t)^\prime \alpha^\star \right]
\end{align*}
and we have the desired equivalence.
The second equality is due to Lemma~\ref{lemma:change_of_reference}.
Now differentiating with respect to~$\beta^\star$ yields
\begin{align*}
0_{q_c}&= \Ep \left[ \sum_{t=1}^T w_{ct} (H_{t+\Delta-1}) \left( Y_{t,\Delta} - g_t(H_t)^\prime \alpha^\star
-  (A_t - \tilde{p}_t (1 \given X_t) ) f_t (X_t)^\prime \beta^\star \right) (A_t - \tilde{p}_t(1 \given X_t)) f_t (X_t) \right]
\end{align*}
By Lemma~\ref{lemma:2}, we have $\Ep \left[ w_{ct} (H_{t+\Delta-1}) (A_t - \tilde{p}_t (1 \given X_t) )
Y_{t,\Delta}  \given H_t \right] =
\tilde{p}_t (1 \given X_t) (1- \tilde{p}_t (1 \given X_t)) \beta_{\eta} (t; X_t)$, and $\Ep \left[ w_{t} (H_{t+\Delta-1}) (A_t - \tilde{p}_t (1 \given X_t) ) g_t(H_t)  \given H_t \right] = 0$.
The final term satisfies
\[
\Ep \left[ w_{t} (H_{t+\Delta-1}) (A_t - \tilde{p}_t (1 \given X_t) )^2 \given H_t \right] =
\tilde{p}_t (1 \given X_t) (1 - \tilde{p}_t (1 \given X_t) ) = \tilde{\sigma}_t^2 (X_t)
\]
by Lemma~\ref{lemma:change_of_reference}.
Then we have
\begin{align*}
0_{q}&= \sum_{t=1}^T \E \left[ \tilde{\sigma}_t^2 (X_t) f_t(X_t) \beta_{\eta}(t;X_t) -  \tilde{\sigma}_t^2 (X_t) f_t (X_t) f_t (X_t)^\prime \beta^\star \right]
\end{align*}
as desired.
\end{proof}

\begin{proof}[Proof of Asymptotic Normality in Lemma~\ref{lemma:centering}]
We now consider the issue of asymptotic normality.  First, let
\[
\epsilon_{ct} = Y_{t,\Delta} - g_t(H_t)^{\prime} \alpha^\star - (A_t-\tilde{p}_t(X_t)) f_t(X_t)^{\prime} \beta^\star,
\]
$\hat{\theta} = (\hat{\alpha}, \hat{\beta})$, and $\theta^\star = (\alpha^\star, \beta^\star)$.
Since $X_t \subset H_t$ we define
$h_t(H_t)^\prime = (g_t(H_t)^\prime, (A_t-\tilde{p}_t(1 \given X_t)) f_t(X_t)^\prime)$. Then
\begin{align*}
\sqrt{N} ( \hat{\theta} - \theta^\star ) &= \sqrt{N} \bigg \{ \bigg(
                                              \mathbb{P}_N \sum_{t=1}^T  w_{ct} (H_{t+\Delta-1})
                                              h_{t}(H_t) h_t(H_t)^\prime \bigg)^{-1}
                                              \bigg[ \bigg( \mathbb{P}_N \sum_{t=1}^T Y_{t,\Delta}
                                              w_{ct} (H_{t+\Delta-1} ) h_t(H_t) \bigg) \\
                                            &- \bigg( \mathbb{P}_N \sum_{t=1}^T
                                              w_{ct} ( H_{t+\Delta-1}) h_t(H_t) h_t (H_t)^\prime \bigg)
                                              \theta^\star \bigg] \bigg \} \\
                                            &= \sqrt{N} \bigg \{ E \bigg[  \sum_{t=1}^T
                                              w_{ct} ( H_{t+\Delta-1}) h_t(H_t) h_t (H_t)^\prime \bigg]^{-1} \\
                                            &\bigg[ \bigg( \mathbb{P}_N \sum_{t=1}^T
                                              w_{ct} ( H_{t+\Delta-1}) \epsilon_{ct}
                                              h_t(H_t) \bigg) \bigg] \bigg \} + o_p ( {\bf 1} )
\end{align*}
By definitions of $\alpha^\star$ and $\beta^\star$ and the previous consistency argument
\[
E \left[ \sum_{t=1}^T  w_{ct} ( H_{t+\Delta-1}) \epsilon_{ct} h_t(H_t) \right]  = 0
\]
Then under moments conditions, we have asymptotic normality with variance $\Sigma_{\theta}$ given by
\begin{align*}
\Sigma_{\theta} &= E \left[ \sum_{t=1}^T w_{ct} ( H_{t+\Delta-1}) h_t(X_t) h_t(X_t)^\prime \right]^{-1} \\
                &E \left[ \sum_{t=1}^T w_{ct} ( H_{t+\Delta-1}) \epsilon_{ct} h_t(X_t)
                  \times  \sum_{t=1}^T w_{ct} ( H_{t+\Delta-1}) \epsilon_{ct} h_t(X_t)^\prime \right] \\
                &E \left[ \sum_{t=1}^T w_{ct} ( H_{t+\Delta-1}) h_t(X_t) h_t(X_t)^\prime \right]^{-1}
\end{align*}
Due to centering and the previous lemma, the expectation of the matrix
$w_{ct} (H_{t+\Delta-1}) h_t(X_t)h_t(X_t)^\prime$ is block diagonal and therefore,
we can extract the sub-covariance matrix~$\Sigma_{\beta}$
``cleanly'' from the above formula.  Using this, we have
\begin{align*}
 \Sigma_{\beta} &=  \left[ \sum_{t=1}^T E[ (A_t - \tilde{p}_t (X_t) )^2 w_t ( H_{t+\Delta}) f_t (X_t) f_t (X_t)^\prime ] \right]^{-1} \\
	&\cdot E \bigg[ \sum_{t=1}^T w_t ( H_{t+\Delta} ) \epsilon_{ct}
                  (A_t - \tilde{p}_t(X_t)) f_t(X_t)
          \times  \sum_{t=1}^T w_t ( H_{t+\Delta} ) \epsilon_{ct}
                  (A_t - \tilde{p}_t(X_t)) f_t(X_t)^\prime \bigg] \\
 &\, \cdot \left[ \sum_{t=1}^T E[(A_t - \tilde{p}_t (X_t) )^2 w_t ( H_{t+\Delta}) f_t(X_t) f_t(X_t)^\prime \right]^{-1}
\end{align*}
as desired.
\end{proof}

\begin{proof}[Proof of asymptotic distribution of~$T_{cN}$]
By the above proof, we have
\[
\sqrt{N} ( \hat{\beta}_c - \beta^\star_c )^\prime \to N(0, \Sigma)
\]
where~$\Sigma = Q^{-1} W Q^{-1}$, as defined in Lemma~\ref{lemma:centering}.
Since a quadratic form of the normal distribution has a chi-square distribution,
we have
\[
N ( \hat{\beta}_c - \beta^\star_c )^\prime \Sigma^{-1} ( \hat{\beta}_c - \beta^\star_c ) \to \chi^2_{q_c}.
\]
$T_{cN}$ replaces~$\Sigma$ with a consistent estimator~$\hat{\Sigma}$.  By independence
of the covariance estimator and the equation above, we have
\[
q_c^{-1} N ( \hat{\beta}_c - \beta^\star_c )^\prime \hat{\Sigma}^{-1} ( \hat{\beta}_c - \beta^\star_c )
\sim^a F (q_c, N-q^\prime-q_c).
\]
For large~$N$, we have the F-distribution with degrees of freedom $q_c$ and $N-q^\prime -q$
is well approximated by a multiple of the $\chi^2$-distribution with degree of freedom~$q_c$;
that is, if $X_N \sim F(q_c, N-q^\prime-q_c)$ then $Y = \lim_{N\to \infty} q_c X_N$
has a chi-square distribution with degrees of freedom~$q_c$.

For large~$N$, asymptotic normality implies~$\hat{\beta}_c$ has
an approximate normal distribution with mean~$\beta^\star_c$ and variance~$\Sigma/N$.
This implies that~$\hat{\beta}_c \Sigma^{-1} \hat{\beta}_c$ has an approximate
non-central chi-square distribution with degrees of freedom~$q_c$ and non-centrality
parameter~$\beta^\star_c \Sigma^{-1} \beta^\star_c$.  That is,
\[
 N \cdot \hat{\beta}_c \Sigma^{-1} \hat{\beta}_c \sim^{a} \chi^2_{q_c} ( N \beta^\star_c \Sigma^{-1} \beta^\star_c )
\]
where~$\sim^a$ means approximately in distribution for large~$N$.
As~$\hat{\Sigma}^{-1}$ is a consistent estimator of~$\Sigma$,
this implies that
\[
T_{cN} = N \hat{\beta}_c \hat{\Sigma}^{-1} \hat{\beta}_c \sim^{a} q_c F(q_c, N-q^\prime -q; \beta^\star_c \Sigma^{-1} \beta^\star_c ).
\]
That is,~$\frac{1}{q_c} T_{cN}$ has an approximate non-central F-distribution.
\end{proof}

\section{Tradeoff between approximation error and degrees of freedom for
	sample size calculations}
\label{app:tradeoff}

This section provides a discussion of two tradeoffs involved in
attempting to reduce the required  sample size for a given power.
The first tradeoff involves the complexity of the projection of the
treatment effect and the second tradeoff involves the complexity of
the projection involving the control variables.
For expositional simplicity, we discuss the tradeoffs in the case
of sample size calculations for marginal treatment effects.
We make the simplifying assumptions that the proximal response
is a known function of the participant's data within a window of
length~$\Delta = 1$ and that the participant is always
available for treatment~(i.e.,~$I_t = 1$ for all $t=1,\ldots,T$).

In the marginal setting, the sample size is
the smallest integer~$N$ that satisfies:
\begin{equation}
\label{eq:ssmarg}
1-F_{q_m, N- (q^\prime+q_m); N \gamma_m}
\left( \frac{N-(q^\prime +1)}{N - (q^\prime+q_m)} F^{-1}_{q_m, N- (q^\prime+q_m ); 0} (1- \alpha_0) \right)
\geq 1- \beta_0.
\end{equation}
where~$\gamma_m$ is the non-centrality parameter
equal to $( \beta^\star_m )^\prime Q_m W_m^{-1} Q_m \beta_m^\star$
with
\begin{align*}
W_m = \E \bigg[ \sum_{t=1}^T w_{mt} ( H_t )\,
{\epsilon}_{mt} (A_t - \tilde{p}_t (1 ) )  f_t
&\times\sum_{t=1}^T  w_{mt}( H_{t} )\, {\epsilon}_{mt}
(A_t - \tilde{p}_t (1) )  f_t^\prime \bigg], \nonumber \\
Q_m = \sum_{t=1}^T \tilde{p}_t (1) (1 -
\tilde{p}_t (1) ) &f_t \, f_t^\prime,
\text{ and }
w_{mt} (H_t) = \frac{\tilde{p}_t (A_t)}{p_t (A_t \given H_t)},
\nonumber
\end{align*}
and $\epsilon_{mt} = Y_{t,\Delta} - g_t (H_t)^\prime \alpha_m^\star -
(A_t - \tilde{p}_t (1) ) f_t^\prime \beta_m^\star$.
$F_{d_1, d_2; \lambda }$ and $F^{-1}_{d_1, d_2; \lambda }$ denote the cumulative and inverse distribution
functions respectively for the non-central $F$-distribution
with degrees of freedom~$(d_1, d_2)$ and non-centrality
parameter~$\lambda$.

Define the error term,~$\epsilon_t$, by
\[
\epsilon_{t} = Y_{t,\Delta} -  E[w_{mt} (H_t) Y_{t,\Delta}\mid H_t] -
(A_t - \tilde{p}_t (1) ) \beta (t).
\]
The error term has conditional mean zero; that is,~$\E [ \epsilon_t
\given H_t, A_t ] = 0$.
See Appendix~\ref{app:tradeofftechnical} for the derivation.
Importantly we can write~$\epsilon_{mt}$ in terms of the
error term~$\epsilon_t$:
\begin{align*}
\epsilon_{mt} &= \epsilon_{t} +
\left( E[w_{mt} (H_t) Y_{t,\Delta}\mid H_t] -
g_t(H_t)^\prime \alpha^\star_m \right) +
(A_t - \tilde{p}_t  (1) ) \left( \beta(t) -
f_t^\prime \beta_m^\star \right) \\
&= \epsilon_{t} + e_{\alpha} (t; H_t) + (A_t - \tilde{p}_t  (1) ) e_{\beta} (t).
\end{align*}
The term~$e_{\alpha} (t; H_t)$ is the approximation error related to
the complexity of the $L_2$ projection involving the control variables.
The term~$e_{\beta} (t)$ is the approximation error related to the complexity of the
$L_2$ projection of the treatment effect.

The goal of the remainder of this section is to provide an intuitive understanding of the tradeoff among
the errors due to~$L_2$ projections (i.e., $e_{\alpha} (t; H_t)$ and
$e_{\beta} (t)$), sample size, and power.
We start by forming an approximate sample size formula
based on equation~\eqref{eq:ssmarg}.
In particular when~$N$ and~$q_m$ are large the sample size can be approximated by
\begin{equation}
\label{eq:approx_ss}
\gamma_m^{-1} \left( 2 \cdot z_{1-\beta_0}^2 +
\sqrt{2 q_m} z_{1-\alpha_0} +
2 z_{1-\beta_0} \sqrt{z_{1-\beta_0}^2 +
	q_m/2 + \sqrt{2 q_m} z_{1-\alpha_0} } \right).
\end{equation}
where~$z_{c} = \Phi^{-1} (c)$ is the inverse-normal distribution evaluated
at~$c \in (0,1)$.
See Lemma~\ref{lemma:ss_approx} in Appendix~\ref{app:tradeofftechnical} for technical details.

Given the approximate sample size formula,
we next discuss the tradeoff
involving the complexity of the
projection of the treatment effect.
We start by making three additional assumptions:
\begin{itemize}
	\item The control variables are correctly specified;
	that is,
	$\E [ w_{mt} (H_{t}) Y_{t,\Delta} \given H_t ] = g_t (H_t)^\prime \alpha^{\star}_m$
	(i.e., $e_{\alpha} (t; H_t) \equiv 0$).
	\item The error term~$\epsilon_t$ satisfies the second moment condition~$\E [ \epsilon^2_t \given H_t, A_t] = \sigma^2$.
	\item $\tilde{p}_t (1)=\E [ p_t (1 \given H_t)]$.
\end{itemize}
Under these conditions,
the non-centrality parameter~$\gamma_m$ can be approximated by
\begin{equation}
\label{eq:firstapprox}
\left[ \frac{\sigma^2}{\bar{\Theta} \sigma^2 + \bar{\Psi} \bar{e}_{\beta}^2} \right]
\left(\frac{\beta_m^\star}{\sigma}\right)^\prime \left( \sum_{t=1}^T
\tilde{p}_t (1) (1 - \tilde{p}_t (1) ) f_t f_t^\prime \right)
\left(\frac{\beta_m^\star}{\sigma}\right)
\end{equation}
where $\bar{e}^2_{\beta} = \frac{1}{T} \sum_{t=1}^T e_{\beta}^2 (t)$, and
\begin{align*}
\bar{\Psi} &= \frac{1}{T} \sum_{t=1}^T \left( \E \left[ \frac{\tilde{p}_t (1) ( 1 - \tilde{p}_t (1))^3
}{p_t (1 \given H_t)} + \frac{( 1 - \tilde{p}_t (1)) \tilde{p}_t
	(1)^3  }{1 - p_t (1 \given H_t)}  \right] - (1 -   \tilde{p}_t (1)
) \tilde{p}_t (1) \right)
\end{align*}
and
\begin{align*}
\bar{\Theta} &= \frac{1}{T} \sum_{t=1}^T \E \left[ \frac{\tilde{p}_t (1) (1 - \tilde{p}_t (1))}{p_t (1 \given H_t) (1 - p_t (1 \given H_t))} \right]
\end{align*}
It turns out that since~$\tilde{p}_t (1)=\E [ p_t (1 \given H_t)]$, then $\bar{\Psi} \geq 0$ and
$\bar{\Psi} = 0$ if
$p_t (1 \given H_t) = \tilde{p}_t (1) = \frac{1}{2}$.
See Lemma~\ref{lemma:noncentrality} in Appendix~\ref{app:tradeofftechnical} for technical details.

We now combine equations~\eqref{eq:firstapprox} and~\eqref{eq:approx_ss}.
For large values of~$q_m$ the  sample size
is approximately
\begin{equation}
\label{eq:tradeoff}
N\approxeq\frac{(z_{1-\beta_0} + z_{1-\alpha_0}) \sqrt{2q_m}}{(\beta_m^\star/\sigma) \left( \sum_{t=1}^T \tilde{p}_t (1) (1 - \tilde{p}_t (1) ) f_t f_t^\prime \right) (\beta_m^\star/\sigma)}
\left[ \bar{\Theta} + \bar{\Psi} \left( \frac{\bar{e}_{\beta}}{\sigma} \right)^2 \right].
\end{equation}
For a fixed average projected treatment effect~$T^{-1} \sum_{t=1}^T
f_t^\prime (\beta_m^\star/\sigma)$,
we have seen, in simulation, little variation
in~$(\beta_m^\star/\sigma)^\prime \left(
\sum_{t=1}^T \tilde{p}_t (1) (1 - \tilde{p}_t (1) ) f_t f_t^\prime \right) (\beta_m^\star/\sigma)$ as a
function of~$q_m$.  So suppose we fix this average projected treatment
effect at some value.
Then the tradeoff between the error in approximating the treatment
effect  and the dimension of the number, $q_m$, of stratification
variables used in this approximation is represented by the term
$\sqrt{q_m}\left[ \bar{\Theta} + \bar{\Psi} \left(
\frac{\bar{e}_{\beta}}{\sigma} \right)^2 \right]$.   This is quite
interesting as even when the size of  approximation error, $\bar{e}_{\beta}$ can be made sufficiently close to $0$ so that $\sqrt{q_m}  \bar{\Psi} \left(
\frac{\bar{e}_{\beta}}{\sigma} \right)^2$ is small, the term  $\sqrt{q_m}
\bar{\Theta}$ remains.  This argues for our recommendation that one use the least
complex (e.g. smallest $q_m$) projection for the treatment effect that
is reasonable in forming the test statistic that determines the sample
size. Also  if the randomization probabilities
were set to $1/2$ then $\Psi=0$, again supporting our recommendation of selecting
the least complex projection that is feasible.

We now turn to the tradeoff involving the complexity of the
projection with respect to the control variables.
We replace the assumptions in the prior
discussion with the following two assumptions:
\begin{itemize}
	\item The marginal treatment effect is correctly specified; that is,
	$\beta(t) = f_t^\prime \beta$
	\item The error term~$\epsilon_t$ satisifes the following second moment condition:
	$\E [ \epsilon^2_t \given H_t, A_t] = \sigma^2$.
\end{itemize}
Under these conditions, the non-centrality
parameter~$\gamma_m$ is approximated by
\[
\left[ \frac{\sigma^2}{\bar{\Theta} \sigma^2 + \bar{\Xi}_{\alpha}} \right] \left(\frac{\beta_m^\star}{\sigma}\right)^\prime \left( \sum_{t=1}^T \tilde{p}_t (1) (1 - \tilde{p}_t (1) ) f_t f_t^\prime \right)\left(\frac{\beta_m^\star}{\sigma}\right)
\]
where~$ \bar{\Xi}_{\alpha} = \frac{1}{T} \sum_{t=1}^T \Xi_{\alpha} (t)$, and
\[
\Xi_{\alpha} (t) = \E \left[ e_{\alpha}^2 (t; H_t) \frac{\tilde{p}_t (1 ) (1 - \tilde{p}_t (1))}{p_t (1 \given H_t) (1 - p_t (1 \given H_t))} \right].
\]
$\bar{\Theta}$ is as previously defined.
See Lemma~\ref{lemma:noncentrality2} in Appendix~\ref{app:tradeofftechnical} for technical details.

The above approximation of the non-centrality parameter~$\gamma_m$
implies that the approximate sample size formula given
by equation~\eqref{eq:approx_ss} has a multiplicative constant
equal to~$ (\bar{\Theta} + \frac{\bar{\Xi}_{\alpha}}{\sigma^2})$.
By including appropriately chosen control variables, one can hope for
a steep reduction in~$e_{\alpha}(t;H_t)^2$ and thus in $\bar{\Xi}_{\alpha}$.
For fixed sample size~$N$, significance level~$\alpha_0$,
and dimension~$q_m$, this leads to an increase in power.
The error reduction's impact on power is less dependent
on design (choice of randomization probabilities
and choice of $\tilde p_t$) unlike
the impact of error~$\bar{e}^2_{\beta}$ from the prior
discussion.
The dimension~$q^\prime$ does not appear in equation~\eqref{eq:approx_ss}
as we have assumed~$N - q^\prime - q_m \gg 0$ but $q^\prime$ does
appear in the small sample analog~\eqref{eq:ssmarg}.
Fixing all other quantities,  the sample size $N$
increases with increasing $q^\prime$ in \eqref{eq:ssmarg}.
This discussion shows the benefit
of choosing a small number of control variables that are strongly correlated with the proximal response.
In the smoking cessation study, for example, a natural control variable
is the fraction of time stressed in the hour prior to time~$t$.
This low dimensional~$(q^\prime = 1)$ control variable
may substantially lower~$\bar{\Xi}_{\alpha}$ leading to
an increase in power for fixed sample size.
%}

\subsection{Technical details}
\label{app:tradeofftechnical}

The following derivation shows that the error term~$\epsilon_t$
has conditional mean zero (i.e.,~$\E [ \epsilon_t \given H_t, A_t ] = 0$).
We do this by showing that we can write~$\E [Y_{t,\Delta} \mid H_t, A_t]$ as a
function of~$\E [ w_{mt} (H_t) Y_{t,\Delta} \mid H_t]$,~$\beta(t)$, and~$\tilde{p}_t (1)$:
\begin{align*}
\E[w_{mt} (H_t) Y_{t,\Delta}\mid H_t] &= \tilde{p}_t (1) \E[ Y_{t,\Delta}\mid H_t, A_t = 1] + (1-\tilde{p}_t  (1)) \E[Y_{t,\Delta}\mid H_t, A_t = 0] \\
\Rightarrow \E[Y_{t,\Delta}\mid H_t, A_t = 0] &= \E[w_{mt} (H_t) Y_{t,\Delta}\mid H_t] -  \tilde{p}_t  (1) \beta (t) \\
\E[Y_{t,\Delta}\mid H_t, A_t = 1] - \beta(t) &= \E[w_{mt} (H_t) Y_{t,\Delta}\mid H_t] -  \tilde{p}_t  (1) \beta (t) \\
\Rightarrow \E[ Y_{t,\Delta}\mid H_t, A_t = 1] &= \E[w_{mt} (H_t) Y_{t,\Delta}\mid H_t] +  (1-\tilde{p}_t  (1)) \beta (t) \\
\Rightarrow \E[ Y_{t,\Delta}\mid H_t, A_t] &= \E[w_{mt} (H_t) Y_{t,\Delta}\mid H_t] +  (A_t-\tilde{p}_t  (1)) \beta (t).
\end{align*}

We now deduce the approximate sample size formula~\eqref{eq:approx_ss}.

\begin{lemma} \normalfont
	\label{lemma:ss_approx}
	Given a specified significance level $\alpha_0$, power~$1-\beta_0$,
	and dimensions~$q^\prime$
	and $q_m$,
	then when~$q_m$ and sample size~$N$ are sufficiently large
	equation~\eqref{eq:ssmarg} implies that $N$ can be
	approximated by
	\begin{equation*}
	%\label{eq:approx_ss_app}
	\gamma_m^{-1} \left( 2 \cdot z_{1-\beta_0}^2 +
	\sqrt{2 q_m} z_{1-\alpha_0} +
	2 z_{1-\beta_0} \sqrt{z_{1-\beta_0}^2 +
		q_m/2 + \sqrt{2 q_m} z_{1-\alpha_0} } \right).
	\end{equation*}
	where~$z_{c} = \Phi^{-1} (c)$ is the inverse-normal distribution evaluated
	at~$c \in (0,1)$.
\end{lemma}

\begin{hproof}
	The following is a sketch proof based on asymptotic normal
	approximations for non-central chi-squared distributions.
	A complete proof requires careful consideration of uniform
	convergence results.

	We know that
	\[
	N ( \hat{\beta}_m - \beta_m^\star)^T \hat{Q}_m \hat{W}_m^{-1} \hat{Q}_m ( \hat{\beta}_m - \beta_m^\star) \to \chi^2_{q_m}.
	\]
	Thus in large samples the distribution of the test statistic~$T_{mN} = \hat{\beta}_m^T \hat{Q}_m \hat{W}_m^{-1} \hat{Q}_m \hat{\beta}_m$
	has an approximate~$\chi^2_{q_m} (N \gamma_m)$ distribution where
	$\gamma_m = (\beta^\star_m)^T \hat{Q}_m \hat{W}_m^{-1} \hat{Q}_m \beta^\star_m$.
	Thus sample size~$N$ is the smallest integer that satisfies
	\[
	1 - F_{\chi^2_{q_m} (N \gamma_m)} \left( F_{\chi^2_{q_m} (0)}^{-1} (1- \alpha_0) \right) \geq 1 - \beta_0.
	\]
	$F_{\chi^2_d (\lambda) }$ and $F^{-1}_{\chi^2_d (\lambda) }$ denote the cumulative and inverse distribution
	functions respectively for the non-central $\chi^2$-distribution
	with degree of freedom~$d$ and non-centrality
	parameter~$\lambda$.

	Now for large~$q_m$ we know that if~$U \sim \chi^2_{q_m} (0)$ then
	\[
	\frac{U - q_m}{\sqrt{2 q_m} }
	\]
	has an approximate standard normal distribution. Thus
	\[
	F_{\chi^2_{q_m} (0)} (u) \approx \Phi \left( \frac{u - q_m}{\sqrt{2 q_m}} \right)
	\]
	where~$\Phi$ is the cdf for the standard normal distribution.
	Thus
	\[
	F_{\chi^2_{q_m} (0)}^{-1} (1-\alpha_0) \approx z_{1-\alpha} \sqrt{2 q_m} + q_m
	\]
	where~$z_{1-\alpha_0} = \Phi^{-1} (1-\alpha_0)$.

	Next we know that a $\chi^2_{q_m} (N \gamma_m)$ is the distribution
	of~$\sum_{j=1}^{q_m} (X_j + \lambda_j)^2$ where~$X_j$ are iid standard
	normal random variables and~$\{ \lambda_j \}_{j=1}^{q_m}$ satisfy
	$\sum_{j=1}^{q_m} \lambda_j^2 = N \gamma_m$.
	But
	\begin{align*}
	\sum_{j=1}^{q_m} (X_j + \lambda_j)^2 &= \sum_{j=1}^{q_m} X_j^2 + 2 \sum_{j=1}^{q_m} \lambda_j X_j + \sum_{j=1}^{q_m} \lambda_j^2 \\
	&= \sum_{j=1}^{q_m} X_j^2 + 2 \sum_{j=1}^{q_m} \lambda_j X_j + N \gamma_m
	\end{align*}
	and we know that
	\[
	\left( \frac{\sum_{j=1}^{q_m} X_j^2 - q_m}{\sqrt{2 q_m}} , \frac{\sum_{j=1}^{q_m} \lambda_j X_j}{\sum_{j=1}^{q_m} \lambda_j^2} \right)
	\]
	converge in distribution as~$q_m \to \infty$ to independent standard normal
	random variables.
	Thus
	\[
	F_{\chi^2_{q_m} (N \gamma_m)} (v) = \pr \left( \sum_{j=1}^{q_m} (X_j + \lambda_j)^2 \leq v \right)
	\]
	For~$v = z_{1-\alpha_0} \sqrt{2 q_m} + q_m$, we have
	\[
	F_{\chi^2_{q_m} (N \gamma_m)} (v)  = \pr \left(  \frac{\sum_{j=1}^{q_m} X_j^2 - q_m}{\sqrt{2 q_m}} +
	\sqrt{\frac{2 N \gamma_m}{ q_m}} \frac{\sum_{j=1}^{q_m} \lambda_j X_j}{\sqrt{N \gamma_m} } + \frac{N \gamma_m}{\sqrt{2 q_m}} \leq z_{1-\alpha_0} \right) .
	\]
	For $q_m$ large but $N$ fixed the right hand side is approximately equal to
	\begin{align*}
	&\pr \left( Z_1 + \sqrt{\frac{2 N \gamma_m}{ q_m}} Z_2 \leq z_{1-\alpha_0} - \frac{N \gamma_m}{\sqrt{2 q_m}} \right) \\
	=&\pr \left( \frac{Z_1 + \sqrt{\frac{2 N \gamma_m}{ q_m}} Z_2}{\sqrt{1 + \frac{2 N \gamma_m}{q_m}}} \leq \frac{z_{1-\alpha_0} - \frac{N \gamma_m}{\sqrt{2 q_m}}}{\sqrt{1 + \frac{2 N \gamma_m}{q_m}}} \right) \\
	=& \Phi  \left( \frac{z_{1-\alpha_0} \sqrt{2 q_m} - N \gamma_m}{\sqrt{2 \cdot (q_m + 2 N \gamma_m)}} \right)
	\end{align*}
	where~$(Z_1,Z_2)$ are independent standard normal random variables.
	So we want the smallest integer~$N$ such that
	\begin{equation}
	\label{eq:approxssform}
	1-\Phi  \left( \frac{z_{1-\alpha_0} \sqrt{2 q_m} - N \gamma_m}{\sqrt{2 \cdot (q_m + 2 N \gamma_m)}} \right) \geq 1 - \beta_0.
	\end{equation}
	This yields the equation
	\[
	z_{1-\alpha_0} \sqrt{2 q_m} - N \gamma_m = -z_{1-\beta_0} \sqrt{2 \cdot (q_m + 2 N \gamma_m)}
	\]
	since~$z_{\beta_0} = - z_{1-\beta_0}$.
	Let~$y = \sqrt{q_m/2 + N \gamma_m}$; then we can rewrite the above equation as
	\[
	y^2 - 2 z_{1-\beta_0} y - \left( z_{1-\alpha_0} \sqrt{2 q_m} + \frac{q_m}{2} \right).
	\]
	The quadratic formula yields
	\[
	y = z_{1-\beta_0} \pm \sqrt{z_{1-\beta_0}^2 + \left( z_{1-\alpha_0} \sqrt{2 q_m} + \frac{q_m}{2} \right) }
	\]
	Solving for~$N$ yields
	\[
	N = \frac{1}{\gamma_m} \left( 2z_{1-\beta_0}^2 + z_{1-\alpha_0} \sqrt{2 q_m} \pm 2 z_{1-\beta_0} \sqrt{z_{1-\beta_0}^2 + z_{1-\alpha_0} \sqrt{2 q_m} + q_m/2} \right).
	\]
	It rests to find the correct sign for the final term.
	We know that sample size is the smallest integer~$N$ to
	satisfy equation~\eqref{eq:approxssform}.
	Using our formula for~$N$ we have
	\begin{align*}
	z_{1-\alpha_0} \sqrt{2 q_m} - N \gamma_m = -2z_{1-\beta_0}^2 \pm 2 z^2_{1-\beta_0} \sqrt{1 + \frac{1}{z_{1-\beta_0}^2} \left(z_{1-\alpha_0} \sqrt{2 q_m} + q_m/2\right)}.
	\end{align*}
	To satisfy equation~\eqref{eq:approxssform} for power greater than
	$50\%$ we need the left-hand side of the above equation to be negative.
	If the second term on the right-hand side of the above equation is
	positive then the whole right hand side is positive
	as the term within the square-root is greater than
	one.  Therefore, the second term must be negative; so
	sample size~$N$ is given by
	\[
	N = \frac{1}{\gamma_m} \left( 2z_{1-\beta_0}^2 + z_{1-\alpha_0} \sqrt{2 q_m} + 2 z_{1-\beta_0} \sqrt{z_{1-\beta_0}^2 + z_{1-\alpha_0} \sqrt{2 q_m} + q_m/2} \right).
	\]

\end{hproof}

Lemma~\ref{lemma:ss_approx} specifies a large sample  analytic relationship
among $N$, $q_m$,~$\alpha_0$,~$\beta_0$, and~$\gamma_m$.
Next, Lemma~\ref{lemma:noncentrality} establishes a relationship between the
non-centrality parameter~$\gamma_m$ and the approximation error due to
$L_2$ projection of the treatment effect (i.e.,~$\{ e_{\beta} (t) \}_{t=1,\ldots, T}$).

\begin{lemma} \normalfont
	\label{lemma:noncentrality}
	Recall~$e_{\beta} (t) := \beta(t) - f_t^\prime \beta^\star_m$,
	$e_{\alpha} (t; H_t) := \E [ w_{mt} (H_{t}) Y_{t,\Delta} \given H_t ] - g_t (H_t)^\prime \alpha^{\star}_m$,
	and the error term~$\epsilon_t$ is given by
	\[
	\epsilon_t = Y_{t,\Delta} -  \E [ w_{mt} (H_{t}) Y_{t,\Delta} \given H_t ] -
	(A_t - \tilde{p}_t (1)) \beta (t)
	\]
	and satisfies $\E [ \epsilon_t \given H_t, A_t] = 0$.
	We make the following assumptions:
	\begin{itemize}
		\item The control variables are correctly specified;
		that is,
		$\E [ w_{mt} (H_{t}) Y_{t,\Delta} \given H_t ] = g_t (H_t)^\prime \alpha^{\star}_m$
		(i.e., $e_{\alpha} (t; H_t) \equiv 0$).
		\item The error term~$\epsilon_t$ satisfies the second moment condition~$\E [ \epsilon^2_t \given H_t, A_t] = \sigma^2$.
		\item $\tilde{p}_t (1)=\E [ p_t (1 \given H_t)]$.
	\end{itemize}
	Under these conditions,
	the non-centrality parameter~$\gamma_m$ can be approximated by
	\[
	\left[ \frac{\sigma^2}{\bar{\Theta} \sigma^2 + \bar{\Psi} \bar{e}_{\beta}^2} \right]
	\left(\frac{\beta_m^\star}{\sigma}\right)^\prime \left( \sum_{t=1}^T
	\tilde{p}_t (1) (1 - \tilde{p}_t (1) ) f_t f_t^\prime \right)
	\left(\frac{\beta_m^\star}{\sigma}\right)
	\]
	where $\bar{e}^2_{\beta} = \frac{1}{T} \sum_{t=1}^T e_{\beta}^2 (t)$, and
	\begin{align*}
	\bar{\Psi} &= \frac{1}{T} \sum_{t=1}^T \left( \E \left[ \frac{\tilde{p}_t (1) ( 1 - \tilde{p}_t (1))^3
	}{p_t (1 \given H_t)} + \frac{( 1 - \tilde{p}_t (1)) \tilde{p}_t
		(1)^3  }{1 - p_t (1 \given H_t)}  \right] - (1 -   \tilde{p}_t (1)
	) \tilde{p}_t (1) \right)
	\end{align*}
	and
	\begin{align*}
	\bar{\Theta} &= \frac{1}{T} \sum_{t=1}^T \E \left[ \frac{\tilde{p}_t (1) (1 - \tilde{p}_t (1))}{p_t (1 \given H_t) (1 - p_t (1 \given H_t))} \right]
	\end{align*}
	Note that if~$\tilde{p}_t (1)=\E [ p_t (1 \given H_t)]$, then $\bar{\Psi} \geq 0$ and
	$\bar{\Psi} = 0$ {\bf if and only if}
	$p_t (1 \given H_t) = \tilde{p}_t (1) = \frac{1}{2}$ (see proof below).
\end{lemma}

\begin{proof}
	Under the above assumptions, the model error~$\epsilon_{mt}$ decomposes
	into two components:
	\begin{align*}
	\epsilon_{mt} &= \epsilon_{t} + (A_t - \tilde{p}_t (1) ) ( \beta(t) -
	f_t^\prime \beta^\star_m ) \\
	&= \epsilon_{t} + (A_t - \tilde{p}_t (1)) e_{\beta} (t).
	\end{align*}
	Plugging the decomposition into the
	formula for~$W_m$, we have
	\begin{align*}
	W_m= \E & \bigg[ \sum_{t=1}^T w_{mt} ( H_{t} )\, \bigg( \epsilon_{t}
	+ (A_t - \tilde{p}_t (1)) e_{\beta} (t) \bigg)(A_t - \tilde{p}_t (1 ) )  f_t \\
	& \times \sum_{t=1}^T w_{mt} ( H_{t} )\, \bigg( \epsilon_{t}
	+ (A_t - \tilde{p}_t (1)) e_{\beta} (t) \bigg)(A_t - \tilde{p}_t (1) )  f_t^\prime .
	\end{align*}
	We can decompose the above formula into various terms,
	which we now walk through step by step.
	The first term involves only the error~$\epsilon_t$.
	By the assumption that~$\E [ \epsilon_t \given H_t, A_t ] = 0$
	and~$\Delta = 1$ the cross terms are zero,
	we have
	\[
	\sum_{s \neq t} \E \bigg[ w_{mt} ( H_{t} ) \, \epsilon_{t} (A_t -   \tilde{p}_t (1) )^2 f_t  \times
	w_{ms} ( H_{s} ) \, \epsilon_{s} (A_s -   \tilde{p}_s (1) )^2 f_s^\prime \bigg].
	\]
	equals zero; then we have the term
	\begin{align*}
	\E \bigg[ \sum_{t=1}^T w^2_{mt} ( H_{t} )\, \epsilon^2_{t} (A_t -   \tilde{p}_t (1) )^2 f_t f_t^\prime \bigg]
	&= \sum_{t=1}^{T} \E \left[ w^2_{mt} ( H_{t} )\, (A_t -   \tilde{p}_t (1) )^2 \E[ \epsilon^2_{t} | H_t, A_t] \right] f_t f_t^\prime \\
	&= \sigma^2 \sum_{t=1}^{T} \E \left[ w^2_{mt} ( H_{t} )\, (A_t -   \tilde{p}_t (1) )^2 \right] f_t f_t^\prime \\
	&= \sigma^2 \sum_{t=1}^{T} \E \left[ \frac{\tilde{p}_t (1)^2 (1 - \tilde{p}_t (1))^2}{p_t (1 \given H_t)} +
	\frac{\tilde{p}_t (1)^2 (1 - \tilde{p}_t (1))^2}{1-p_t (1 \given H_t)} \right] f_t f_t^\prime \\
	&= \sigma^2 \sum_{t=1}^{T} \E \left[ \frac{(\tilde{p}_t (1) (1 - \tilde{p}_t (1)))^2}{p_t (1 \given H_t) (1- p_t (1 \given H_t))} \right] f_t f_t^\prime.
	\end{align*}

	Due to the same reasoning above, the cross-terms
	involving both $\epsilon_t$ and error $e_{\beta} (t)$
	are zero.

	The next term involves only the
	approximation error~$(A_t -  \tilde{p}_t (1) ) e_{\beta} (t)$:
	\begin{align*}
	W_{m,\beta} = \E \bigg[ &\sum_{t=1}^T w^2_{mt} ( H_{t} ) \,
	e_{\beta} (t)^2 (A_t -   \tilde{p}_t (1) )^4 f_t f_t^\prime \bigg] \\
	+ \E \bigg[ &\sum_{s \neq t} w_{mt} ( H_{t} )\,  e_{\beta} (t) (A_t -   \tilde{p}_t (1) )^2 f_t
	\times w_{ms} ( H_{s} )\,  e_{\beta} (s) (A_s -   \tilde{p}_s (1) )^2 f_s^\prime \bigg]
	\end{align*}
	We first investigate the second term (i.e. the cross-product term).
	Taking expectations, we have
	\begin{align*}
	&\sum_{s \neq t}  e_{\beta} (t) (1 -   \tilde{p}_t (1) ) \tilde{p}_t
	(1)  f_t \times e_{\beta} (s) (1 -   \tilde{p}_s (1)) \tilde{p}_s
	(1) f_t^\prime \\
	= &\sum_{t = 1}^{T} \left[ e_{\beta} (t) (1 -   \tilde{p}_t (1) ) \tilde{p}_t
	(1)  f_t \times \sum_{s \neq t} e_{\beta} (s) (1 -   \tilde{p}_s
	(1)) \tilde{p}_s (1) f_s^\prime \right]
	\end{align*}
	By definition of the~$L_2$ projections, we know that
	\[
	\sum_{s=1}^T  (1 -   \tilde{p}_s (1) ) \tilde{p}_s (1)  (\beta (s) - f_s^\prime \beta) f_s = 0_p
	\]
	This implies for each~$t = 1,\ldots, T$, we have
	\[
	\sum_{s \neq t}  (1 -   \tilde{p}_s (1) ) \tilde{p}_s (1)  (\beta (s)
	- f_s^\prime \beta) f_s = - (1 -   \tilde{p}_t (1) ) \tilde{p}_t (1)  (\beta (t) - f_t^\prime \beta) f_t.
	\]
	Plugging this in the cross-term becomes
	\[
	- \sum_{t=1}^T  \left( e_{\beta} (t) (1 -   \tilde{p}_t (1) ) \tilde{p}_t (1) \right)^2  f_t f_t^\prime.
	\]
	The first term can be simplified by
	\[
	\E \left[ w^2_{mt} ( H_{t} ) \, (A_t -   \tilde{p}_t (1) )^4 \right] =
	\tilde{p}_t (1) ( 1 - \tilde{p}_t (1)) \E \left[
	\frac{\tilde{p}_t (1) ( 1 - \tilde{p}_t (1))^3}{p_t (1 \given H_t)}
	+ \frac{( 1 - \tilde{p}_t (1)) \tilde{p}_t (1)^3  }{1 - p_t (1 \given H_t)}
	\right].
	\]
	The above implies that~$W_{m,\beta}$ is equal to
	\[
	\sum_{t=1}^T (1 -   \tilde{p}_t (1) ) \tilde{p}_t (1) \Psi_t \, e_{\beta} (t)^2 f_t f_t^\prime
	\]
	where
	\begin{align*}
	\Psi_t &= \E \left[ \frac{\tilde{p}_t (1) ( 1 - \tilde{p}_t (1))^3
	}{p_t (1 \given H_t)} + \frac{( 1 - \tilde{p}_t (1)) \tilde{p}_t
		(1)^3  }{1 - p_t (1 \given H_t)}  \right] - (1 -   \tilde{p}_t (1)
	) \tilde{p}_t (1)
	\end{align*}
	We assume~$\tilde{p}_t (a)$ is chosen such that it equals
	$\E \left[ p_t (a \given H_t) \right]$
	for~$a \in \{0,1\}$.
	Under this assumption, Jensen's inequality implies that
	$\E \left[ \frac{\tilde{p}_t (a)}{p_t (a \given H_t)} \right] \geq 1$
	for~$a \in \{0,1\}$ with equality holding only if~$p_t (a \given H_t)$
	is constant almost surely (i.e., $p_t (a \given H_t) = \tilde{p}_t (a)$
	a.s.).
	Therefore,
	\[
	\Psi_t \geq ( 1 - \tilde{p}_t (1))^3 +  \tilde{p}_t (1)^3 -
	(1 -   \tilde{p}_t (1) ) \tilde{p}_t (1) \geq 0
	\]
	for~$\tilde{p}_t (1) \in [0,1]$ with equality
	holding only if~$\tilde{p}_t (1) = 1/2$.
	Therefore~$\Psi_t = 0$ \emph{if and only if} $p_t (1 \given H_t) =
	\tilde{p}_t (1) = \frac{1}{2}$ almost surely.

	Combining all of the above we have
	\[
	W_m = \sum_{t=1}^T \tilde{p}_t (1) (1 - \tilde{p}_t (1)) \left[ \Theta_t \sigma^2 +
	\Psi_t e_\beta (t)^2 \right] f_t f_t^\prime
	\]
	where~$\Theta_t = \E \left[ \frac{\tilde{p}_t (1) (1 - \tilde{p}_t (1))}{ p_t (1 \, | \, H_t) (1-p_1 (1 \, | \,  H_t))} \right]$.
	We approximate this by
	\[
	W_m \approx \left[ \bar{\Theta} \sigma^2 +
	\bar{\Psi} \bar{e}_\beta^2 \right] \cdot \sum_{t=1}^T \tilde{p}_t (1) (1 - \tilde{p}_t (1))  f_t f_t^\prime = \left[ \bar{\Theta} \sigma^2 +
	\bar{\Psi} \bar{e}_\beta^2 \right] Q_m
	\]
	where~$\bar{u}$ is the average of~$T^{-1} \sum_{t=1}^T u_t$.
	This implies the non-centrality parameter is
	approximated by
	\[
	\left[ \frac{\sigma^2}{\bar{\Theta} \sigma^2 + \bar{\Psi} \bar{e}_{\beta}^2 }
	\right]
	\left(\frac{\beta_m^\star}{\sigma}\right)^\prime \left( \sum_{t=1}^T
	\tilde{p}_t (1) (1 - \tilde{p}_t (1) ) f_t f_t^\prime \right)
	\left(\frac{\beta_m^\star}{\sigma}\right)
	\]
	as desired.
\end{proof}

The following lemma provides a complementary result
to Lemma~\ref{lemma:noncentrality}.  In particular,
it provides a relation between the non-centrality
parameter and choice of control variables.

\begin{lemma} \normalfont
	\label{lemma:noncentrality2}
	Recall~$e_{\beta} (t) := \beta(t) - f_t^\prime \beta^\star_m$,
	$e_{\alpha} (t; H_t) := \E [ w_{mt} (H_{t}) Y_{t,\Delta} \given H_t ] - g_t (H_t)^\prime \alpha^{\star}_m$,
	and the error term~$\epsilon_t$ is given by
	\[
	\epsilon_t = Y_{t,\Delta} -  \E [ w_{mt} (H_{t}) Y_{t,\Delta} \given H_t ] -
	(A_t - \tilde{p}_t (1)) \beta (t)
	\]
	and satisfies $\E [ \epsilon_t \given H_t, A_t] = 0$.
	We make the following assumptions:
	\begin{itemize}
		\item The marginal treatment effect is correctly specified; that is,
		$\beta(t) = f_t^\prime \beta$
		\item The error term~$\epsilon_t$ satisifes the following second moment condition:
		$\E [ \epsilon^2_t \given H_t, A_t] = \sigma^2$.
	\end{itemize}
	Under these conditions, the non-centrality
	parameter~$\gamma_m$ is approximated by
	\[
	\left[ \frac{\sigma^2}{\bar{\Theta} \sigma^2 + \bar{\Xi}_{\alpha}} \right] \left(\frac{\beta_m^\star}{\sigma}\right)^\prime \left( \sum_{t=1}^T \tilde{p}_t (1) (1 - \tilde{p}_t (1) ) f_t f_t^\prime \right)\left(\frac{\beta_m^\star}{\sigma}\right)
	\]
	where~$ \bar{\Xi}_{\alpha} = \frac{1}{T} \sum_{t=1}^T \E [  \bar{\Xi}_{\alpha} (t)]$, and
	\[
	\bar{\Xi}_{\alpha} (t) = \E \left[ e_{\alpha}^2 (t; H_t) \frac{\tilde{p}_t (1 ) (1 - \tilde{p}_t (1))}{p_t (1 \given H_t) (1 - p_t (1 \given H_t))} \right].
	\]
	$\bar{\Theta}$ is as defined in Lemma~\ref{lemma:noncentrality}.
\end{lemma}

\begin{proof}
	Under the above assumptions, the model error~$\epsilon_{mt}$ decomposes
	into the error~$\epsilon_t$
	and two approximation error terms:
	\begin{align*}
	\epsilon_{mt} &= \epsilon_{t} + (\alpha (t; H_t) - g_t (H_t)^\prime
	\alpha_m ) + (A_t - \tilde{p}_t (1) ) ( \beta(t) -
	\f_t^\prime \beta_m ) \\
	&= \epsilon_{t} +
	(\alpha (t; H_t) - g_t (H_t)^\prime \alpha_m )
	\end{align*}
	The third term is zero by the assumption of properly
	specified treatment effect.
	Plugging the decomposition into the
	formula for~$W_m$, we have
	\begin{align*}
	W_m= \E & \bigg[ \sum_{t=1}^T w_{mt} ( H_{t} )\, \bigg( \epsilon_{t}
	+ e_{\alpha} (t; H_t) \bigg)  (A_t - \tilde{p}_t (1) ) f_t \\
	& \times \sum_{t=1}^T w_{mt} ( H_{t} )\, \bigg( \epsilon_{t}
	+ e_{\alpha} (t; H_t) \bigg)(A_t - \tilde{p}_t (1) )  f_t^\prime\bigg].
	\end{align*}
	We can decompose the above formula into various terms,
	which we now walk through step by step.
	The first terms involve only the
	error~$\epsilon_t$. These were taken care of
	in the prior proof; following the logic in that proof,
	the cross-terms involving both $\epsilon_t$ and
	error $e_{\alpha} (t)$ are zero.

	The next terms involve only the
	approximation error~$e_{\alpha} (t; H_t)$:
	\begin{align*}
	W_{m,\alpha} = \E \bigg[ &\sum_{t=1}^T w^2_{mt} ( H_{t} ) \,
	e^2_{\alpha} (t; H_t) (A_t -   \tilde{p}_t (1) )^2 f_t f_t^\prime \bigg] \\
	+ \E \bigg[ &\sum_{s \neq t} w_{mt} ( H_{t} )\,  e_{\alpha} (t; H_t) (A_t -   \tilde{p}_t (1) ) f_t
	\times w_{ms} ( H_{s} )\,  e_{\alpha} (s; H_s) (A_s -   \tilde{p}_s (1) ) f_s^\prime \bigg]
	\end{align*}
	The cross-product term is zero due to
	centering (i.e.,~$\E [ w_{ms} (H_s) (A_s - \tilde{p}_s (1))|H_t] = 0$).
	The first term can be simplified by
	\[
	\E \left[ e^2_{\alpha} (t; H_t) w^2_{mt} ( H_{t} ) \, (A_t -   \tilde{p}_t (1) )^2 \right] =
	\E \left[ e^2_{\alpha} (t; H_t)
	\frac{\left(\tilde{p}_t (1) ( 1 - \tilde{p}_t (1))\right)^2}
	{p_t (1 \given H_t) (1 - p_t (1 \given H_t))}
	\right].
	\]
	Define
	\[
	\Xi_{\alpha} (t) = \E \left[ e^2_{\alpha} (t; H_t)
	\frac{\left(\tilde{p}_t (1) ( 1 - \tilde{p}_t (1))\right)}
	{p_t (1 \given H_t) (1 - p_t (1 \given H_t))} \right]
	\]
	Then~$W_{m,\alpha}$ is equal to
	\[
	\sum_{t=1}^T (1 -   \tilde{p}_t (1) ) \tilde{p}_t (1) \Xi_{\alpha} (t) f_t f_t^\prime
	\]

	Combining all of the above we have
	\[
	W_m = \sum_{t=1}^T \tilde{p}_t (1) (1 - \tilde{p}_t (1)) \left[ \Theta_t \sigma^2 +
	\Xi_{\alpha} (t) \right] f_t f_t^\prime
	\]
	We approximate this by
	\[
	W_m = \left[ \bar{\Theta} \sigma^2 +
	\bar{\Xi}_{\alpha} \right] \cdot \sum_{t=1}^T \tilde{p}_t (1) (1 - \tilde{p}_t (1))  f_t f_t^\prime = \left[ \bar{\Theta} \sigma^2 +
	\bar{\Xi}_{\alpha} \right] Q_m
	\]
	where~$\bar{u}$ is the average of~$T^{-1} \sum_{t=1}^T u_t$.
	This gives the desired result.
\end{proof}

\section{Sample size calculation for marginal case}
\label{app:marginal_ss}

To test $H_0 : \beta(t) = 0, t = 1\ldots, T$ we use the test statistic
\[
T_{mN}= N \hat{\beta}_m^\prime  \hat{Q}_m \hat{W}_m^{-1} \hat{Q}_m \hat{\beta}_m
\]
where $N$ is the sample size and  $\hat{W}_m$ is given by
\[
\mathbb{P}_n \left[ \sum_{t=1}^T I_t \, w_{mt} ( H_{t+\Delta-1} )\, \hat{\epsilon}_{mt} (A_t - \tilde{p}_t (1) )  f_t
\times\sum_{t=1}^T  I_t \,w_{mt}( H_{t+\Delta-1} )\, \hat{\epsilon}_{mt} (A_t - \tilde{p}_t (1) )  f_t^\prime \right]
\]
with $\hat{\epsilon}_{mt}= Y_{t,\Delta} - g_t(H_t)^\prime \hat{\alpha}_m
-  (A_t - \tilde{p}_t (1) ) f_t^\prime\hat{\beta}_m$, and $\hat{Q}_m$ is given by
\[
\sum_{t=1}^T\mathbb{P}_n  \left[ I_t \, w_{mt} (H_{t+\Delta-1}) (A_t - \tilde{p}_t (1) )^2 f_t \, f_t^\prime \right].
\]
Here we have implicitly assumed that $\hat{W}_m$ is invertible.
The following lemma provides the distribution of $T_{mN}$:

\begin{lemma}[Asymptotic Distribution of  $T_{mN}$]
	\normalfont
	\label{lemma:centering_marginal}
	Under invertibility assumptions,
	$$
	N \left(\hat{\beta}_m-\beta_m^\star\right)^\prime  \hat{Q}_m \hat{W}_m^{-1} \hat{Q}_m \left(\hat{\beta}_m-\beta_m^\star\right)\longrightarrow_d \chi^2_{q_m}.
	$$
	When $N$ is large, consistency of mean and variance estimators
	as well as asymptotic normality imply the distribution of $q_m^{-1} T_{mN}$
	is well-approximated by a noncentral F-distribution distribution,
	$F_{q_m, N-q^\prime-q_m; N\gamma_{m}}$,
	where
	\begin{align}
	\gamma_{m} =  \beta_m^\star {Q_m} W_m^{-1} {Q_m} &\beta_m^\star, \label{noncent_marginal} \\
	W_m = E \bigg[ \sum_{t=1}^T I_t \, w_{mt} ( H_{t+\Delta-1} )\,
	{\epsilon}_{mt} (A_t - \tilde{p}_t (1) )  f_t
	&\times\sum_{t=1}^T  I_t \,w_{ct}( H_{t+\Delta-1} )\, {\epsilon}_{ct}
	(A_t - \tilde{p}_t (1) )  f_t^\prime \bigg], \nonumber \\
	{\epsilon}_{mt} = Y_{t,\Delta} - g_t(H_t)^\prime {\alpha_c^\star}
	-  (A_t - \tilde{p}_t (1) ) &f_t^\prime{\beta_m^\star},
	\text{ and } \nonumber \\
	Q_m = \sum_{t=1}^T E  \bigg[ I_t \, \tilde{p}_t (1) (1 -
	\tilde{p}_t (1) ) ) &f_t \, f_t^\prime \bigg]. \nonumber
	\end{align}
\end{lemma}
We set the rejection region for the test $H_0 : \beta(t; x) = 0, t = 1\ldots, T$:
\begin{equation}
\label{eq:reject_marginal}
\left \{ T_{mN} > \frac{q_m \, ( N - (q^\prime +1) )}{N- (q^\prime+q_m )} F_{q_c, N - (q^\prime+q_c);0}^{-1} \left( 1-\alpha_0 \right) \right \}
\end{equation}
with $\alpha_0$ a specified significance level.    For details regarding
further small sample size adjustments, used when analyzing the data,  see
Appendix~\ref{app:ssa}.

The sample size is the smallest value~$N$ such that
\begin{equation}
\label{eq:ss_marginal}
1-F_{q_m, N- (q^\prime+q_m ); N \gamma_m}
\left(  \frac{N-(q^\prime+1)}{N-(q^\prime+q_m)} F^{-1}_{q_m, N- (q^\prime+q_m ); 0} (1- \alpha_0) \right)
\geq 1- \beta_0.
\end{equation}
Calculation of the sample size $N$ is non-trival due to the unknown
form of the noncentrality parameter, $N\gamma_m$ in (\ref{noncent_marginal}).
We now review the three-step procedure for sample size calculations.

In the first step, equation~\eqref{noncent_marginal} along
with information elicited from the scientist is used to
calculate, via Monte-Carlo integration, $\gamma_m$ in
the non-centrality parameter.
The resulting non-centrality parameter,~$\hat{\gamma}_m$,
is plugged in to Equation~\eqref{eq:ss_marginal}
to solve for an \emph{initial} sample size estimate~$\hat{N}_0$.
In the second step we use a binary search algorithm
to search over a neighborhood of~$\hat{N}_0$.
For each sample size~$N$ required by the binary search
algorithm, $K$ samples each of~$N$ simulated participants
are run.  Within each simulation, the rejection region
for the test is given by equation~\eqref{eq:reject_marginal}
at the specified significance level.
The average number of rejected null hypotheses
across the~$K$ simulations is the estimated
power for the sample size~$N$.  The
sample size is the minimal~$N$ with estimated
power above the pre-specified threshold~$1-\beta_0$.
In the last, third, step we conduct a variety of simulations
to assess the robustness of the sample size calculator
to any assumptions and to make adjustments
to ensure robustness.

\subsection{Application to motivating example}

Table~\ref{app:tab:est_sample_sizemarg} presents the estimated
sample size from our proposed method
to detect a specified alternative for the
conditional proximal effect
given power of~$80\%$ and significance level~$5.0\%$ for the smoking cessation study.
We assume the day of maximal effect is day~$5$ and
the initial effect is $0$ for both levels of the time-varying variable~$X_t$.
The average treatment effect~$\bar{\beta}$ is
set to three plausible values.
\begin{table}[!htb]
	\caption{Estimated sample size, $N$, and
		achieved power.}
	\label{app:tab:est_sample_sizemarg}
	\centering
	\begin{tabular}{l | c c}
		& Sample size & Power \\ \hline
		$\bar{\beta} = 0.030$ & 50 & 80.0 \\
		$\bar{\beta} = 0.025$ & 77 & 80.0 \\
		$\bar{\beta} = 0.020$ & 121 & 80.4 \\ \hline
	\end{tabular}
\end{table}
We perform $1000$ simulations.
Each simulation is based on the Markov chain~$P$,
the sequence of markov chain under action~$P_t^{(1)}$,
and the randomization probability~$p_t (1 \given H_t)$.
These components completely specify the generative model.
Each simulation consists of generating data for~$N$ individuals
and performing the hypothesis test using equation~\eqref{eq:reject}
with the small-sample size adjustment described in Appendix~\ref{app:ssa}.

The third step in forming the simulation-based sample size
calculator is to assess robustness to the assumptions.
We are particularly concerned with the sensitivity
of the calculator to the assumptions on the form of
the transition matrix $P^{(0)}$.
The following is an illustration of how we might
assess robustness to the form of the transition matrix
and, how as a result of the assessment, we make
the calculator more robust to the assumptions.

\subsubsection{Misspecification of transition matrix~$P^{(0)}$}
\label{app:subsubsection:ball}

As in Section~\ref{subsubsection:ball},
we test robustness of the sample size calculator
to misspecification of the transition matrix~$P^{(0)}$
for the Markov chain,~$V_t$, under no treatment;
the treatment effect is still correctly specified.
We suppose the misspecification stems from
noise related to the information elicited from
scientists.
Let~$B_{(\epsilon, \epsilon^\prime)}$ denote
an~$(\epsilon,\epsilon^\prime)$-ball around the
inputs~$(\bar{W}, \bar{Z})$
and~$\Omega_{(\epsilon, \epsilon^\prime)}$ be
the subset of $B_{(\epsilon, \epsilon^\prime)}$ as defined
in Section~\ref{subsubsection:ball}.
Table~\ref{app:tab:ball_1} presents
estimated power under the previously calculated
sample sizes for~$\Omega_{(0.02,4)}$ and~$\Omega_{(0.01,2)}$
respectively.
For both~$(\epsilon,\epsilon^\prime) = (0.01,2)$
and~$(\epsilon,\epsilon^\prime) = (0.02, 4)$,
the estimated power is significantly
below the pre-specified 80\% level
for all three choices of the
average treatment effect~$\bar{\beta}$.

\begin{table}[!htb]
      \caption{Misspecification of transition matrix~$P^{(0)}$:
        minimum estimated power \\
        over set of matrices in~$\Omega_{\epsilon, \epsilon^\prime}$
      }
      \label{app:tab:ball_1}
      \centering
        \begin{tabular}{c |c c}
          & \multicolumn{2}{c}{$(\epsilon, \epsilon^\prime) = $} \\
          & $(0.02,4)$ & $(0.01,2)$ \\ \hline
          $\bar{\beta} = 0.030$ & 43.1 & 66.3 \\
          $\bar{\beta} = 0.025$ & 37.6 & 63.8 \\
          $\bar{\beta} = 0.020$ & 27.3 & 57.6 \\ \hline
        \end{tabular}
\end{table}

\subsubsection{Deviations from a time-inhomogenous transition matrix
under no treatment}
\label{app:subsubsection:weekend}

As in Section~\ref{subsubsection:weekend}, next we test robustness
of the sample size calculator to a different type of misspecification of the transition
matrix~$P^{(0)}$,
that of time-inhomogeneity; as before the treatment effect is still
correctly specified.
We specify~$P^{(0)}_{\text{weekend}}$ via
inputs~$(\bar{W}_{\text{weekend}},\bar{Z}_{\text{weekend}})$
given in Section~\ref{subsubsection:weekend}.
Using the inputs we construct two alternate versions of
what the true transition matrix~$P^{(0)}_{\text{weekend}}$ might be.

To test the calculator, we generate data using the no-treatment
transition matrices $P^{(0)}_{\text{weekend}}$ (for the weekend)
and $P^{(0)}$(for the weekday).  This data is simulated so that
the treatment effect used by the calculator is still correct (e.g. we
select the transition matrices under treatment, $P_{d(t)}^{(1)}$, to
ensure this).

Table~\ref{app:tab:power_weekendeffect}
presents power calculations under these
alternative generative models.
We see that the power is below the
pre-specified 80\% threshold for both
inputs across the three possible values of the
average treatment effect~$\bar{\beta}$.
If the scientist thought such deviations feasible,
then the above analysis suggests for the smoking
cessation example that the sample size be set
to ensure a least~$80\%$ power \emph{over a
set of feasible choices for time-inhomogeneous
choices for the no-treatment transition matrix.}

\begin{table}[!ht]
  \caption{Estimated power under generative
    model with time-inhomogeneous Markov chain.}
  \label{app:tab:power_weekendeffect}
  \centering
  \begin{tabular}{c | r r}
    & \multicolumn{2}{c}{Estimated power} \\
    & Weekend Input 1 & Weekend Input 2 \\ \hline
    $\bar{\beta} = 0.030$ & 82.9 & 75.4 \\
    $\bar{\beta} = 0.025$ & 78.6 & 77.0 \\
    $\bar{\beta} = 0.020$ & 76.4 & 76.9 \\ \hline
  \end{tabular}
\end{table}

\subsection{Deviations from a Markovian generative model}

Here we test robustness of the sample size calculator
to the over-fit semi-Markovian deviations presented in the main text.
To test the calculator, we generate data using the
no-treatment semi-Markov model specified in Appendix~\ref{app:SMC}.
The data is simulated so that the treatment effect used
by the calculator is correct. See Appendix~\ref{app:SMC} for
a discussion of how this was achieved.

Table~\ref{tab:power_SMCeffect_marginal} 
presents achieved power under these alternative 
generative models. 
We see that the achieved power is 
well above the pre-specified 80\%
threshold in each case. 
Therefore the sample size calculator is robust to
such complex deviations from the Markovian generative model.

\begin{table}[!ht]
  \caption{Estimated power under 
    semi-Markov generative.}
  \label{tab:power_SMCeffect_marginal}
  \centering
  \begin{tabular}{c | r }
    & Estimated power \\ \hline
    $\bar{\beta} = 0.030$ & 92.5 \\
    $\bar{\beta} = 0.025$ & 91.2 \\
    $\bar{\beta} = 0.020$ & 88.3 \\ \hline
  \end{tabular}
\end{table}

\subsection{Adjustments to the  simulation-based calculator}

We have now evaluated the simulation calculator.
Here we make adjustments to the simulation calculator to ensure
robustness. First, we note that the simulation calculator is robust to
the potential semi-Markovian deviation. 
We make the decision that we are not
concerned with lack of robustness to
deviations from a time-inhomogenous transition matrix.
Instead we focus on making the simulation calculator
robust to misspecification of transition matrix.

The above analysis suggests for the smoking
cessation example that the sample size should be
set to ensure at least 80\% power \emph{over a set of
feasible choices for the transition matrix~$P^{(0)}$}.
We fix~$(\epsilon,\epsilon^\prime) = (0.01,2)$ to be our
tolerance to misspecification of the inputs.
For each set of inputs~$(W,Z) \in \Omega_{0.01,2}$,
the sample size calculator will compute a sample size, and the
maximum of this set of computed sample sizes will be chosen
to ensure tolerance to misspecification of the transition
matrix. Table~\ref{app:tab:est_sample_size_robust} presents
the sample size under this procedure as well as the
avhieved \emph{minimum power} over the set~$\Omega_{\epsilon,  \epsilon^\prime}$.

\begin{table}[!htb]
    \caption{Estimated sample size, $N$, and
      computed power under~$\epsilon = 2$
      and $\epsilon^\prime = 0.01$.
    }
    \label{app:tab:est_sample_size_robust}
    \centering
    \begin{tabular}{l | c c}
      & Sample size & Minimum Power \\ \hline
      $\bar{\beta} = 0.030$ & 66 & 80.2 \\
      $\bar{\beta} = 0.025$ & 113 & 80.5 \\
      $\bar{\beta} = 0.020$ & 195 & 80.6 \\ \hline
    \end{tabular}
\end{table}

We have illustrated the three-step procedure to forming a
sample size calculator for the smoking cessation study example.
Suppose the scientist specifies an average treatment
effect $\bar{\beta}$ equal to $0.025$.
Based on the above discussion a sample size,~$N$,
of $113$ would be recommended to ensure power
above the pre-specified 80\% threshold across a set of feasible
deviations from the assumed generative model.

\section{Small sample size adjustment for covariance estimation}
\label{app:ssa}

The robust sandwich covariance
estimator~\cite{Mancl2001}
for the entire variance matrix is given
by~$Q^{-1} \Lambda Q^{-1}$.  The
first term,~$Q$, is given by
\[
\left( \sum_{i=1}^N D_i^T W_i D_i \right)
\]
where $D_i$ is the model matrix for
individual~$i$ associated with
equation~\eqref{eq:conditional_ls}, and
$W_i$ is a diagonal matrix of weights either
constructed from $w_{ct} (H_{t+\Delta-1})$
or $w_{mt} (H_{t+\Delta-1})$ for the conditional and marginal
test statistics respectively.
The middle term~$\Lambda$ is given
by
\[
\sum_{i=1}^N D_i^\prime W_i (I_i - H_i)^{-1}
e_i e_i^\prime (I_i - H_i)^{-1} W_i D_i
\]
where $I_i$ is an identity matrix
of correct dimension, $e_i$ is
the individual-specific residual
vector and
\[
H_i = D_i
\left( \sum_{i=1}^N D_i^\prime W_i D_i \right)^{-1}
D_i^\prime W_i
\]
From $Q^{-1} \Lambda Q^{-1}$ we extract
$\hat{\Sigma}_{\beta}$.

\section{Additional details for smoking cessation example sample size calculation}
\label{app:add_details}

Table~\ref{tab:std_effect1} presents the standardized effect
sizes for the two levels of the stratifying variable~$X_t$ 
under the Markov generative model introduced in 
Section~\ref{subsubsection:calculator_smoking}.

\begin{table}[!htb]
  \caption{Standardized effects under the Markovian generative model}
  \label{tab:std_effect1}
  \centering
  \begin{tabular}{l | c c}
    & $X_t = $ ``Non-stress'' & $X_t = $``Stress'' \\ \hline
    $\bar{\beta} = 0.030$ & 0.059 & 0.034 \\
    $\bar{\beta} = 0.025$ & 0.052 & 0.030 \\
    $\bar{\beta} = 0.020$ & 0.038 & 0.020 \\ \hline
  \end{tabular}
\end{table}

Table~\ref{tab:std_effect1} presents the standardized effect
sizes for the two levels of the stratifying variable~$X_t$ 
under the semi-Markov generative model introduced in 
Section~\ref{subsection:SMCdev}.

\begin{table}[!htb]
  \caption{Standardized effects under the semi-Markovian generative model}
  \label{tab:std_effect2}
  \centering
  \begin{tabular}{l | c c}
    & $X_t = $ ``Non-stress'' & $X_t = $``Stress'' \\ \hline
    $\bar{\beta} = 0.030$ & 0.074 & 0.036 \\
    $\bar{\beta} = 0.025$ & 0.063 & 0.028 \\
    $\bar{\beta} = 0.020$ & 0.049 & 0.024 \\ \hline
  \end{tabular}
\end{table}

\subsection{Analytic form of the treatment effect for the smoking cessation example}
\label{app:example_treatment_effect}

For the smoking cessation study, action at decision time~$t$ implies that the individual
is unavailable for treatment for the subsequent hour; therefore~$p_{t+s} (A_{t+s} = 0 \given H_{t+s} ) = 1$
for~$s = 1,\dots, \Delta-1$ given $A_t = 1$.
In this case, we have~$\prod_{j=t+1}^{t+\Delta-1} \frac{1_{A_j=0}}{p_j(A_j|H_j)} = 1$.
Recall the proximal response~$Y_{t,\Delta}$ is equal to $\Delta^{-1} \sum_{u=1}^T 1_{X_{t+u} = 1}$.
Therefore we have
\begin{align*}
	&\E \bigg [ \E \bigg [ \prod_{j=t+1}^{t+\Delta-1}
	\frac{1_{A_j=0}}{p_j(A_j|H_j)}Y_{t,\Delta} \, \bigg| \, A_t = 1 , H_t \bigg] \Given I_t = 1, X_t = x \bigg ] \\
	= &\Delta^{-1} \sum_{s=1}^{\Delta} \E \bigg [ \pr \left(
            X_{t+s} = 1 \given A_t = 1 , H_t \right) \given I_t = 1,
            X_t = x \bigg ]  \\
	= &\Delta^{-1} \sum_{s=1}^{\Delta} \sum_{u \in \{0,1,2\} } \E \bigg [ \pr \left(
            X_{t+s} = 1, U_{t+s} = u \given A_t = 1 , V_t \right) \given I_t = 1,
            X_t = x \bigg ]  \\
	= &\Delta^{-1} \sum_{s=1}^{\Delta} \sum_{u \in \{0,1,2\} }\pr
            \left( X_{t+s} = 1, V_{t+s} = u \given A_t = 1 , X_t  = x,
            U_t = 1 \right)
\end{align*}
where the second equality is due to the Markov property assumption.
Under the specified generative model, the time-varying stratification variable is a Markov process
with transition rule~$P_{d(t)}^{1}$ \emph{for the hour following action}.  Therefore,
\begin{align*}
	\Delta^{-1} \sum_{s=1}^{\Delta} \sum_{u \in \{0,1,2\} } &\pr
        \left( X_{t+s} = 1, U_{t+s} = u \given A_t = 1 , X_t  = x, U_t
        = 1 \right) \\
	= \Delta^{-1} \sum_{s=1}^{\Delta} \sum_{u \in \{0,1,2\}} &\left[ \left( P_{d(t)}^{1} \right)^{s} \right]_{(x,1), (1,u)}.
\end{align*}
If instead~$A_t = 0$, then for~$s > 1$
\begin{align*}
	&\E \bigg [ \prod_{j=1}^{\Delta-1}
	\frac{1_{A_j=0}}{p_j(A_j|H_j)} 1_{X_{t+s} = 1} \, \bigg| \, A_t = 0 , H_t \bigg] \\
	= &\E \bigg [ \prod_{j=1}^{s-1}
	\frac{1_{A_j=0}}{p_j(A_j|H_j)} 1_{X_{t+s} = 1} \E \left[ \prod_{j=s}^{\Delta-1}
	\frac{1_{A_j=0}}{p_j(A_j|H_j)} \Given H_{t+s} \right] \, \bigg| \, A_t = 0 , H_t \bigg] \\
	= &\E \bigg [ \prod_{j=1}^{s-1}
	\frac{1_{A_j=0}}{p_j(A_j|H_j)} 1_{X_{t+s} = 1} \, \bigg| \, A_t = 0 , H_t \bigg].
\end{align*}
Where the third inequality is a consequence of Lemma~\ref{lemma:change_of_reference}.
Taking~$s = 2$, we see that
\begin{align*}
	&\E \bigg [ \frac{1_{A_{t+1}=0}}{p_{t+1} (A_{t+1} | H_{t+1})} 1_{X_{t+2} = 1} \, \bigg| \, A_t = 0 , H_t \bigg] \\
	= &\E \bigg [ \frac{1_{A_{t+1}=0}}{p_{t+1} (A_{t+1} | H_{t+1})} \E \left[ 1_{X_{t+2} = 1} \given H_{t+1}, A_{t+1} \right] \, \bigg| \, A_t = 0 , H_t \bigg] \\
	= &\E \bigg [ \sum_{a \in \{0,1\} } p_{t+1} (A_{t+1} = a \given H_{t+1})
            \frac{1_{A_{t+1}=0}}{p_{t+1} (A_{t+1} | H_{t+1})}
            \pr \left( X_{t+2} = 1 \given H_{t+1}, A_{t+1} = a \right) \, \bigg| \, A_t = 0 , H_t \bigg] \\
	= &\E \bigg [
            \pr \left( X_{t+2} = 1 \given H_{t+1}, A_{t+1} = 0 \right)
            \, \bigg| \, A_t = 0 , H_t \bigg] \\
	= &\E \bigg [ \sum_{u \in \{0,1,2\} } \pr \left( X_{t+2} = 1,
            U_{t+2} = u \given X_{t+1}, U_{t+1}, A_{t+1} = 0 \right)
            \, \bigg| \, A_t = 0 , H_t \bigg] \\
	= &\sum_{x^\prime \in [k]} \sum_{u^\prime \in \{0,1,2\}}
            \sum_{u \in \{0,1,2\} } \pr
            \left( X_{t+2} = 1, U_{t+2} = u \given X_{t+1} = x^\prime,
            U_{t+1} = u^\prime, A_{t+1} = 0 \right) \\
        &\times \pr \left( X_{t+1}
            = x^\prime, U_{t+1} = u^\prime \given H_t , A_{t+1} = 0
          \right)  \\
\end{align*}
\begin{align*}
	= &\sum_{x^\prime  \in [k]}  \sum_{u^\prime \in \{0,1,2\}}
            \sum_{u \in \{0,1,2\} } \pr \left( X_{t+2} = 1, U_{t+2} =
            u \given X_{t+1} = x^\prime, U_{t+1} = u^\prime, A_{t+1} =
            0 \right) \\
        &\times \pr \left( X_{t+1} = x^\prime, U_{t+1} = u^\prime \given
            X_t , U_t = 1, A_{t+1} = 0 \right) \\
	= & \sum_{ u \in \{0,1,2\} }\left[ \left( P^{0} \right)^{2} \right]_{(X_t,1), (1,u)}
\end{align*}
Expanding on this for~$s \geq 2$ the following equality holds
\[
	\E \bigg [ \prod_{s=1}^{\Delta-1}
	\frac{1_{A_j=0}}{p_j(A_j|H_j)} 1_{X_{t+s} = 1} \, \bigg| \, A_t = 0 , H_t \bigg]
	= \sum_{u \in \{0,1,2\} } \left[ \left( P^{0} \right)^{s} \right]_{(X_t,1), (1,u)}
\]
For~$s = 1$ the result holds trivially. The analysis implies that
\[
\E \bigg [ \E \bigg [ \prod_{s=1}^{\Delta-1}
	\frac{1_{A_j=0}}{p_j(A_j|H_j)} Y_{t,\Delta} \, \bigg| \, A_t = 0 , H_t \bigg] \Given X_t = x, I_t =1 \bigg]
	= \Delta^{-1} \sum_{s=1}^{\Delta -1} \sum_{u \in \{0,1,2\} } \left[ \left( P^{0} \right)^{s} \right]_{(x,1), (1,u)}.
\]

\section{Technical details on the semi-Markovian generative model}
\label{app:SMC}

For a semi-Markov process, we let~$Y_n$ denote the~$n$th
state the process enters and~$S_{n}$ denote the time of that
transition.  Each time-homogeneous, semi-Markov process
is characterized by the kernel function:
\[
Q_{ij} (x) \equiv P [ Y_{n+1} = j, S_{n+1} - S_{n} \leq x \given Y_n = i ].
\]
This is the conditional probability of next being in state~$j$
and the transition occurring before time~$x$ given the prior
state is state~$i$.The probability that any transition takes place within the next
$x$ time units is given by summing up the
leaving probabilities for each possible state
$j$,~$Q_{i} (x) = \sum_{j \neq i} Q_{ij} (x)$, not taking into
account transitions from~$i$ to~$i$.

We have a closed form expression for~$Q_{ij} (x)$ for each~$x$ given the
parameters from the above model.
If a transition is allowed
from~$i$ to~$j$ then
\[
Q_{ij} (x; \theta) = \left( 1 - \exp \left[ - \left(
    \frac{x + 0.5}{\lambda_i}\right)^{\kappa_i} \right] \right) \Omega_{ij}
\]
where~$\lambda_{i}$ and~$\kappa_{i}$ are the Weibull distribution
parameters given state~$i$, $\Omega_{ij}$ is the probability
of transitioning from state~$i$ to state~$j$, and~$\theta$ denotes the 
entire set of parameters underlying the semi-Markov model.
Note that from any state~$j$ one can transition
to only $1$ or $2$ other states and therefore the
kernel function is quite low-dimensional.

For a time-homogeneous, semi-Markov process we need to know
the probability of ending up in state $j$ at a time $x$
conditional on starting in state~$i$ at time~$0$.
\begin{align*}
  p_{ij} (x; \theta) & = \delta_{ij} [ 1- Q_{i} (x; \theta) ] \\
                     &+ \sum_{k \neq i}
                       \sum_{v=1}^{x} p_{kj} (x - v; \theta) \left[ F (v; \lambda_i,
                       \kappa_i) - F (v-1; \lambda_i,
                       \kappa_i) \right]
                       \Omega_{ik}
\end{align*}
This requires knowledge of~$p_{kj} (x^\prime; \theta)$ for all~$x^\prime < x$
which is not known a priori.  However, we have
initial conditions~$p_{ij} (0; \theta) = \delta_{ij}$.
Then for the first point in the discretization we have
\begin{align*}
p_{ij} (1) &= \delta_{ij} [ 1 - Q_{i} (1)] +
\sum_{k \neq i} p_{kj} (0) \cdot \left[ F (1; \lambda_i,
               \kappa_i) - F (0; \lambda_i,
               \kappa_i) \right] \Omega_{ik} \\
           &= \delta_{ij} [ 1 - Q_{i} (1)] +
             \left( 1 - \delta_{ij} \right)
             \left[ F (1; \lambda_i,
               \kappa_i) - F (0; \lambda_i,
               \kappa_i) \right] \Omega_{ij}.
\end{align*}
We omitted dependence of~$p_{ij}$ and~$Q_i$ on~$\theta$ for
the sake of space.
This completely determines~$p_{ij} (1)$ for all~$i$ and~$j$.
We can iterate on knowing these parameters to solve for~$p_{ij} (x)$
for each $x = 1,\ldots, \Delta$.
For~$x = 2$, for example, we have:
\begin{align*}
p_{ij} (2) &= \delta_{ij} [ 1 - Q_{i} (2)] +
\sum_{k \neq i} \bigg[ p_{kj} (0) \cdot \left[ F (2; \lambda_i,
               \kappa_i) - F (1; \lambda_i,
               \kappa_i) \right] \Omega_{ik}  \\
           &+ p_{kj} (1) \cdot \left[ F (1; \lambda_i,
             \kappa_i) - F (0; \lambda_i,
             \kappa_i) \right] \Omega_{ik} \bigg].
\end{align*}
Therefore, for a given model specification we can compute the
expected fraction of time classified as ``stressed'' in the next hour.
Let~$\mathcal{A}$ denote the set of states that correspond to
currently being classified as stressed, then
\[
\mu (i; \theta) = \Delta^{-1} \sum_{x=1}^{\Delta} \sum_{j \in \mathcal{A}} p_{ij}
(x; \theta) .
\]
where again~$\theta$ denotes the set of parameters
of the transition and duration models and~$i$ denotes the
current state.
Let~$\theta_0$ denote the parameters for the baseline
generating model.  
Then we define the proximal outcome conditional on being
currently non-stressed as
\[
\mu_{0}(\theta^\prime) = \frac{\sum_{i \in A^{c}} \pi_{i} (\theta_0)
\mu (i; \theta^\prime) }{\sum_{i \in A^c} \pi_{i} (\theta_0) }
\]
and the proximal outcome conditional on being
currently stressed as
\[
\mu_{1}(\theta^\prime) = \frac{\sum_{i \in A} \pi_{i} (\theta_0)
\mu (i; \theta^\prime) }{\sum_{i \in A} \pi_{i} (\theta_0) }.
\]
Note that the proximal outcomes above are defined summing
over the stationary distribution with respect to~$\theta_0$ (i.e.,
$\pi_i (\theta_0)$).  We do this because it is a decent approximation
to the true setting.  Interventions occur infrequently and so 
we expect the stationary distribution over the baseline parameters to
be close to the true distribution of the states given~$X_t = x$

To construct the semi-Markov generating model
under treatment we wish to find
\[
\arg \min_{\theta^\prime} \max_{x \in \{0,1\} } \left \| (\mu_{x;
    \theta^\prime} - \mu_{x;\theta_0}) - \tilde{\beta}_x \right \| 
\]
where~$\tilde{\beta}_x$ is specified alternative treatment effect
for~$X_t = x$.
Thus the problem now turns into an optimization problem.  
Fortunately generic black-box optimization routines in R
were found to be sufficient.

\section{Code to Generate Simulation Results}

The R code used to generate the smoking cessation example results in this paper can be obtained
from \url{https://github.com/wdempsey/stratified_mrt}.

\end{document}